\newif\iffull
\fulltrue

\documentclass[conference]{IEEEtran} %

\usepackage[svgnames,table]{xcolor}
\usepackage{ebproof}
\usepackage{mathpartir}
\usepackage[scaled=0.95]{zlmtt}
\usepackage{setspace}
\usepackage{amsmath}
\usepackage{mathrsfs}
\usepackage{amssymb}
\usepackage[]{todonotes}
\usepackage[inline]{enumitem}
\usepackage{tikz}
\usetikzlibrary{shapes, arrows.meta, calendar, matrix, calc, decorations, overlay-beamer-styles, positioning, fit, decorations.pathreplacing}

\usepackage{listings}
\usepackage{comment}
\usepackage{url}
\usepackage{float}
\usepackage[flushleft]{threeparttable}
\usepackage{array}
\usepackage{caption}
\usepackage{booktabs}
\usepackage{amsthm}
\usepackage{xparse}
\usepackage{thmtools}
\usepackage[normalem]{ulem}

\newcolumntype{C}[1]{>{\centering\arraybackslash}p{#1}}  %
\newcolumntype{M}[1]{>{\centering\arraybackslash}m{#1}}  %
\newcolumntype{L}[1]{>{\raggedright\arraybackslash}m{#1}}

\usepackage[colorlinks=true,urlcolor=black]{hyperref}

\ifCLASSOPTIONcompsoc
  \usepackage[nocompress]{cite}
\else
  \usepackage{cite}
\fi

\newcounter{parcounter}
\newenvironment{labeledparagraph}[1]{%
	\refstepcounter{parcounter}%
	\noindent\label{#1}%
	\begin{center}
		\begin{tabular}{| m{0.85\columnwidth}  m{0.01\columnwidth}}
			\begin{itshape}%
			}
			{%
			\end{itshape} & (\theparcounter) \\ 
		\end{tabular}%
	\end{center}%
}

\newcommand{\nextrule}{\hspace{1em}}

\newcommand{\negativespace}{\hspace{-0.1em}}
\newcommand{\cons}{ ,  ... , }
\newcommand{\concat}{{+}\mspace{-8mu}{+} }

\newcommand{\hash}{\rotatebox{45}{\#}}

\newcommand{\hexnumber}[1]{\ensuremath{0\texttt{x}{#1}}}

\newcommand{\xv}{{\mathit{st}}}

\newcommand{\val}{\mathit{v}}
\newcommand{\lval}{\mathit{lval}}

\newcommand{\bitv}{\overline{b}}

\newcommand{\listindexed}[2]{ #1_\mathrm{1} \cons #1_\mathrm{#2} }

\newcommand{\listindexedcol}[3]{ #1_{\mathrm{1}} \negativespace : \negativespace #2_{\mathrm{1}} ; ... ;\negativespace  #1_{\mathrm{#3}} \negativespace : \negativespace #2_{\mathrm{#3}} }

\newcommand{\tscope}{\gamma}

\newcommand{\stmt}{\mathit{s}}

\newcommand{\stmtif}[3]{\texttt{if} \ \texttt{#1} \ \texttt{then} \ \texttt{#2} \ \texttt{else} \ \texttt{#3}}

\newcommand\restrf[2]{{\left.\kern-\nulldelimiterspace #1 \vphantom{\big|} \right|_{#2}}}

\newcommand{\type}{\tau}

\newcommand{\band}{\ \land \ }

\newcommand{\typestrnew}[3]{\langle\listindexedcol{#1}{#2}{#3}\rangle}
\newcommand{\structval}[3]{\{ #1_1 = #2_1 \cons #1_{#3} = #2_{#3} \}}
\newcommand{\selexp}[3]{\texttt{select} \, e \, \{ #1_1 : #2_1 \cons #1_{#3} : #2_{#3} \} \, #2}

\newcommand{\lbl}{\ell}

\newcommand{\tbvn}[3]{\langle{#1}\rangle^{#2}_{#3}}
\newcommand{\sumindex}[2]{\underset{#1}{\textstyle \overset{#2}\sum}  \ }

\newcommand{\joinOnHigh}[2]{\mathrm{joinOnHigh}(#1, #2)}
\newcommand{\refine}[2]{\mathrm{refine}(#1,#2)}
\newcommand{\getlabel}{\mathrm{lbl}}
\newcommand{\labelof}[2]{\mathrm{lblOf}(#1,#2)}
\newcommand{\raisel}{\mathrm{raise}}

\newcommand{\join}{\mathrm{join}}
\newcommand{\isOut}{\mathrm{isOut}}
\newcommand{\tCall}{\mathrm{tCall}}

\newcommand{\emptyframe}{\bullet}
\newcommand{\lesstype}{\leq}

\newcommand{\consistent}[2]{\hspace{-0.2em} \underset{#1}{\overset{#2}{\sim}} \hspace{-0.2em}}

\newcommand{\tc}{T}
\newcommand{\tcv}{(C,F)}

\newcommand{\mem}{m}
\newcommand{\lequiv}[1]{\hspace{-0.2em} \underset{#1}{\sim} \hspace{-0.2em}}

\newcommand{\stmtctx}{E}
\newcommand{\stmtctxv}{(X,F)}
\newcommand{\stmtctxvn}[1]{(X_{#1},F)}

\newcommand{\expsem}[2]{#1 ( #2 ) }
\newcommand{\sem}[3]{#1 \xrightarrow{#2} #3}

\newcommand{\tightpar}[1]{{\smallskip\noindent\bf #1. }}
\newcommand{\parheading}[1]{{\smallskip\noindent~\bf{#1}}}

\newcommand{\low}[0]{\textcolor{blue}{\text{\emph{L}}}}
\newcommand{\high}[0]{\textcolor{red}{\text{\emph{H}}}}

\newcommand{\toolname}[0]{\textsc{Tap4s}}
\newcommand{\usecases}[0]{$5$}
\newcommand{\tests}[0]{$25$}

\newcommand{\casesItem}[1]{\noindent$\diamond$~\textbf{#1}}

\theoremstyle{plain}

\newtheorem{theorem}{Theorem}
\newtheorem{lemma}{Lemma}
\newtheorem{hypothesis}{Hyp}

\theoremstyle{definition}
\newtheorem{example}{Example}
\newtheorem{definition}{Definition}

\definecolor{P4commentgreen}{HTML}{006400} %
\definecolor{P4stringred}{HTML}{800000} %
\definecolor{P4typeteal}{HTML}{008080} %
\definecolor{P4keyblue}{HTML}{0000ff} %
\definecolor{P4constantpurple}{HTML}{800080} %
\definecolor{P4TableKeywords}{HTML}{0084c6}
\definecolor{P4PacketKeywords}{HTML}{ff557f}
\definecolor{P4Externs}{HTML}{E9692C}
\lstdefinelanguage{P4}{
	basicstyle=\fontsize{7.2}{8.3}\ttfamily\fontseries{l}\selectfont,
	fontadjust=true,
	keepspaces=true,
	tabsize=2,
	breaklines=false,
	breakautoindent=false,
	belowskip = 0pt,
	xleftmargin=13pt, %
	aboveskip = 0pt,
	lineskip=1.3pt,
	captionpos=b,
	numbersep=5pt,
	numbers=left, 
	morekeywords={
		abstract, action, apply, control,
		default, else, extern, exit, false,
		if, package, parser, return, select,
		state, switch, table, this, transition,
		true, type, typedef, value\_set, verify},
	keywordstyle=\bfseries\color{P4keyblue},
	morecomment=[l]{\/\/},
	morecomment=[s]{\/\*\*}{\*\/},
	morecomment=[s]{\/\*}{\*\/},
	commentstyle=\color{P4commentgreen},
	stringstyle=\color{P4stringred},
	morestring=[b]",
	classoffset=1,
	morekeywords={bit, bool, const, enum, error,
		header, header\_union, in, inout, int,
		list, match\_kind, out, string, struct, tuple,
		varbit, void},
	keywordstyle=\bfseries\color{P4typeteal},
	classoffset=2,
	morekeywords={actions, key, default\_action},
	keywordstyle=,
	classoffset=3,
	morekeywords={extract, emit, packet\_in, packet\_out},
	keywordstyle=,
	classoffset=4,
	morekeywords={mark_to_drop},
	keywordstyle=,
	classoffset=0,
	classoffset=2,
	moredelim=[is][\color{P4constantpurple}]{??}{??},
	classoffset=0,
}[keywords,comments]

\definecolor{keyword}{HTML}{0084c6}
\definecolor{darkblue}{RGB}{0,0,200}
\definecolor{darkred}{RGB}{200,0,0}
\lstdefinelanguage{policy}{
	basicstyle=\fontsize{8}{8}\selectfont\ttfamily,
	fontadjust=true,
	keepspaces=true,
	tabsize=4,
	breaklines=false,
	breakautoindent=false,
	belowskip = 5pt,
	lineskip=1.3pt,
	captionpos=b,
	numbers=none,
	literate={->}{{$\rightarrow$}}1,
	morekeywords={packet\_in, packet\_out, standard\_metadata},
	keywordstyle=\bfseries\color{keyword},
	classoffset=1,
	morekeywords={low},
	keywordstyle=\color{darkblue},
	classoffset=2,
	morekeywords={high},
	keywordstyle=\color{darkred},
	classoffset=0,
}[keywords,comments]

\definecolor{actionKeyword}{HTML}{ff557f}
\lstdefinelanguage{contract}{           
	basicstyle=\fontsize{8}{8}\selectfont\ttfamily,
	fontadjust=true,
	keepspaces=true,
	tabsize=2,
	breaklines=false,
	breakautoindent=false,
	belowskip = 5pt,
	lineskip=1.3pt,
	captionpos=b,
	numbers=none,
	literate={->}{{$\rightarrow$}}1,
	morekeywords={table, extern},
	keywordstyle=\bfseries\color{keyword},
	classoffset=1,
	morekeywords={low},
	keywordstyle=\color{darkblue},
	classoffset=2,
	morekeywords={high},
	keywordstyle=\color{darkred},
	classoffset=3,
	keywordstyle=\color{actionKeyword},
	escapeinside={(*@}{@*)}, %
}[keywords,comments]

\newcommand{\actiontt}[1]{\textcolor{darkblue}{\texttt{#1}}}

\ExplSyntaxOn

\NewDocumentCommand{\bitvector}{m}
{
	\tl_set:Nn \l_tmpa_tl { #1 }
	\scriptsize
	\setlength{\extrarowheight}{2pt}
	\setlength{\arraycolsep}{2pt}
	\begin{array}{ | *{ \tl_count:N \l_tmpa_tl } { c | } }
		\hline
		\seq_set_split:Nnn \l_tmpa_seq { } { #1 }
		\seq_use:Nn \l_tmpa_seq { & } \\
		\hline
	\end{array}
}

\ExplSyntaxOff

\begin{document}
\sloppy

\title{Securing P4 Programs by Information Flow Control}

\makeatletter
\newcommand{\linebreakand}{%
  \end{@IEEEauthorhalign}
  \hfill\mbox{}\par
  \mbox{}\hfill\begin{@IEEEauthorhalign}
}
\makeatother

\author{\IEEEauthorblockN{Anoud Alshnakat\IEEEauthorrefmark{1}, Amir M. Ahmadian\IEEEauthorrefmark{1}, Musard Balliu\IEEEauthorrefmark{1}, Roberto Guanciale\IEEEauthorrefmark{1} and Mads Dam\IEEEauthorrefmark{1}}
\IEEEauthorrefmark{1}\textit{KTH Royal Institute of Technology}
}

\maketitle

\begin{abstract}
 	Software-Defined Networking (SDN) has transformed network architectures by decoupling the control and data-planes, enabling fine-grained control over packet processing and forwarding. P4, a language designed for programming data-plane devices, allows developers to define custom packet processing behaviors directly on programmable network devices. This provides greater control over packet forwarding, inspection, and modification. However, the increased flexibility provided by P4 also brings significant security challenges, particularly in managing sensitive data and preventing information leakage within the data-plane.
 	
 	This paper presents a novel security type system for analyzing information flow in P4 programs that combines security types with interval analysis. 
 	The proposed type system allows the specification of security policies in terms of input and output packet bit fields rather than program variables.
 	We formalize this type system and prove it sound, guaranteeing that well-typed programs satisfy noninterference. 
 	Our prototype implementation, {\toolname}, is evaluated on several use cases, demonstrating its effectiveness in detecting security violations and information leakages.
\end{abstract}

\section{Introduction}\label{sec:intro}

Software-Defined Networking (SDN) \cite{goransson2016software} is a software-driven approach to networking that enables programmatic control of network configuration and packet processing rules. SDN achieves this by decoupling the routing process, performed in the control-plane, 
from the forwarding process performed in the data-plane. 
The control-plane is often implemented by a logically-centralized SDN controller that 
is responsible for network configuration and the setting of forwarding rules. 
The data-plane consists of network devices, such as programmable switches, that process and forward packets based on instructions received from the control-plane. 
Before SDN, hardware providers had complete control over the supported functionalities of the devices, leading to lengthy development cycles and delays in deploying new features. 
SDN has shifted this paradigm, allowing application developers and network engineers to implement specific network behaviors, such as deep packet inspection, load balancing, and VPNs, and execute them directly on networking devices.

Network Functions Virtualization (NFV) further expands upon this concept, enabling the deployment of multiple virtual data-planes over a single physical infrastructure \cite{matias2015toward}. 
SDN and NFV together offer increased agility and optimization, making them cornerstones of future network architectures. 
Complementing this evolution, the Programming Protocol-independent Packet Processors (P4) \cite{bosshart2014p4} domain-specific language has emerged as a leading standard for programming the data-plane's programmable devices, such as FPGAs and switches.
Additionally, P4 serves as a specification language to define the behavior of the switches as it provides a suitable level of abstraction, yet is detailed enough to accurately capture the behavior of the switch. It maintains a level of simplicity and formalism that allows for effective automated analysis~\cite{albab2022switchv}.

NFVs and SDNs introduce new security challenges that extend beyond the famous and costly outages caused by network misconfigurations\cite{chica2020security}.
Many data-plane applications process sensitive data, such as cryptographic keys and internal network topologies. %
The complexity of these applications, the separation of ownership of platform and data-plane in virtualized environments, and the
integration of third-party code facilitate undetected information leakages.
Misconfiguration may deliver unencrypted packets to a public network, bugs may leak sensitive packet metadata or routing configurations that expose internal network topology, and malicious code may build covert channels to exfiltrate data via
legitimate packet fields such as TCP sequence numbers and TTL fields \cite{zander2006covert}.

In this domain, the core challenge lies in the data dependency of what is observable, what is secret, and the packet forwarding behavior.
An attacker may be able to access only packets belonging to a specific subnetwork, only packets for a specific network protocol may be secret, and switches may drop packets based on the matching of their fields with routing configurations.
These data dependencies make information leakage a complex problem to address in SDN-driven networks.

Existing work in the area of SDN has focused on security of routing configurations by analyzing network flows that are characterized by port numbers and endpoints. %
However, these works ignore indirect flows that may leak information via other packet fields.
In the programming languages area, current approaches (including P4BID \cite{grewal2022p4bid}) substantially
ignore data dependencies and lead to overapproximations unsuitable for SDN applications. For example, the sensitivity of a field in a packet might depend on the packet's destination.

We develop a new approach to analyze information flow in P4 programs.
A key idea is to augment a security type system
(which is a language-based approach to check how information can flow in a program)
with interval analysis, which in the domain of SDNs can be used to
abstract over the network's parameters such as subnetwork segments, port ranges, and non-expired TTLs. Therefore, in our approach, in addition to a security label, the security type also keeps track of an interval.

The analysis begins with an input policy, expressed as an assignment of types to fields of the input packet. For instance, a packet might be considered sensitive only if its source IP belongs to the internal network.
The analysis conservatively propagates labels and intervals throughout the P4 program in a manner reminiscent of dynamic information flow control~\cite{le2006automata} and symbolic execution, cf. \cite{balliu2012encover}. This process is not dependent on a prior assignment of security labels to internal program variables, thus eliminating the need for the network engineer to engage with P4 program internals. The proposed analysis produces multiple final output packet typings, corresponding to different execution paths. %
These types are statically compared with the output security policy, which allows to relate observability of the output to intervals of fields of the resulting packets and their metadata.

The integration of security types and intervals is challenging.
On one hand, the analysis should be path-sensitive and be driven by values in the packet fields to avoid rejecting secure programs due to overapproximation. On the other hand the analysis must be sound and not miss indirect information flows.
Another challenge is that the behaviors of P4 programs depend on tables and external functions, but these components are not defined in P4.
We address this by using user-defined contracts that overapproximate their behavior. %

\tightpar{Summary of contributions}
\begin{itemize}
	\item We propose a security type system which combines security labels and abstract domains to provide noninterference guarantees on P4 programs.
	\item Our approach allows defining data-dependent policies without the burden of annotating P4 programs.
	\item We implement the proposed type system in a prototype tool {\toolname}\cite{tap4sTool} and evaluate the tool on a test suite and {\usecases} use cases.
\end{itemize}

\section{P4 Language and Security Challenges}\label{sec:background}

This section provides a brief introduction to the P4 language and its key features, while motivating the need for novel security analysis that strikes a balance between expressiveness of security policies and automation of the verification process.

P4 manipulates and forwards packets via a pipeline %
consisting of three stages: parser, match-action, and deparser. The parser stage dissects incoming packets, converting the byte stream into
structured header formats. In the match-action stage, these headers are matched against rules to determine the appropriate actions, such as modifying, dropping, or forwarding the packet to specific ports. Finally, the deparser stage reconstructs the processed packet back into a byte stream, ready for transmission over the network.

We use Program~\ref{prg:basic_congestion} as a running example throughout the paper.
The program implements a switch that manages congestion in the network of Fig. \ref{fig:basic_congestion}.
In an IPv4 packet, the Explicit Congestion Notification (ECN) field provides the status of congestion experienced by switches while transmitting the packet along the path from source to destination.
ECN value $0$ indicates that at least one of the traversed switches does not support the ECN capability, values $1$ and $2$ indicate that all traversed switches support ECN and the packet can be marked if congestion occurs, and $3$ indicates that the packet has experienced congestion in at least one of the switches.

\begin{figure}[ht]
	\centering
	{%
		\input{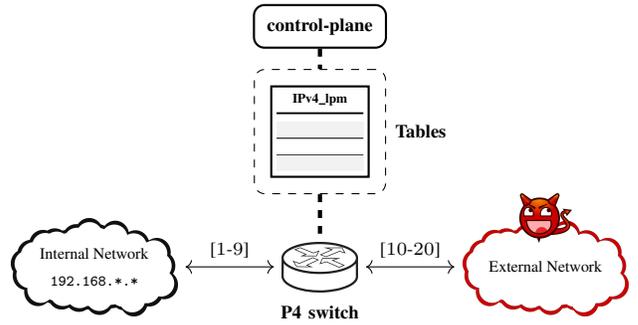}%
	}%
	\caption{\emph{Congestion notifier} network layout}
	\label{fig:basic_congestion}
\end{figure}

We assume for this example that the switch is the \emph{only} ingress and egress point for traffic entering and exiting the internal network, connecting it to external networks as shown in Fig. \ref{fig:basic_congestion}.
To illustrate our approach, we assume that the switch is designed to prevent any information leakage about internal network congestion to the external network.
In addition to standard packet forwarding, the switch sets ECN to $3$ if its queue length exceeds a predefined threshold. This holds only if the packet’s destination is within the internal network. Conversely, if the packet is destined to the external network, the switch sets the ECN field to 0. This indicates that ECN is not supported for outbound traffic and ensures that congestion signals experienced within the internal network are not exposed externally.\footnote{In scenarios with multiple ingress and egress points, where external traffic may fully traverse the internal network, the identification of outbound packets cannot rely solely on IP addresses. Instead, classification would need to be based on the forwarding port to allow the use of the ECN field while the packet is inside the internal network.}

\tightpar{P4 structs and headers}
Structs are records used to define the format of P4 packets.
Headers are special structs with an additional implicit boolean indicating the header's validity, which is set when the header is extracted.
Special function \texttt{isValid} (line 53) is used to check the validity of a header.

For example, the struct \texttt{headers} on line 3 has two headers of type \texttt{ethernet\_t} and \texttt{ipv4\_t},
as depicted in Fig. \ref{fig:packet_header}.
The fields of the \texttt{ethernet\_t} specify the source and destination MAC addresses and the Ethernet type.
The header \texttt{ipv4\_t} represents a standard IPv4 header with fields
such as ECN, time-to-live (TTL), and source and destination IP addresses.

\tightpar{Parser}
The parser dissects incoming raw packets (\texttt{packet} on line 12), extracts the raw bits, and groups them into headers.
The parser's execution begins with the \texttt{start} state and terminates either
in \texttt{reject} state or
\texttt{accept} state accepting the packet and moving to the next stage of the pipeline.

For example, \texttt{MyParser} consists of three states.
The parsing begins at the \texttt{start} state (line 17) and
transitions to \texttt{parse\_ethernet} extracting the Ethernet header from the input packet (line 22), which automatically sets the header's validity boolean to true.
Next, depending on the value of \texttt{hdr.eth.etherType}, which indicates the packet's protocol,
the parser transitions to either state \texttt{parse\_ipv4} or state \texttt{accept}.
If the value is \hexnumber{0800}, indicating an IPv4 packet, the parser transitions to state \texttt{parse\_ipv4} and extracts the IPv4 header (line 30).
Finally, it transitions to the state \texttt{accept} (line 31), accepts the packet, and moves to the match-action stage.

\begin{figure*}[ht]
	\centering
		\begin{tikzpicture}[node distance=0pt,outer sep=0pt]
		\node[draw, minimum height=20pt] (dstAddr) {\scriptsize\texttt{dstAddr}};
		\node[draw, minimum height=20pt, right=of dstAddr] (srcAddr) {\scriptsize\texttt{srcAddr}};
		\node[draw, minimum height=20pt, right=of srcAddr] (etherType) {\scriptsize\texttt{etherType}};
		\draw [decorate, decoration={brace,amplitude=10pt,mirror,raise=3pt},xshift=0pt,yshift=0pt]
		(dstAddr.south west) -- (etherType.south east) node [black, midway, yshift=-0.8cm] {\footnotesize{\texttt{ethernet\_t} fields}};

		\node[draw, minimum height=20pt, right=of etherType] (omitted) {\scriptsize $\qquad \hdots \qquad$};
		\node[draw, minimum height=20pt, right=of omitted] (ecn) {\scriptsize\texttt{ecn}};
		\node[draw, minimum height=20pt, right=of ecn] (omitted2) {\scriptsize $\qquad \hdots \qquad$};
		\node[draw, minimum height=20pt, right=of omitted2] (ttl) {\scriptsize\texttt{ttl}};
		\node[draw, minimum height=20pt, right=of ttl] (omitted3) {\scriptsize $\qquad \hdots \qquad$};
		\node[draw, minimum height=20pt, right=of omitted3] (srcAddrIP) {\scriptsize\texttt{srcAddr}};
		\node[draw, minimum height=20pt, right=of srcAddrIP] (dstAddrIP) {\scriptsize\texttt{dstAddr}};
		\draw [decorate, decoration={brace,amplitude=10pt,mirror,raise=3pt},xshift=0pt,yshift=0pt] (omitted.south west) -- (dstAddrIP.south east) node [black, midway, yshift=-0.8cm] {\footnotesize{\texttt{ipv4\_t} fields}};
	\end{tikzpicture}
	\caption{Packet header}
	\label{fig:packet_header}
	\hfil
	\newline
	\minipage[b]{0.5\textwidth}
\begin{lstlisting}[
		language=P4,
		caption={},
		firstnumber=last]
const bit<19> THRESHOLD = 10;

struct headers {
	ethernet_t  eth;
	ipv4_t      ipv4;
}

void decrease (inout bit<8> x) { 
	x = x - 1; 
}

parser MyParser(packet_in packet, out headers hdr,
                inout metadata meta, 
                inout standard_metadata_t 
                      standard_metadata) {
  
	state start { 
		transition parse_ethernet; 
	}

	state parse_ethernet {
		packet.extract(hdr.eth);
		transition select(hdr.eth.etherType) {
			0x0800: parse_ipv4;
			default: accept;
		}
	}  

	state parse_ipv4 {
		packet.extract(hdr.ipv4);
		transition accept;
	}
}
\end{lstlisting}
	\endminipage\hfill
	\minipage[b]{0.5\textwidth}
\begin{lstlisting}[
		language=P4,
		caption={},
		firstnumber=last]
control MyCtrl(inout headers hdr,
               inout metadata meta,
               inout standard_metadata_t standard_metadata) {
	action drop() {
		mark_to_drop(standard_metadata);
	}
	action ipv4_forward(bit<48> dstAddr, bit<9> port) {
		standard_metadata.egress_spec = port;
		hdr.eth.srcAddr = hdr.eth.dstAddr;
		hdr.eth.dstAddr = dstAddr;
		decrease(hdr.ipv4.ttl);
	}
	table ipv4_lpm {
		key = { hdr.ipv4.dstAddr: lpm; }
		actions = { ipv4_forward; drop; }
		default_action = drop();
	}

	apply {
		if (hdr.ipv4.isValid()) {
			if (hdr.ipv4.dstAddr[31:24] == 192 &&
				  hdr.ipv4.dstAddr[23:16] == 168){
				if (standard_metadata.enq_qdepth >= THRESHOLD)
					hdr.ipv4.ecn = 3;
			} else {
				hdr.ipv4.ecn = 0;
			}
			ipv4_lpm.apply(); //forward all valid packets
		} else {
			drop();
		}
	}	
}
\end{lstlisting}
	\endminipage\hfill
	\captionof{lstlisting}{Congestion notifier}
	\label{prg:basic_congestion}
	\vspace{-8pt}
\end{figure*}

\tightpar{Match-Action}
This stage processes packets as instructed by control-plane-configured tables.
A table consists of key-action rows and each row determines the action to be performed based on the key value.
Key-action rows are updated by the control-plane, externally to P4.
By applying a table, the P4 program matches the key value against table entries and executes the corresponding action.
An action is a programmable function performing operations on a packet, such as forwarding, modifying headers, or dropping the packet.

The match-action block \texttt{MyCtrl} of Program \ref{prg:basic_congestion} starts at line 34.
If the IPv4 header is not valid (line 53) the packet is dropped.
Otherwise, if the packet's destination (line 54-56) is the internal network, the program checks for congestion.
The standard metadata's \texttt{enq\_qdepth} field indicates the length of the queue that stores packets waiting to be processed.
A predefined \texttt{THRESHOLD} is used to determine the congestion status and store it in the \texttt{ecn} field (line 57 and 59).
Finally, the packet is forwarded by applying the \texttt{ipv4\_lpm} table (line 61).
This table, defined at line 46,
matches based on longest prefix (\emph{lpm}) of the IPv4 destination address (\texttt{hdr.ipv4.dstAddr}),
and has two actions (shown on line 48): \texttt{ipv4\_forward} which forwards the packet and \texttt{drop} which drops the packet. %
If no match exists, the default action on line 49 is invoked.

\tightpar{Calling conventions}
P4 is a heapless language, implementing a unique copy-in/copy-out calling convention that allows static allocation of resources.
P4 function parameters are optionally annotated with a direction (\texttt{in}, \texttt{inout} or \texttt{out}).
The direction indicates how arguments are handled during function invocation and termination,
offering fine-grained control over data visibility and potential side effects.

For example, \texttt{inout} indicates that the invoked function can both read from and write to a local copy of the caller's argument. %
Once the function terminates, the caller receives the updated value of that argument.
For instance, assume \texttt{hdr.ipv4.ttl} value is $10$ in line 43.
The invocation of \texttt{decrease} copies-in the value $10$ to parameter \texttt{x}, and the assignment on line 9 modifies \texttt{x} to value $9$.
Upon termination, the function copies-out the value $9$ back to the caller's parameter,
changing the value of \texttt{hdr.ipv4.ttl} to $9$ in line 44.

\tightpar{Externs}
Externs are functionalities that are implemented outside the P4 program and their behavior is defined by the underlying hardware or software platform.
Externs are typically used for operations that are either too complex or not directly expressible in P4's standard constructs. This includes operations like hashing, checksum computations, and cryptographic functions.
Externs can directly affect the global architectural state that is external to the P4 state,
but their effects to the P4 state are controlled by the copy-in/copy-out calling convention.

For example, the extern function \texttt{mark\_to\_drop} (line 38) signals to the forwarding pipeline that a packet should be discarded.
Generally, the packet is sent to the port identified by the standard metadata's \texttt{egress\_spec} field,
and dropping a packet is achieved by setting this field to the drop port of the switch.
The drop port's value depends on the target switch; we assume the value is $0$.

\subsection{Problem statement}\label{sec:problem_statement}

The power and flexibility of P4 to programmatically process and forward packets across different networks provides opportunities for security vulnerabilities such as information leakage and covert channels.
For instance, in Program \ref{prg:basic_congestion}, additionally to the ECN, the standard metadata's \texttt{enq\_qdepth} field, which indicates the length of the queue that stores packets waiting to be processed, indirectly reveals the congestion status of the current switch.

Programming errors and misconfigurations can cause information leakage. Consider the application of the \texttt{ipv4\_lpm} table (line 61) which forwards the packet to a table-specified port.
A bug in the branch condition on line 54, which checks the least significant bits of the \texttt{dstAddr} (e.g. by mistakenly checking \texttt{hdr.ipv4.dstAddr[7:0] == 192} instead),  would result in setting the \texttt{ecn} field on the packets leaving the internal network, thus causing  the packets forwarded to an external network to leak information about the internal network's congestion state.
Covert channels can also result from buggy or malicious programs. For example, by encoding the \texttt{ecn} field into the \texttt{ttl} field, an adversary can simply inspect \texttt{ttl} to deduce the congestion status.

To detect these vulnerabilities,
we set out to study the security of P4 programs by means of information flow control (IFC).
IFC tracks the flow of information within a program, preventing leakage from sensitive sources to public sinks.
Information flow security policies are typically expressed by assigning security labels to the sources and sinks and the flow relations between security labels describe the allowed (and disallowed) information flows.  In our setting, the sensitivity of sources (sinks) depends on predicates on the input (output) packets and standard metadata. Therefore, we specify the security labels of sources (i.e. input packet and switch state) by an \emph{input policy}, while the security labels of the sinks (i.e. output packet and switch state) are specified by an \emph{output policy}.

The input policy of Program \ref{prg:basic_congestion}, describing the security label of its sources is defined as:
\begin{labeledparagraph}{par:policy}
	If the switch's input packet has the protocol IPv4 (i.e. \textnormal{\texttt{hdr.eth.etherType}} is \textnormal{\hexnumber{0800}}) and its IPv4 source address \textnormal{\texttt{hdr.ipv4.srcAddr}} belongs to the internal network subnet \textnormal{\texttt{192.168.*.*}}, then the \textnormal{\texttt{ecn}} field is secret, otherwise it is public.
	All the other fields of the input packet are always public,
	while the switch's \textnormal{\texttt{enq\_qdepth}} is always secret.
\end{labeledparagraph}

Program \ref{prg:basic_congestion} should not leak sensitive information 
to external networks. An output policy defines public sinks by the ports associated with the external network and labels the fields of the corresponding packets as public. 

\begin{labeledparagraph}{par:policy_out}
	Packets leaving the switch through ports $10$ to $20$ are forwarded to the external network and are observable by attackers. 
	Therefore, all fields of such packets should be public. 
	All the other packets are not observable by attackers.
\end{labeledparagraph}

Our goal is to design a static security analysis that strikes a balance between expressiveness and automation of the verification process.
We identify three main challenges that a security analysis of P4 programs should address:
\begin{enumerate}
	\item Security policies are data-dependent.
	      For instance, the \texttt{ecn} field is sensitive only if the packet is IPv4 and its IP source address is in the range \texttt{192.168.*.*}.
	      
	\item The analysis should be value- and path-sensitive, reflecting the different values of header fields. 
		  For example, the value of the field \texttt{etherType} determines the packet's protocol and its shape.
		  This information influences the reachability of program paths; for instance if the packet is IPv4 the program will not go through the parser states dedicated to processing IPv6 packets.
	      
	\item Externs and tables behavior are not defined in P4.
Tables are statically-unknown components and configured at runtime.
	      For example, a misconfiguration of the \texttt{ipv4\_lpm} table may insecurely forward packets with sensitive fields to an external network.
\end{enumerate}

Note that P4 lacks many features that could negatively affect analysis precision, including heap, memory aliasing, recursion, and loops.

\tightpar{Threat model}
Our threat model considers a network attacker that knows the code of the P4 program and observes data on public sinks, as specified by a policy.
We also assume that the keys and the actions of the tables are public and observable,
but tables can pass secret data as the arguments of the actions. 
Because of the batch-job execution model,
security policies can be specified as data-dependent security types over the initial and final program states. 
We aim at protecting against storage channels pertaining to explicit and implicit flows,
while deferring other side channels, e.g. timing, to future work.

\section{Solution Overview}\label{sec:overview}

We develop a novel combination of security type systems and interval abstractions to check information flow policies.
We argue that our lightweight analysis of P4 programs provides a sweet spot balancing expressiveness, precision and automation. 

Data-dependent policies are expressed by security types augmented with intervals, 
and the typing rules ensure that the program has no information flows from secret (\high) sources to public (\low) sinks. 
Specifically, a security type is a pair $( I , \ell )$ of an interval $I$ indicating a range of possible values
and a security label $\ell \in \{{\low}, {\high}\}$.
For simplicity, we use the standard two-element security lattice
$\{{\low}, {\high}\}$ ordered by $\sqsubseteq$ with lub $\sqcup$.
For example, the type $(\langle 1,5 \rangle , \low )$ of the \texttt{ttl} field of the \texttt{ipv4} header specifies that the \texttt{ttl} field contains public data ranging between $1$ and $5$.

The security types allow to precisely express data-dependent policies such as (\ref{par:policy}). 
The input and output policies in our approach specify the shape of the input and output packets. 
Since packets can have many different shapes (e.g. IPv4 or IPv6), these policies may result in multiple distinct policy cases.
For example, input policy (\ref{par:policy}) results in two cases:

In the \emph{first input policy case}, the packet's \texttt{hdr.eth.etherType} is $\hexnumber{0800}$, its IPv4 source address is in the internal network of interval $\langle \texttt{192.168.0.0}, \texttt{192.168.255.255} \rangle$, \texttt{hdr.ipv4.ecn} and standard metadata's \texttt{enq\_qdepth} can contain any value (represented as $\langle*\rangle$) but are classified as {\high}, while all other header fields are $(\langle * \rangle , \low)$ (omitted here). We express this policy using our security types as follows:
\begin{flushleft}
	\footnotesize
	\vspace*{-6mm}
	\begin{align*}
		\texttt{hdr.eth.etherType} &: (\langle \hexnumber{0800}, \hexnumber{0800} \rangle , \low) \\
		\texttt{hdr.ipv4.srcAddr} &: (\langle \texttt{192.168.0.0}, \ \texttt{192.168.255.255} \rangle , \low) \\
		\texttt{hdr.ipv4.ecn} &: (\langle * \rangle , \high) \\
		\texttt{standard\_metadata.enq\_qdepth} &: (\langle * \rangle , \high)
	\end{align*}
\end{flushleft}

The intervals and labels in these security types describe the values and labels of the initial state of the program under this specific input policy case.

The \emph{second input policy case} describes all the packets where \texttt{hdr.eth.etherType} is not $\hexnumber{0800}$
or IPv4 source address is outside the range $\langle\texttt{192.168.0.0}, \ \texttt{192.168.255.255}\rangle$, 
all of the packet header fields are $(\langle *\rangle , \low)$, while the standard metadata's \texttt{enq\_qdepth} is still $(\langle * \rangle , \high)$.

Similarly, the output policy (\ref{par:policy_out}) can be expressed with the output policy case:
``if the standard metadata's \texttt{egress\_spec} is $(\langle 10,20 \rangle , \low)$, then all of the packet's header fields are $(\langle * \rangle , \low)$.''

It turns out that this specific output policy case is the only interesting one, even though output policy (\ref{par:policy_out}) can result in two distinct policy cases. 
In the alternative case, the fact that the attacker is unable to observe the output packet can be represented by assigning $(\langle * \rangle , \high)$ to all the fields of the packet.  
The flow relation among security labels, as determined by the ordering of the security labels, only characterizes flows from {\high} sources to {\low} sinks as insecure. 
This implies that any policy cases where the source is {\low} or the sink is {\high} cannot result in insecure flows.  
Thus, the alternative case is irrelevant and can be safely ignored.

Driven by the data-dependent types, we develop a new security type system that uses the intervals to provide a finer-grained assignment of security labels. 
Our interval analysis allows the type system to statically eliminate execution paths that are irrelevant to the security policy under consideration, thus addressing the second challenge of precise analysis.
For example, our interval analysis can distinguish between states where \texttt{hdr.eth.etherType} is \hexnumber{0800} 
and  states where it is not, essentially providing a path-sensitive analysis.
This enables the analysis to avoid paths where \texttt{hdr.eth.etherType} is \emph{not} \hexnumber{0800} when checking the policy of IPv4 packets. 
As a result, we exclude paths visited by non-IPv4 packets when applying the \texttt{ipv4\_lpm} table in line 61. 
This reduces the complexity of the analysis as we avoid exploring irrelevant program paths, and helps reduce false positives in the results.

Finally, to address the challenge of tables and externs, 
we rely on user-defined contracts which capture a bounded model of the component's behavior. 
Upon analyzing these components, the type system uses the contracts to drive the analysis.
For Program~\ref{prg:basic_congestion}, 
the contract for a correctly-configured table \texttt{ipv4\_lpm} ensures that if the packet's \texttt{hdr.ipv4.dstAddr} belongs to the internal network, then the action \texttt{ipv4\_forward} (line 40) forwards the packet to ports and MAC addresses connected to the internal network.

Even if the \texttt{ipv4\_lpm} table is correctly configured and its contract reflects that, bugs in the program can still cause unintended information leakage. For example, on line 54, the branch condition might have been incorrect and instead of checking the 8 most significant bits (i.e. [31:24]) of the \texttt{hdr.ipv4.dstAddr}, it checks the least significant bits (i.e. [7:0]). This bug causes the \texttt{hdr.ipv4.ecn} field in some packets destined for the external network to include congestion information, leading to unintended information leaks. 
Such errors are often overlooked but can be detected by our type system.

We ensure that the program does not leak sensitive information by checking  the final types produced by the type system  against an output policy. If these checks succeed, the program is deemed secure. 
The details of this process and the role of the interval information in the verification process are explained in Section~\ref{sec:type_system}.

\section{Semantics}\label{sec:semantics}
In this section, we briefly summarize a big-step semantics of P4.  
The language's program statements, denoted by $\stmt$, include standard constructs such as assignments, conditionals, and sequential composition. Additionally, P4 supports transition statements, function calls, table invocations, and extern invocations as shown in Fig.~\ref{fig:syntax_grammar}.

\begin{figure}[ht]
	\centering
	\vspace{-10pt}
	\fontsize{8}{13}\selectfont
	\begin{align*}
			v    & ::= \bitv
			\mid \structval{f}{v}{n}                          \\
			e    & ::= v
			\mid x
			\mid \ominus e
			\mid  e \oplus e'
			\mid e.f
			\mid e [\bitv:bitv']
			\mid \structval{f}{e}{n}                          \\
			lval & ::=  x \mid lval.f \mid lval[\bitv:\bitv'] \\
			s    & ::= \texttt{skip}
			\mid lval := e
			\mid s_1  ;  s_2
			\mid  \texttt{if} \; e \; \texttt{then} \; s_1 \; \texttt{else} \; s_2
			\mid \texttt{apply}  \; tbl \mid  \\ 
			     & \qquad f(e_1 \cons e_n)
			\mid \texttt{transition} \; {\selexp{\val}{\xv}{n}}
		\end{align*}
	\vspace{-10pt}
	\caption{Syntax}
	\label{fig:syntax_grammar}
\end{figure}

Values, represented by $\val$, are either big-endian bitvectors $\bitv$ (raw packets) or structs
$\structval{f}{\val}{n}$ (representing headers).

P4 states $\mem$ are mappings from variables $x$ to values $\val$. In this slightly simplified semantics, variables are either global or local. States can thus be represented as disjoint unions $(\mem_g,\mem_l)$, where $\mem_g$ ($\mem_l$) maps global (local) variables only. 

While externs in P4 can modify the architectural state, they cannot change the P4 state itself.
To simplify our model, we integrate the architectural state into P4's global state, treating it as a part of the global state.
Therefore, in our model the externs are allowed to modify the global state of P4.
To maintain isolation between the program's global variables and the architectural state, we assume that the variable names used to represent the global state are distinct from those used for the architectural state.

Expressions $e$ use a standard selection of operators including binary $\oplus$, unary $\ominus$, comparison $\otimes$, and struct field access, as well as bitvector slicing $e[b:a]$ extracting the slice from index $a$ to index $b$ of $e$, and $\mem(e)$ is the evaluation of $e$ in state $\mem$.
An lvalue $\lval$ is an assignable expression, either a variable, a field of a struct, or a bitvector slice. 
The semantics of expressions is standard and consists of operations over bitvectors and record access.

The semantics of statements uses a mapping $E$ from function names $f$ to pairs
$(\stmt,\overline{(x,d)})$, where $\overline{(x,d)}$ is the signature of $f$, a sequence of pairs $(x_i,d_i)$ of function parameters with their directions $d_i\in\{\texttt{in},\texttt{out},\texttt{inout}\}$.
Additionally, $E$ maps parser state names $st$ to their bodies. Furthermore,
since P4 programs may depend on external components, $E$ also maps externs $f$ and tables $t$ to their respective implementations.

\begin{figure}[ht!]
	{
	\centering
	\fontsize{9}{13}\selectfont
	{
	$
		\inferrule*[Lab=\textsc{S-Skip}]
		{
		}
		{
			{\stmtctx} : \sem{\mem}{\texttt{skip}}{\mem}
		}
	$
	\nextrule
	$
		\inferrule*[Lab=\textsc{S-Assign}]
		{
		m'=m[lval \mapsto \expsem{\mem}{e}]
		}
		{
		{\stmtctx} : \sem{\mem}{lval := e}{\mem'}
		}
	$
	\nextrule
	$
		\inferrule*[Lab=\textsc{S-Seq}]
		{
		{\stmtctx} : \sem{\mem}{\stmt_1}{\mem'} \\\\
		{\stmtctx}: \sem{\mem'}{\stmt_2}{\mem''}
		}
		{
		{\stmtctx} : \sem{\mem}{\stmt_1;\stmt_2}{\mem''}
		}
	$

	\nextrule

	$
		\inferrule*[Lab=\textsc{S-Cond-T}]
		{
		\expsem{\mem}{e}={\mathrm{true}} \\\\
		{\stmtctx} : \sem{\mem}{\stmt_1}{\mem'}
		}
		{
		{\stmtctx} : \sem{\mem}{\texttt{if } e \texttt{ then } \stmt_1 \texttt{ else } \stmt_2}{\mem'}
		}
	$
	\nextrule
	$
		\inferrule*[Lab=\textsc{S-Cond-F}]
		{
		\expsem{\mem}{e}={\mathrm{false}} \\\\
		{\stmtctx} : \sem{\mem}{\stmt_2}{\mem'}
		}
		{
		{\stmtctx} : \sem{\mem}{\texttt{if } e \texttt{ then } \stmt_1 \texttt{ else } \stmt_2}{\mem'}
		}
	$

	\nextrule

	$\inferrule*[Lab=S-Call]
		{
		(\stmt,\overline{(x,d)}) = E(f) \\
		\mem_l' =\{x_{i} \mapsto \expsem{(\mem_{g}, \mem_{l})}{e_{i}}\}\\\\
		E : \sem{(\mem_{g}, \mem_{l}')}{s}{(\mem'_{g}, \mem''_{l})}}
		{
		E : \sem{(\mem_{g}, \mem_{l})}{f(e_{1} \cons e_{n})}{(\mem'_{g}, \mem_{l})[e_{i} \mapsto \mem''_{l}(x_{i}) \mid \isOut(d_{i})]}}
	$

	\nextrule

	$
		\inferrule*[Lab=\textsc{S-Extern}]
		{
		(sem_{f},\overline{(x,d)}) = \stmtctx(f) \\
		\mem'_{l} =\{x_{i} \mapsto \expsem{(\mem_{g}, \mem_{l})}{e_{i}}\}\\\\
		(\mem'_{g}, \mem''_{l}) = sem_{f}(\mem_{g}, \mem'_{l})
		}
		{
		\stmtctx : \sem{\mem}{f(e_{1} \cons e_{n})}{(\mem'_{g}, \mem_{l})[e_{i} \mapsto \mem''_{l}(x_{i}) \mid \isOut(d_{i})]}
		}
	$}

	\nextrule

	$\inferrule*[Lab=S-Trans]
		{
		\xv' =  {\begin{cases}
			\xv_{i} & \text{if } \expsem{\mem}{e} = \val_{i} \\
			\xv     & \text{otherwise}
		\end{cases}} \\
		E : \sem{\mem}{E(\xv')}{\mem'}
		}
		{
		E : \sem{\mem}{\texttt{transition } {\selexp{\val}{\xv}{n}}}{\mem'}
		}
	$

	\nextrule

	$\inferrule*[Lab=S-Table]
		{
			(\overline{e}, sem_{tbl}) = E(tbl)\\
			sem_{tbl}(\expsem{(\mem_{g}, \mem_{l})}{e_{1}} \cons \expsem{(\mem_{g}, \mem_{l})}{e_{n}}) = (a, \overline{\val})\\
			(\stmt,(x_{1},\mathrm{none}) \cons (x_{n}, \mathrm{none})) = E(a) \\
			\mem_l' =\{x_{i} \mapsto {\val_i}\}\\
			E:\sem{(\mem_{g}, \mem_{l}')}{\stmt}{(\mem_{g}', \mem_{l}'')}
		}
		{
			E:\sem{(\mem_{g}, \mem_{l})}{\texttt{apply} \ tbl }{(\mem_{g}', \mem_{l})}
		}
	$

}
	\caption{Semantic rules}
	\label{fig:semantics}
\end{figure}

The semantic rules presented in Fig.~\ref{fig:semantics} rely on judgments of the form $E : \sem{\mem_{1}}{s}{\mem_{2}}$ to represent the execution of statement $s$ under mapping $E$ which starts from state $\mem_{1}$ and terminates in $\mem_{2}$. %

Many of the rules in Fig.~\ref{fig:semantics} are standard and are therefore not explained here.
Rule \textsc{S-Call} fetches the invoked function's body $\stmt$ and signature, and copies in the arguments into $\mem'_l$, which serves as the local state for the called function and is used to execute the function's body.
Note that the function's body can modify the global state, but cannot change the caller's local state due to P4's calling conventions.
After executing the function's body, the variables in final local state $\mem_l''$ must be copied-out according to the directions specified in the function's signature. Given a direction $d_i$, the auxiliary function $\mathrm{isOut}$ returns true if the direction is \texttt{out} or \texttt{inout}. We rely on this function to copy-out the values from $\mem_l''$ back to the callee only for parameters with \texttt{out} and \texttt{inout} direction.

For example, in Program~\ref{prg:basic_congestion} let $\mem_l = \{\texttt{hdr.ipv4.ttl} \mapsto 2 \}$ when invoking \texttt{decrease} at line 44.
The local state of \texttt{decrease} (i.e. the copied-in state) becomes $\mem_l' = \{\texttt{x} \mapsto 2 \}$.
After executing the function's body (line 8), the final local state will be $\mem_l'' = \{\texttt{x} \mapsto 1 \}$ while the global state $\mem_g$ remains unchanged.
Finally, the copying out operation updates the caller's state to $\mem''=(\mem_g, \{\texttt{ttl} \mapsto 1 \})$ by updating its local state.

The \textsc{S-Extern} rule is similar to \textsc{S-Call}.
The key difference is that instead of keeping a body in $E$, we keep the extern's behavior defined through $sem_f$.
This function takes a state containing the global $\mem_g$ and copied-in state $\mem_l'$ and returns (possibly) modified global and local states, represented as $sem_f(\mem_g,\mem_l')=(\mem_g',\mem_l'')$.
Finally, the extern rule preforms a copy-out procedure similar to the function call.

The \textsc{S-Trans} rule defines how the program transitions between parser states based on the evaluation of expression $e$. It includes a default state name $\xv$ for unmatched cases.
If in program state $\mem$, expression $e$ evaluates to value $\val_{i}$, the program transitions to state name $\xv_{i}$ according to the defined value-state pattern.
However, if the evaluation result does not match any of the $\val_{i}$ values, the program instead transitions to the default state $\xv$.
For example, assume that $\mem(\texttt{hdr.eth.etherType})= \hexnumber{0800}$ on line 23 of Program~\ref{prg:basic_congestion}.
The select expression within the transition statement will transition to the state \texttt{parse\_ipv4}, and executes its body.

Rule \textsc{S-Table} fetches from $E$ the table's implementation $sem_{tbl}$ and a list of expressions $\overline{e}$ representing table's keys. 
It then proceeds to evaluate each of these expressions in the current state $(\mem_{g}, \mem_{l})$, passing the evaluated values as key values to $sem_{tbl}$. The table's implementation $sem_{tbl}$ then returns an action $a$ and its arguments $\overline{\val}$.
We rely on $E$ again to fetch the body and signature of action $a$, however, since in P4 action parameters are directionless we use \texttt{none} in the signature to indicate there is no direction. Finally, similar to \textsc{S-Call} we copy-in the arguments into $\mem_l'$, which serves as the local state for the invoked action and is used to execute the action's body.
For example, let $ \mem(\texttt{hdr.ipv4.dstAddr}) = 192.168.2.2$ at line 61, and the semantics of table \texttt{ipv4\_lpm} contains:
\begin{align*}
	192.168.2.2 & \mapsto \texttt{ipv4\_forward (4A:5B:6C:7D:8E:9F, 5)}
\end{align*}
then the table invokes action \texttt{ipv4\_forward} with arguments \texttt{4A:5B:6C:7D:8E:9F} and \texttt{5}.

\section{Types and Security Condition}
In our approach types are used to represent and track both bitvector abstractions (i.e. intervals) and security labels, and we use the same types to represent input and output policies.

In P4, bitvector values represent packet fragments, where parsing a bitvector involves slicing it into sub-bitvectors (i.e. slices), each with different semantics such as payload data or header fields like IP addresses and ports.
These header fields are typically evaluated against various subnetwork segments or port ranges.
Since header fields or their slices are still bitvectors, they can be conveniently represented as integers, enabling us to express the range of their possible values as $I = \langle a,b \rangle$, the interval of integers between $a$ and $b$.

We say a bitvector $\val$ is typed by type $\type$, denoted as $\val : \type$, if
$\type$ induces a slicing of $v$ that associates each slice with a suitable interval $I$ and security label $\ell \in \{{\low}, {\high}\}$.
We use the shorthand ${I}^{\lbl}_{i}$ to represent a slice of length $i$, with interval $I \subseteq \langle 0, 2^{i}-1 \rangle$ and security label $\lbl$.
The bitvector type can therefore be presented as $\type={{I_n}^{\lbl_n}_{i_n}} {\cdots} {{I_1}^{\lbl_1}_{i_1}}$, representing a bitvector of length $\Sigma_{j=1}^n i_j$ with $n$ slices, where each slice $i$ has interval $I_i$ and security label $\lbl_i$.
Singleton intervals are abbreviated $\langle a \rangle$, $\langle \rangle$ is the empty interval, and $\langle * \rangle$ is the complete interval, that is, the range $\langle 0, 2^{i}-1 \rangle$ for a slice of length $i$.
Function $\getlabel(\type)$ indicates the least upper bound of the labels of slices in $\type$.

To illustrate this, let $\type_1$ be $\tbvn{*}{\high}{2} {\cdot} \tbvn{0,1}{\low}{3}$ which types a bitvector of length $5$ consisting of two slices.
The first slice has a length of $3$, with values drawn from the interval $\langle 0,1\rangle$ and security label {\low}.
The second slice, with a length of 2, has a security label {\high}, and its values drawn from the complete interval $\langle 0,3 \rangle$ (indicated by $*$).
Accordingly, $\getlabel(\type_1)$ evaluates to $\high \sqcup \low = \high$.

Type $\type$ is also used to denote a record type, where record ${\structval{f}{\val}{n}}$ is typed as ${\typestrnew{f}{\type}{n}}$ if each value $\val_i$ is typed with type $\type_i$. %

In this setting, the types are not unique, as it is evident from the fact that a bitvector can be sliced in many ways and a single value can be represented by various intervals. For example, bitvector
$\bitvector{100}$
can be typed as $\tbvn{4}{\low}{3}$,
or $\tbvn{*}{\low}{1} {\cdot} \tbvn{0,1}{\low}{2}$,
or $\tbvn{2}{\low}{2} {\cdot} \tbvn{*}{\low}{1}$. %

\tightpar{State types}
A type environment, or \emph{state type}, $\tscope = (\tscope_g,\tscope_l)$ is a pair of partial functions from variable names $x$
to types $\type$. %
Here, $\tscope_{g}$ and $\tscope_{l}$ represent global and local state types, respectively, analogous to the global ($\mem_{g}$) and local ($\mem_{l}$) states in the semantics.
We say that $\tscope$ can type state $\mem$, written as ${\tscope \vdash \mem}$, if for every $lval$ in the domain of $\mem$, the value $\mem(lval)$ belongs to the interval specified by $\tscope(lval)$; formally, $\forall \ lval \in \text{domain}(\mem), \ {\expsem{\mem}{lval} : \tscope(lval)}$. Note that the typing judgment ${\tscope \vdash \mem}$ is based on the interval inclusion and it is independent of any security labels. For example, if $\mem(x) = 257$ then $257$ is considered well-typed wrt. $\tscope$ if and only if $\tscope(x)$ is an interval that contains $257$ (e.g., $\tbvn{0,257}{}{}$, $\tbvn{257,257}{}{}$, or $\tbvn{100,300}{}{}$).

A state type might include a type with an empty interval; we call this state type \emph{empty} and denote it as $\bullet$.

Let $\labelof{\lval}{\tscope}$ be the least upper bound of the security labels of all the slices of ${\lval}$ in state type $\tscope$.
The states $m_1$ and $m_2$ are considered \emph{low equivalent with respect to} $\tscope$,
denoted as $\mem_1 \lequiv{\tscope} \mem_2$,
  if for all $\lval$ such that
  $\labelof{\lval}{\tscope} = \low$, then
  $\expsem{\mem_1}{\lval} = \expsem{\mem_2}{\lval}$ holds.

\begin{example}
  Assume a state type ${\tscope = \{x \mapsto { {\tbvn{*}{\high}{1}} {\cdot} {\tbvn{0,1}{\low}{2}}  } \}}$
  The following states $\mem_1 = \{x \mapsto \bitvector{000} \ \}$ and $\mem_2 = \{x \mapsto \bitvector{100} \ \}$ are low equivalent wrt. $\tscope$.
  However, states $\mem_1 = \{x \mapsto \bitvector{000} \ \}$ and $\mem_3 = \{x \mapsto \bitvector{101} \ \}$ are not low equivalent even though both can be typed by $\tscope$.
\end{example}

\tightpar{Contracts}
A table consists of key-action rows, and in our threat model, we assume the keys and actions of the tables are always public (i.e. {\low}), but the arguments of the actions \emph{can} be secret (i.e. {\high}). %
Given that tables are populated by the control-plane, the behavior of a table is unknown at the time of typing. We rely on user-specified contracts to capture a bounded model of the behavior of the tables.
In our model, a table's contract has the form $(\overline{e},\text{Cont}_{tbl})$, where $\overline{e}$ is a list of expressions indicating the keys of the table, and $\text{Cont}_{tbl}$ is a set of tuples $(\phi, (a, \overline{\tau}))$, where $\phi$ is a boolean expression defined on $\overline{e}$, and $a$ denotes an action to be invoked with argument types $\overline{\tau}$ when $\phi$ is satisfied.

For instance, the \texttt{ipv4\_lpm} table of Program \ref{fig:basic_congestion} uses \texttt{hdr.ipv4.dstAddr} as its key, and can invoke two possible actions: \texttt{drop} and \texttt{ipv4\_forward}. An example of a contract for this table is depicted in Fig. \ref{fig:table_contract_example}. This contract models a table that forwards the packets with $\texttt{hdr.ipv4.dstAddr} = \texttt{192.*.*.*}$ to ports $1$-$9$, the ones with $\texttt{hdr.ipv4.dstAddr} = \texttt{10.*.*.*}$ to ports $10$-$20$, and {\texttt{drop}}s all the other packets. Notice that in the first case, the first
argument resulting from the table look up is secret.
\begin{figure}[ht]
	\vspace{-\baselineskip}
	\footnotesize
	\centering
	\begin{align*}
		\big(\ 
		&[\texttt{hdr.ipv4.dstAddr}], \\
			\big\{
				&\big(\texttt{dstAddr}[31:24] = 192 , (\actiontt{ipv4\_forward} , [\tbvn{*}{{\high}}{48},\tbvn{1,9}{{\low}}{9}]) \big) \\
				&\big(\texttt{dstAddr}[31:24] = 10 , (\actiontt{ipv4\_forward} , [\tbvn{*}{{\low}}{48},\tbvn{10,20}{{\low}}{9}]) \big) \\
				&\big(\texttt{dstAddr}[31:24] \not= 192 \wedge \texttt{dstAddr}[31:24] \not= 10 , (\actiontt{drop}, \ []) \big)
			\big\}
		\big)
	\end{align*}
	\caption{The contract of \texttt{ipv4\_lpm} table}
	\label{fig:table_contract_example}
\end{figure}

The table contracts are essentially the security policies of the tables, where $\phi$ determines a subset of table rows that invoke the same action ($a$) with the same argument types ($\overline{\type}$). %
Using the labels in $\overline{\type}$, and given action arguments $\overline{\val}_1$ and $\overline{\val}_2$, we define $\overline{\val}_1 \lequiv{\overline{\type}} \overline{\val}_2$ as 
$|\overline{\val}_1| = |\overline{\val}_2| = |\overline{\type}|$ and for all $i$, 
$\val_{1_i} : \type_i $ and $\val_{2_i} : \type_i $, and
if $\getlabel(\type_i) = \low$ then $\val_{1_i} = \val_{2_i}$.
Note that $\getlabel(\type_i)$ returns the least upper bound of the labels of all $\type_i$'s slices, hence if there is even one {\high} slice in $\type_i$, $\getlabel(\type_i)$ would be {\high}.
We use mapping $T$ to associate table names $tbl$ with their contracts. 

We say that two mappings $E_1$ and $E_2$ are considered \emph{indistinguishable} wrt. $T$, denoted as ${E_1 \lequiv{T} E_2}$, if for all tables $tbl$ such that $({\overline{e}_1, sem_{1_{tbl}}) = E_1(tbl)}$, $({\overline{e}_2, sem_{2_{tbl}}) = E_2(tbl)}$, $(\overline{e},\text{Cont}_{tbl}) = T(tbl)$ then $\overline{e}_1 = \overline{e}_2 = \overline{e}$, and for all $(\phi, (a, \overline{\type})) \in \text{Cont}_{tbl} $, and for all arbitrary states $\mem_1$ and $\mem_2$,
such that $\mem_1(\overline{e}) = \mem_2(\overline{e}) = \overline v$ and $\mem_1(\phi) \Leftrightarrow \mem_2(\phi)$, 
if $\mem_1(\phi)$ then
$\overline v_{1}, \overline v_{2}$ exist such that
${sem_{1_{tbl}}(\overline{v}) = (a, \overline{\val}_1)}$, ${sem_{2_{tbl}}(\overline{v}) = (a, \overline{\val}_2)}$, and $\overline{\val}_1 \lequiv{\overline{\type}} \overline{\val}_2$.
In other words, $T$-indistinguishability of $E_1$ and $E_2$ guarantees that given equal key values, $E_1$ and $E_2$ return the same actions with $\overline{\type}$-indistinguishable arguments $\overline{v_1}$ and $\overline{v_2}$ such that these arguments are in bound wrt. their type $\overline{\type}$.

\tightpar{Security condition}
As explained in Section \ref{sec:overview} the input and output policy cases are expressed by assigning types to program variables.
State types, specifying security types of program variables, are used to formally express input and output policy cases. Hereafter, we use $\tscope_{i}$ and $\tscope_{o}$ to denote input and output policy cases, respectively.
Using this notation, the input policy, denoted by $\Gamma_{i}$, is represented as a set of input policy cases $\tscope_{i}$.
Similarly, the output policy is expressed as a set of output policy cases $\tscope_{o}$ and denoted by $\Gamma_{o}$.

Given this intuition, we say two states $\mem_{1}$ and $\mem_{2}$ are \emph{indistinguishable} wrt. a policy case $\tscope$ if $\tscope\vdash m_1$, $\tscope\vdash m_2$, and $\mem_1 \lequiv{\tscope} \mem_2$. 
Relying on this, we present our definition of noninterference as follows:

  \begin{definition}[Noninterference]\label{def:sec_cond}
    A program $s$ is \emph{noninterfering} wrt. the input and output policy cases $\tscope_{i}$ and $\tscope_{o}$,  and table contract mapping $T$,
     if for all mappings $E_1$, $E_2$ and states $\mem_{1}$, $\mem_{2}$, $m'_{1}$ such that:
    \begin{itemize}
      \item $E_1 \consistent{T}{} E_2$,
      \item $\tscope_{i} \vdash \mem_{1}$, $\tscope_{i} \vdash \mem_{2}$, and $\mem_1 \lequiv{\tscope_{i}} \mem_2$,
      \item  $E_1:\sem{m_{1}}{s}{m'_{1}}$
    \end{itemize}
    there exists a state $\mem'_2$ such that:
    \begin{itemize}
      \item $E_2:\sem{\mem_{2}}{s}{\mem'_{2}}$,
      \item  if $\tscope_{o} \vdash \mem'_{1}$, then $\tscope_{o} \vdash \mem'_{2}$
            and $\mem'_1 \lequiv{\tscope_{o}} \mem'_2$.
    \end{itemize}
  \end{definition}

The existential quantifier over the state $\mem'_2$ does not mean that the language is non-deterministic, in fact if such state exists it is going to be unique. 
This existential quantifier guarantees that our security condition is termination sensitive, meaning that it only accepts the cases where the program terminates for both initial states $\mem_1$ and $\mem_2$.

Intuitively, Definition~\ref{def:sec_cond}  relies on two different equivalence relations: one induced by the input policy case and one by the output policy case.
The former induces a  partial equivalence relation (PER) \cite{DBLP:journals/lisp/SabelfeldS01},  $P_{\tscope_{i}}(\mem_1, \mem_2) = \tscope_{i} \vdash \mem_1 \land \tscope_{i} \vdash \mem_2 \land \mem_1\lequiv{\tscope_{i}} \mem_2$, such that the domain contains
only states that satisfy the intervals of $\tscope_{i}$.
Similarly, the latter induces $Q_{\tscope_{o}}(\mem_1, \mem_2) = (\tscope_{o} \vdash \mem_1 \land \tscope_{o} \vdash \mem_2 \land \mem_1 \lequiv{\tscope_{i}} \mem_2) \vee (\tscope_{o} \not \vdash \mem_1 \land \tscope_{o} \not \vdash \mem_2 )$, which is an equivalence relation (ER).
A program is then noninterfering wrt. $\tscope_i$ and $\tscope_o$ if every class of the PER $P_{\tscope_i}$ is mapped to a class of the ER $Q_{\tscope_o}$.

  This condition implies the following intuitive assumptions: (1) the policy cases are public knowledge, (2) entailment of a state on the intervals of a policy case is public knowledge, (3) the states that do not entail the intervals of an input policy (i.e., those outside the domain of the PER) are considered entirely public, and their corresponding execution is unconstrained,
  (4) the attacker can observe whether the final state entails the intervals of the output policy, and (5) the attacker cannot observe any additional information about states that do not entail the intervals of the output policy.
 We use the following examples to further discuss our security condition.

\begin{example}
  \label{example:policydistinguish}
 
   Consider the input and output policy cases: 

  \vspace{-12pt}
  \begin{align*}
    &\tscope_{i} = \{a \mapsto {\tbvn{0,256}{\high}{9}},\; b \mapsto {\tbvn{*}{\low}{9}}\} \\
    &\tscope_{o} = \{a \mapsto {\tbvn{*}{\high}{9}},\; b \mapsto {\tbvn{1025,*}{\low}{9}}\}
  \end{align*}
  \vspace{-12pt}

  \begin{itemize}
    \item The inputs $a_1 = 257$ and $a_2 = 258$ are distinguishable by the attacker, since they do not fall in the {\high} interval $\tbvn{0,256}{}{}$ under $\tscope_{i}$. %
    \item The inputs $a_1 = 0$ and $a_2 = 256$ are indistinguishable, since they belong to the {\high} interval $\tbvn{0,256}{}{}$ under $\tscope_{i}$.
    \item Under $\tscope_{o}$, any value $b \geq 1025$  is distinguishable by the attacker, otherwise it is indistinguishable since it falls outside the  {\low} interval $\tbvn{1025,*}{}{}$.
  \end{itemize}

\end{example}

  We consider pairs of states $\mem_1$, $\mem_2$ such that $\tscope_{i} \vdash \mem_{1}$, $\tscope_{i} \vdash \mem_{2}$, and $\mem_1 \lequiv{\tscope_{i}} \mem_2$. For example,

  \begin{enumerate}
   \item $\mem_1(a) = a_1$ and $\mem_2(a) = a_2$, where $a_1 \ , \ a_2 \in \tbvn{0,256}{}{}$.
  \item $\mem_1(b)=\mem_2(b) = b_0$.
 \end{enumerate}
 
\noindent
We use the above policies and states to discuss the security condition of the following one-line programs:

\begin{enumerate}
  \item[$\bullet$] \texttt{b=a} 
  
  This program yields $m_1'(b)=a_1$ and $m_2'(b)=a_2$. Since $\mem'_1(b) \not\in \tbvn{1025,*}{}{}$, then $\tscope_o \not\vdash \mem_1'$, hence noninterference is trivially satisfied.
  Intuitively, despite the variable $a$ being {\high}, the output on the variable $b$  is not observable by the attacker. %

  \item[$\bullet$] \texttt{if (a<=1024) then b=a else skip} 
  
  This program yields $m_1'(b)=a_1$ and $m_2'(b)=a_2$. Since both $a_1$ and $a_2$ are in $\tbvn{0,256}{}{}$, the program executes \text{b=a} in the true branch. Thus, $\tscope_o \not\vdash \mem_1'$ and  noninterference is trivially satisfied.

  \item[$\bullet$] \texttt{b=a+1000} 
    
  \begin{enumerate}[label={\roman*}]
    \item[(i)] Let $a_1 =25$ and $a_2=26$. Then $m_1'(b) = 1025$ and $m_2'(b)=1026$ indicating $\tscope_o \vdash \mem_1'$ and $\tscope_o \vdash \mem_2'$, however $\mem'_1 \not\lequiv{\tscope_{o}} \mem'_2$, hence the program is interfering.
    \item[(ii)] Let $a_1 =25$ and $a_2=0$. Then $m_1'(b) = 1025$ and $m_2'(b)=1000$ indicating $\tscope_o \vdash \mem_1'$ and $\tscope_o \not\vdash \mem_2'$, hence the program is interfering.
  \end{enumerate}

\end{enumerate}

\begin{example}
  Assume program $\stmtif{y==1}{x=1}{x=x+1}$,
  input policy case ${\tscope_i = \{x \mapsto {\tbvn{*}{\high}{2}}, y \mapsto {\tbvn{1}{\low}{3}}\}}$,
  and initial states $\mem_1 = \{x \mapsto \bitvector{10}, y \mapsto \bitvector{001} \  \}$
  and $\mem_2 = \{x \mapsto \bitvector{01}, y \mapsto \bitvector{001} \ \}$.
  We can see that
  $\tscope_{i} \vdash \mem_{1}$, $\tscope_{i} \vdash \mem_{2}$, and $\mem_1 \lequiv{\tscope_{i}} \mem_2$.
  In a scenario where the only initial states are $\mem_{1}$ and $\mem_{2}$,
  executing this program would result in final states
  $\mem'_1 = \{x \mapsto \bitvector{01}, y \mapsto \bitvector{001} \  \}$
  and $\mem'_2 = \{x \mapsto \bitvector{01}, y \mapsto \bitvector{001} \ \}$, respectively.
  Given output policy case $\tscope_o = [x \mapsto {\tbvn{*}{\low}{2}}, y \mapsto {\tbvn{1}{\low}{3}}]$, 
  we say that this program is noninterfering wrt. $\tscope_o$ because
  $\tscope_{o} \vdash \mem'_{1}$, $\tscope_{o} \vdash \mem'_{2}$,
  and $\mem'_1 \lequiv{\tscope_{o}} \mem'_2$.
\end{example}

We extend the definition of noninterference to input policies $\Gamma_{i}$ and output policies $\Gamma_{o}$, requiring the program to be noninterfering for \emph{every pair} of input and output policy cases.
In our setting, the output policy, which indicates the shape of the output packets, describes what the attacker observes.
As such, it is typically independent of the shape of the input packet and the associated input policy.
Thus, our approach does not directly pair input and output policy cases. 
Instead, it ensures that the program is noninterfering for all combinations of input and output policy cases.

\label{sec:security_types}

\section{Security Type System}\label{sec:type_system}

We introduce a security type system that combines security types and interval abstractions.
Our approach begins with an input policy case and conservatively propagates labels and intervals of P4 variables.
In the following, we assume that the P4 program is well-typed.

\subsection{Typing of expressions}\label{sec:type_exp}
The typing judgment for expressions is $\tscope \vdash {e} : \type$. Rules for values, variables, and records are standard and omitted here.

P4 programs use bitvectors to represent either raw packets (e.g. \verb|packet_in packet| of line 12) or finite integers (e.g. \verb|x| of line 8).
While there is no distinction between these two cases at the language level, it is not meaningful to add or multiply two packets, as it is not extracting a specific byte from an integer representing a time-to-live value.
For this reason, we expect that variables used to marshal records have multiple slices but are not used in arithmetic operations, while variables used for integers have one single slice and are not used for sub-bitvector operations.
This allows us to provide a relatively simple semantics of the slice domain, which is sufficient for many P4 applications.

\begin{figure}[h]
	{
		\centering
		\fontsize{9}{13}\selectfont
		$\inferrule*[Lab=T-SingleSliceBs]
			{
			{\tscope}\vdash{e_{1}} :  \tbvn{I_{1}}{\lbl_{1}}{i}  \\
			{\tscope}\vdash{e_{2}} :  \tbvn{I_{2}}{\lbl_{2}}{i}
			}
			{
			{\tscope}\vdash{e_{1} \oplus e_{2}} :
			\tbvn{I_{1} \oplus I_{2}}{\lbl_{1} \sqcup \lbl_{2}}{i} \\\\
			{\tscope}\vdash \ominus e_{1} :
			\tbvn{ \ominus I_{1}    }{\lbl_{1}}{i}   \\\\
			{\tscope}\vdash e_{1} \otimes e_{2} :
			\tbvn{I_{1} \otimes I_{2}}{\lbl_{1} \sqcup \lbl_{2}}{i}
			}
		$

	}
\end{figure}

\textsc{T-SingleSliceBs} rule allows the reuse of standard interval analysis for
binary, unary, and comparison operations over bitvectors that have only \emph{one} single slice.
The resulting label is the least upper bound of labels associated with the input types.

\begin{figure}[h]
	{
		\centering
		\fontsize{9}{13}\selectfont
		$\inferrule*[Lab=T-AlignedSlice]
		{
			\tscope \vdash e:
			\tbvn{I_{n}}{\lbl_{n}}{i_{n}} {\cdots} \tbvn{I_{1}}{\lbl_{1}}{i_{1}}
		}
		{
			\tscope \vdash
			e[\sumindex{j=1}{b} i_j: \sumindex{j=1}{a} i_j ]:
			\tbvn{I_{b}}{\lbl_{b}}{i_{b}} {\cdots} \tbvn{I_{a}}{\lbl_{a}}{i_{a}}
		}
		$
		
		\nextrule
		
		$\inferrule*[Lab=T-NonAlignedSlice]
		{
			\tscope \vdash e:
			\tbvn{I_{n}}{\lbl_{n}}{i_{n}} {\cdots} \tbvn{I_{1}}{\lbl_{1}}{i_{1}}
		}
		{
			\tscope \vdash e[b:a]:
			\tbvn{*}{\bigsqcup \lbl_{i_{j}}}{b-a}
		}
		$
		
	}
\end{figure}

In the slicing rules, sub-bitvector (i.e. $e[b:a]$) preserves precision only if the slices of the input are aligned with sub-bitvector's indexes, otherwise sub-bitvector results in $\langle * \rangle$, representing all possible values.
The following lemmas show that interval and labeling analysis of expressions is sound:
\begin{lemma}
	Given expression $e$, state $\mem$, and state type $\tscope$ such that $\tscope \vdash \mem$, if the expression is well-typed $\tscope \vdash e:\type$, and evaluates to a value $\expsem{\mem}{e} = \val$, then:
	\begin{itemize}
		\item $\val$ is well-typed wrt. to the interval of type $\type$ (i.e. $\val:\type$).
		\item for every state $\mem'$ such that $\mem \lequiv{\tscope} \mem'$,
		      if $\getlabel(\type) = {\low}$, then $\expsem{\mem'}{e} = \val$.
	\end{itemize}
\end{lemma}

\subsection{Typing of statements}\label{sec:type_stm}
To present the typing rules for statements, we rely on some auxiliary notations and operations to manipulate state types, which are introduced informally here due to space constraints. The properties guaranteed by these operations are reported in {\iffull Appendix \ref{sec:type_stm_op_app}. \else the full version of the paper \cite{fullreport}. \fi} 

$\tscope[\lval \mapsto \type]$ indicates updating the
type of $\lval$, which can be a part of a variable, in state type $\tscope$.
 $\tscope \concat \tscope'$ updates $\tscope$ such that for every variable in the domain of $\tscope'$, the type of that variable in $\tscope$ is updated to match $\tscope'$.
$\refine{\tscope}{e}$
returns an overapproximation of $\tscope$ that satisfy the abstraction of
$\tscope$ and the predicate $e$.
$\join(\tscope_{1}, \tscope_{2})$
returns an overapproximation of $\tscope_{1}$, whose labels are at least as restrictive as $\tscope_{1}$ \emph{and} $\tscope_{2}$.
These operations tend to overapproximate, potentially causing a loss of precision in either the interval or the security label, as illustrated in the following example:
\begin{example}
  Let $x$ be mapped to an interval between $2$ and $8$, or in binary, bitvectors between $\bitvector{0010}$ and  $\bitvector{1000}$, in $\tscope$. That is,
$\tscope = \{ x \mapsto   \tbvn{2,8}{{\low}}{4} \}$.
	The following update ${\tscope[x[3:3] \mapsto \tbvn{0}{{\high}}{1} ]}$ modifies the slice $x[3:3]$ and results in the state type
	$\{ x \mapsto \tbvn{0}{{\high}}{1} {\cdot} \tbvn{*}{{\low}}{3} \}$.
	Here, lvalue $x[2:0]$ loses precision because after updating $x[3:3]$,
	the binary representation of the interval of lvalue $x[2:0]$ would be between $\bitvector{010}$ and $\bitvector{000}$,
	that is every 3-bit value except $\bitvector{001}$.
	Such value set cannot be represented by a single continuous interval, hence we overapproximate to the complete interval $\langle * \rangle$.

	Similarly, the operation $\refine{\tscope}{x[3:3] < 1}$ updates the interval of lvalue $x[3:3]$
	which results in $ \{ x \mapsto \tbvn{0}{{\low}}{1} {\cdot}  \tbvn{*}{{\low}}{3} \}$
	where lvalue $x[2:0]$ again loses precision.

	On the other hand, an operation such as $\text{join}(\tscope,  \{ x \mapsto \tbvn{*}{{\high}}{1} {\cdot} \tbvn{*}{{\low}}{3}  \})$
	does not modify the intervals of $\tscope$, but since the $\tbvn{*}{{\high}}{1}$ slice overlaps with a slice of $x$ in $\tscope$
	its label should be raised, which results in $\tscope' = \{ x \mapsto   \tbvn{2,8}{{\high}}{4} \}$.
\end{example}

The security typing of statement $\stmt$ uses judgments of the form $T, pc, \tscope \vdash {\stmt} : \Gamma$, where $pc$ is the security label of the current program context, $T$ is a static mapping, and $\tscope$ is a state type.
We use $T$ to map a parser state name ($st$) or function name ($f$) to their bodies. For functions, $T$ also returns their signatures.
Moreover, as described in Section \ref{sec:security_types}, we also use $T$ to map externs and tables to their contracts.
The typing judgment concludes with $\Gamma$, which is a set of state types.
In our type system, the security typing is not an on-the-fly check that immediately rejects a program when encountering an untypeable statement.
Instead, we proceed with typing the program and produce a state type for each path and accumulate all of those in a final set $\Gamma$. This is done in order to increase precision, by minimizing the need to unify, and hence overapproximate, intermediate typings during type derivation. This is indeed one of the key technical innovations of our type system, as explained in more detail below.
Once the final set $\Gamma$ is obtained, the state types within $\Gamma$ are then verified against the output security policy $\Gamma_o$, ensuring that they meet all the output policy cases $\tscope_{o}$ in $\Gamma$.

In the following rules, we use $\raisel(\type, \lbl)$ to return a type where each label $\lbl'$ within $\type$ has been updated to $\lbl' \sqcup \lbl$.

\begin{figure}[h]
	{
		\centering
		\fontsize{9}{13}\selectfont
		$\inferrule*[Lab=T-Assign]
		{
			\tscope \vdash {e} : \type \\
			\type' = \raisel(\type, pc) \\
			\tscope' = \tscope[\lval \mapsto \type']
		}
		{
			T, pc, \tscope \vdash \lval := e : \{\tscope'\}
		}
		$

	}
\end{figure}

\parheading{\textsc{T-Assign}} rule follows the standard IFC convention. It updates the type of the left-hand-side of the assignment (i.e. $lval$) with the type of expression $e$ while raising its security label to the current security context $pc$ in order to capture indirect information flows.

\begin{figure}[h]
	{
		\centering
		\fontsize{9}{13}\selectfont
		$\inferrule*[Lab=T-Seq]
		{
			T, pc, \tscope \vdash {\stmt_1} : \Gamma_1 \\
			\forall \tscope_{1} \in \Gamma_{1}. \ T, pc, \tscope_1 \vdash \stmt_2 : \Gamma^{\tscope_{1}}_2 \\
			\Gamma' =  \bigcup_{\tscope_1 \in \Gamma_1} \Gamma^{\tscope_{1}}_2
		}
		{
			T, pc, \tscope \vdash {\stmt_1};{\stmt_2} : \Gamma'
		}
		$
	
	}
\end{figure}

\parheading{\textsc{T-Seq}} types the sequential composition of two statements. This rule type checks the first statement $\stmt_1$, gathering all possible resulting state types into an intermediate set $\Gamma_1$. Then, for each state type in this intermediate set, the rule type checks the second statement $\stmt_2$, and accumulates all resulting state types into the final state type set $\Gamma'$.

\begin{figure}[h]
	{
		\centering
		\fontsize{9}{13}\selectfont
		$\inferrule*[Lab=T-Cond]
		{
			\tscope \vdash e : \type \\
			\lbl = \getlabel(\type)\\
			pc' = pc \sqcup \lbl \\\\
			T, pc', (\refine{\tscope}{e}) \vdash \stmt_1 : \Gamma_1 \\
			T, pc', (\refine{\tscope}{\neg e}) \vdash \stmt_2 : \Gamma_2
		}
		{
			T, pc, \tscope \vdash \texttt{if } e \texttt{ then } \stmt_1 \texttt{ else } \stmt_2 :
			\joinOnHigh{\Gamma_{1} \cup \Gamma_{2}}{\lbl}
		}
		$
		
	}
\end{figure}

\parheading{\textsc{T-Cond}} rule types the two branches using state types refined with the branch condition and its negation, which results in the state type sets $\Gamma_1$ and $\Gamma_2$, respectively.
The final state type set is a simple union of $\Gamma_1$ and $\Gamma_2$.

However, in order to prevent implicit information leaks, if the branch condition is {\high}, the security labels of $\Gamma_1$ and $\Gamma_2$ should be joined. We do this by the auxiliary function $\mathrm{joinOnHigh}$, defined as follows:
\begin{align*}
	\joinOnHigh{\Gamma}{\lbl} =  {\begin{cases}
						\join(\Gamma) & \text{if } \lbl = {\high} \\
						\Gamma   & \text{otherwise}
				\end{cases}}
\end{align*}
where the join operator has been lifted to $\Gamma$ and defined as $\join(\Gamma) = \{\join(\tscope, \Gamma) \mid \tscope \in \Gamma\}$, $\join(\tscope, \{\tscope'\} \cup \Gamma) = \join(\join(\tscope, \tscope'), \Gamma)$ and $\join(\tscope, \varnothing) = \tscope$.

\begin{example}\label{example:cond}
	Consider the conditional statement on line 56 of Program \ref{prg:basic_congestion}, where initially $\tscope = \{ \texttt{enq\_qdepth} \mapsto \tbvn{*}{{\high}}{19}, \texttt{hdr.ipv4.ecn} \mapsto \tbvn{*}{{\low}}{2}, \dots \}$.
	Since the label of \texttt{enq\_qdepth} is {\high}, after the assignment on line 57, \texttt{hdr.ipv4.ecn} becomes {\high} in $\Gamma_{1}$.
	However, since there is no \texttt{else} branch, $\stmt_2$ is trivially \texttt{skip}, meaning that \texttt{hdr.ipv4.ecn} remains {\low} in $\Gamma_{2}$.
	Typically, in IFC, the absence of an update for \texttt{hdr.ipv4.ecn} in the \texttt{else} branch leaks that the if statement's condition does not hold. To prevent this, we join the security labels of all state types if the branch condition is ${\high}$.
	Therefore, in the final state set $\Gamma'$, \texttt{hdr.ipv4.ecn} is labeled ${\high}$.
\end{example}

Even on joining the security labels, \textsc{T-Cond} does not merge the final state types in order to maintain abstraction precision.
To illustrate this consider program $\stmtif{b}{x[0:0]=0}{x[0:0]=1}$,
where the $pc$ and the label of $b$ are both ${\low}$,
and an initial state type $\tscope = \{ x \mapsto \tbvn{*}{{\high}}{3} ; b \mapsto \tbvn{*}{{\low}}{1}\}$.
After typing both branches, the two typing state sets are
$\Gamma_1 = \{\{ x \mapsto \tbvn{*}{{\high}}{2} {\cdot} \tbvn{0}{{\low}}{1} \}\} $
and $\Gamma_2 = \{\{ x \mapsto \tbvn{*}{{\high}}{2} {\cdot} \tbvn{1}{{\low}}{1} \}\}$.
Performing a union after the conditional preserves the labeling and abstraction precision of \texttt{x[0:0]},
whereas merging them would result in a loss of precision.

\begin{figure}[h]
	{
		\centering
		\fontsize{9}{13}\selectfont
		$\inferrule*[Lab=T-Trans]
		{
			{\tscope} \vdash {e}: \type \\
			\lbl = \getlabel(\type) \\
			pc' = pc \sqcup \lbl \\\\
			\tscope'_i=  \refine{\tscope}{e=v_i \land \bigwedge_{j<i} e \neq v_{j}} \\
			T, pc', \tscope'_i \vdash T(st_i) : \Gamma_i \\\\
			\tscope'_d=  \refine{\tscope}{\bigwedge_{i} e \neq v_{i}} \\
			T, pc', \tscope'_d \vdash T(st) : \Gamma_d \\\\
			\Gamma' = \Gamma_d \cup ( \bigcup_{i} \Gamma_i ) \\
			\Gamma'' =  \joinOnHigh{\Gamma'}{\lbl}
		}
		{
			T, pc, \tscope \vdash \texttt{transition } {\selexp{\val}{\xv}{n}}: \Gamma''
		}
		$
		
	}
\end{figure}

\parheading{\textsc{T-Trans}} rule types parser transitions. Similar to \textsc{T-Cond}, it individually types each state's body and then joins or unions the final state types based on the label of $pc$.

\begin{figure}[h]
	{
		\centering
		\fontsize{9}{13}\selectfont
		$\inferrule*[Lab=T-EmptyType]
		{
		}
		{
			T, {\low}, \emptyframe \vdash \stmt: \Gamma
		}
		$
		
	}
\end{figure}

\parheading{\textsc{T-EmptyType}}
Refining a state type might lead to an empty abstraction for some variables.
We call these states empty and denote them by {$\bullet$}.
An empty state indicates that there is no state $\mem$ such that $\bullet \vdash \mem$.
The rule states that from an empty state type,
any statement can result in any final state type,
since there is no concrete state that matches the initial state type.
Notice that $\Gamma$ can simply be \emph{empty} and allow the analysis to prune unsatisfiable paths.
This rule applies \emph{only} when $pc$ is {\low}.
For cases where $pc$ is {\high}, simply pruning the empty states is \emph{unsound}, as illustrate by the following example:

\begin{example}\label{example:invalid_state}
	Assume the state type 
	$\tscope = \{ \texttt{enq\_qdepth} \mapsto \tbvn{5}{{\high}}{19}, \texttt{hdr.ipv4.ecn} \mapsto \tbvn{*}{{\low}}{2}, \hdots \}$, 
	upon reaching the conditional statement on line 56 of Program \ref{prg:basic_congestion}.
	The refinement of the \texttt{then} branch under the condition $\texttt{enq\_qdepth} \ge \texttt{THRESHOLD}$
	(where \texttt{THRESHOLD} is a constant value $10$) results in the empty state
	$\bullet = \{ \texttt{enq\_qdepth} \mapsto \tbvn{}{{\high}}{19}, \hdots \}$, where $\tbvn{}{{\high}}{19}$ denotes an empty interval.
	If we prune this empty state type, the final state type set $\Gamma'$ contains only the state types obtained from the \texttt{else} branch (which is \texttt{skip}).
	This is unsound because a {\low}-observer would be able see that the value of \texttt{hdr.ipv4.ecn} has remained unchanged
	and  infer that the {\high} field \texttt{enq\_qdepth} was less than $10$.
\end{example}

There is a similar problem of implicit flows in dynamic information flow control, where simply upgrading a {\low} variable to {\high} in only one of the branches when $pc$ is {\high} might result in partial information leakage.
This is because the variable contains {\high} data in one execution while it might remain {\low} on an alternative execution.
To overcome this problem, many dynamic IFC methods employ the so-called no-sensitive-upgrade (NSU) check~\cite{austin2009efficient}, which terminates the program's execution whenever a {\low} variable is updated in a {\high} context.
Here, to overcome this problem, we type all the statements in all branches whenever the $pc$ is ${\high}$,
even when the state type is empty \cite{1550998,BalliuSS17}.
For instance, in Example \ref{example:invalid_state}, we type-check the then branch under an empty state type, and by rule \textsc{T-Cond} the security labels of the final state types of both branches are joined, resulting in \texttt{hdr.ipv4.ecn}'s label being ${\high}$ in all the final state types.

\begin{figure}[h]
	{
		\centering
		\fontsize{9}{13}\selectfont
		$\inferrule*[Lab=T-Call]
			{
			  \gamma \vdash  \vec e : \vec \type \\
			  \tCall(T, f, pc, \vec \type, \gamma, \Gamma)
			}
			{
			  T, pc, \gamma \vdash f(\vec e) : \Gamma
		  }
		  $
		  
	}
  \end{figure}

\parheading{\textsc{T-Call}} rule types function calls.
It individually types the function arguments $e_{i}$ to obtain their types $\type_{i}$, and passes them to auxiliary function $\tCall$, defined as:
\begin{align*}
	&{(\stmt, \overline{(x,d)})  = T(f)} \quad
	{\tscope_f =\{x_{i} \mapsto \type_{i}\}} \quad
	{T, pc, (\tscope_g, \tscope_f) \vdash \stmt : \Gamma'}
	\\
	&\Gamma = \{(\tscope'_{g}, \tscope_{l})  [e_{i} \mapsto \tscope'_{f}(x_{i}) \mid \isOut(d_{i})] \mid (\tscope'_{g}, \tscope'_{f}) \in \Gamma' \}
\end{align*}
which retrieves the function's body $\stmt$ and its signature $\overline{(x,d)}$ from the mapping $T$.
Creates a new local state type $\tscope_{f}$ by assigning each argument to its corresponding type (i.e. copy-in), and then types the function's body to obtain the resulting state type set $\Gamma'$.
Finally, $\tCall$ produces $\Gamma$ by copying out the \texttt{out} and \texttt{inout} parameters (identified by the $\isOut$ function), which means updating the passed lvalues (i.e. $e_{i}$) with the final types of their corresponding parameters (i.e. $\tscope'_{f}(x_{i})$).

\begin{example}
	Assume that at line 44 of Program~\ref{prg:basic_congestion}, the \texttt{ttl} in the state type is mapped to $\tbvn{1,10}{{\low}}{8}$.
	Calling \texttt{decrease} entails creating a new local state type and copying in the arguments, which yields $\tscope_\texttt{decrease} = \{ x \mapsto \tbvn{1,10}{{\low}}{8} \}$.
	Typing the function's body (\texttt{x = x - 1}) results in the state type $\tscope'_\texttt{decrease} = \{ x \mapsto \tbvn{0,9}{{\low}}{8} \}$.
	The final $\Gamma''$ is produced by copying out arguments back to the initial state type which would map \texttt{ttl} to $\tbvn{0,9}{{\low}}{8}$.
\end{example}

In contrast to standard type systems, we directly type the body of the function, instead of typing functions separately and in isolation.
The main reason is that the intervals and labels of the types of actual arguments can be different
for each invocation of the function.
Notice that the nested analysis of the invoked function does not hinder termination of our analysis since P4 does not support recursion, eliminating the need to find a fix point for the types \cite{HuntS06}.

\begin{figure}[h]
	{
		\centering
		\fontsize{9}{13}\selectfont
		$\inferrule*[Lab=T-Table]
		{
			(\overline{e},\text{Cont}_{\mathrm{tbl}})= T(tbl) \\
			\tscope \vdash {{e}_i} : {\type}_i \\
			\lbl = \bigsqcup_{i} \getlabel(\type_i)  \\
			pc' = pc \sqcup \lbl \\
			\forall (\phi_j, (a_j, \overline{\tau}_j)) \in \text{Cont}_{\mathrm{tbl}}. \\
			\gamma_{j} = \refine{\tscope}{\phi_j} \\
			\tCall(T, a_{j}, pc', \overline \type_{j}, (\tscope_{g}, \tscope_{l}), \Gamma_{j})
		}
		{
			T, pc, \tscope \vdash \texttt{apply} \ tbl : \joinOnHigh{\cup_{j} \Gamma_{j}}{\lbl}
		}
		$

	}
  \end{figure}

  \parheading{\textsc{T-Table}} rule is similar to \textsc{T-Cond} and \textsc{T-Call}.
It relies on user-specified contracts to type the tables. A contract, as introduced in Section \ref{sec:security_types}, has the form $(\overline{e},\text{Cont}_{tbl})$, where $\text{Cont}_{tbl}$ consists of a set of triples $(\phi, (a, \overline{\type}))$.
Each triple specifies a condition $\phi$, under which an action $a$ is executed with arguments of specific types $\overline{\type}$.
A new context $pc'$ is produced by the initial $pc$ with the least upper bound of the labels of the keys.
 
For each triple $(\phi_j, (a_j, \overline{\tau}_j))$, \textsc{T-Table} relies on $\tCall$ to type the action $a_j$'s body under $pc'$, similar to \textsc{T-Call},
and accumulates the resulting state types into a set (i.e. $\cup_{j} \Gamma_{j}$). Finally, \textsc{T-Table} uses $\joinOnHigh{\cup_{j} \Gamma_{j}}{\lbl}$ to join their labels if $\ell$ was {\high}.

\begin{example}\label{example:table}

  Given the table contract depicted in Fig. \ref{fig:table_contract_example}, assume a state type $\tscope$ where $pc$ is $\low$ and \texttt{hdr.ipv4.dstAddr} is typed as $\tbvn{192}{{\low}}{8} {\cdot} \tbvn{168}{{\low}}{8} {\cdot} \tbvn{*}{{\low}}{16}$.
  According to \textsc{T-Table}, refining $\tscope$ produces three state types, out of which only one is not empty: $\refine{\tscope}{\texttt{dstAddr}[31:24] = 192}$. This refined state is used to type the action \texttt{ipv4\_forward} with arguments $[\tbvn{*}{{\high}}{48},\tbvn{1,9}{{\low}}{9}]$. The two empty states should be used to type the actions \texttt{ipv4\_forward} (with arguments $[\tbvn{*}{{\low}}{48},\tbvn{10,20}{{\low}}{9}]$) and \texttt{drop}. However, these states can be pruned by \textsc{T-EmptyType}, since $pc'$ is $\low$.

\end{example}

\begin{figure}[h]
	{
		\centering
		\fontsize{9}{13}\selectfont
		$\inferrule*[Lab=T-Extern]
		{
			(\tscope_g, \tscope_l) \vdash  e_{i} : \type_{i} \\
			(\text{Cont}_{\mathrm{E}},(x_{1},d_{1}) \cons (x_{n}, d_{n})) = T(f) \\\\
			\tscope_f =\{x_{i} \mapsto \type_{i}\} \\
			\forall (\tscope_i, \phi,\tscope_t) \in \text{Cont}_{\mathrm{E}}. \ (\tscope_g, \tscope_f) \sqsubseteq \tscope_i \\\\
			{\begin{aligned}\Gamma' = \{ \tscope' \concat \raisel(\tscope_t, pc) \mid \  & (\tscope_i, \phi,\tscope_t) \in \text{Cont}_{\mathrm{E}}
					\\ &\band \refine{(\tscope_{g}, \tscope_f)}{\phi} = \tscope' \neq \bullet  \}\end{aligned}} \\\\
			\Gamma'' =
			\{(\tscope'_{g}, \tscope_{l}) [e_{i} \mapsto \tscope'_{f}(x_{i}) \mid \isOut(d_{i})]
			\mid (\tscope'_{g}, \tscope'_{f}) \in \Gamma'\}
		}
		{
			T, pc, (\tscope_g, \tscope_l) \vdash f(e_{1}\cons e_{n}) : \Gamma''
		}
		$
		
	}
\end{figure}

\parheading{\textsc{T-Extern}} types the invocation of external functions.
It is similar to \textsc{T-Call} with the main difference that the semantics of external functions are not defined in P4,
therefore, we rely on user-specified contracts to approximate their behavior.
An extern contract is a set of tuples $(\tscope_i, \phi, \tscope_t)$, where $\tscope_i$ is the input state type,
$\phi$ is a boolean expression defined on the parameters of the extern,
and $\tscope_{t}$ indicates the state type components updated by the extern function (i.e., its side effects).

$\tscope_i$ denotes a contract-defined state type that must be satisfied prior to the invocation of the extern,
and the rule \textsc{T-Extern} ensure that the initial state $(\tscope_g, \tscope_f)$ is at most as restrictive as $\tscope_i$.
This approach is standard in type systems where functions are type-checked in isolation using predefined pre- and post-typing environments.
For each $(\tscope_i, \phi,\tscope_t)$ tuple in the contract, \textsc{T-Extern} refines the initial state type $(\tscope_g, \tscope_l)$ by $\phi$ yielding $\tscope'$, and filters out all $\tscope'$s that do not satisfy $\phi$ (i.e., the refinement $\refine{(\tscope_{g}, \tscope_f)}{\phi}$ is $\bullet$).
This is sound because we assume for all the variables appeared in $\phi$, the least upper bound of their labels within $\tscope_i$ is less restrictive than the lower bound of $\tscope_t$.
We raise the label of all elements in the $\tscope_{t}$ to $pc$ to capture indirect flows arising from updating the state type $\tscope'$ in a {\high} context, and then use use $\concat$ operation to update $\tscope'$ with the types in $\tscope_{t}$.
The final state type set $\Gamma'$ is produced by copying out the \texttt{out} and \texttt{inout} parameters from $\tscope'$.

\begin{example}
	In Program \ref{prg:basic_congestion}, let the contract for \texttt{mark\_to\_drop} at line 38 be defined as:
	\begin{align*}
		(\{ \texttt{egress\_spec} \mapsto \tbvn{*}{{\low}}{9} \}, \mathrm{true}, \{ \texttt{egress\_spec} \mapsto \tbvn{0}{{\low}}{9} \})
	\end{align*}
	which indicates that given an input state type $\{\texttt{egress\_spec} \mapsto \tbvn{*}{{\low}}{9} \}$
	the extern always sets the value of \texttt{egress\_spec} to zero.
	Assuming an initial state type $\tscope= \{ \texttt{egress\_spec} \mapsto \tbvn{7}{{\low}}{9} \}$. 
	Since the condition of the contract is $\mathrm{true}$ the refinement in this state type does not modify $\tscope$.
	This state type will be updated with the contract's $\tscope_t$ to become $\{ \texttt{egress\_spec} \mapsto \tbvn{0}{{\low}}{9} \}$
	if $pc$ is ${\low}$, otherwise $\{ \texttt{egress\_spec} \mapsto \tbvn{0}{{\high}}{9} \}$.
\end{example}

To guarantee the abstraction soundness of externs, for any input state $\mem$ to the externs semantics $\mem' = sem_f(\mem)$ and the contracts set $(\tscope_i,\phi,\tscope_t)$ must satisfy the following properties:
\begin{enumerate}
	\item Every input state $\mem$ must satisfy some condition in the contract set $\phi$, i.e., $\exists \phi \ . \phi(\mem)$
	\item All modified variables in output state $\mem'$ must be in the domain of $\tscope_t$,
	and their abstraction types in $\tscope_t$ must hold,
	i.e. $\{x. \ \mem(x) \neq \mem'(x) \} \subseteq \text{domain}(\tscope_t)$, and for all $x \in \text{domain}(\tscope_t)$ holds $\mem'(x):\tscope_t(x)$.
\end{enumerate}

Additionally, to guarantee the labeling soundness of externs, the contracts must satisfy the following properties:
\begin{enumerate}
	\item Conditions must preserve secrecy with respect to the output state type.
	For all variable names $\{x_1 \cons x_n\}$ appearing in the contract's condition $\phi$, holds ${\getlabel(\tscope_i(x_1)) \sqcup ... \sqcup \getlabel(\tscope_i(x_n)) \sqsubseteq lb(\tscope_t)}$.
	
	\item Extern semantics must preserve low-equivalence.
	Given any states $\mem_1$ and $\mem_2$,
	If $\phi(\mem_1)$ and $\mem_1\lequiv{\tscope_i}\mem_2$, then
	$\mem_1' = sem_f(\mem_1)$, $\mem_2' = sem_f(\mem_2)$, then the difference between the two output states must be also low equivalent
	$(\mem_1' \setminus \mem_1)\lequiv{\tscope_t}(\mem_2' \setminus  \mem_2)$.
\end{enumerate}

\subsection{Soundness}\label{sec:soundness}

Given initial state types $\tscope_{1}$ and $\tscope_{2}$,
and initial states $\mem_{1}$ and $\mem_{2}$, we write ${\mem_{1}} \consistent{\tscope_1}{\tscope_2} {\mem_{2}}$
to indicate that $\tscope_{1} \vdash \mem_{1}$, $\tscope_{2} \vdash \mem_{2}$, and ${\mem_{1} \lequiv{\tscope_{1} \sqcup \tscope_{2}} \mem_{2}}$.

The type system guarantees that a well-typed program terminates, and the final result is well-typed wrt. at least one of the resulting state types.

\begin{lemma}[Soundness of abstraction and labeling]\label{lemma:Soundness_labeling}
	Given initial state types $\tscope_{1}$ and $\tscope_{2}$,
	and initial states $\mem_{1}$ and $\mem_{2}$,
	such that $T, pc, \tscope_{1} \vdash s : \Gamma_{1}$ and
	$T, pc, \tscope_{2} \vdash s : \Gamma_{2}$, and
	$E_{1} \lequiv{T} E_{2}$, and $\mem_{1} \consistent{\tscope_{1}}{\tscope_{2}} \mem_{2}$
	then
	there exists $\mem'_{1}$ and $\mem'_{2}$ such that $E_{1}:\sem{\mem_{1}}{s}{\mem'_{1}}$,
	$E_{2}:\sem{\mem_{2}}{s}{\mem'_{2}}$,
    $\tscope'_{1} \in \Gamma_{1}$,
	$\tscope'_{2} \in \Gamma_{2}$,
	and $\mem'_{1} \consistent{\tscope'_{1}}{\tscope'_{2}} \mem'_{2}$.
\end{lemma}
Lemma~\ref{lemma:Soundness_labeling} states that starting from two indistinguishable states wrt. $\tscope_{1} \sqcup \tscope_{2}$, a well-typed program results in two indistinguishable states wrt. \emph{some} final state types in $\Gamma_{1}$ and $\Gamma_{2}$ that can also type the resulting states $\mem'_{1}$ and $\mem'_{2}$.

We rely on Theorem~\ref{thm:noninterference} to establish noninterference, that is, if every two states $\mem_{1}$ and $\mem_{2}$ that are indistinguishable wrt. \emph{any} two final state types are also indistinguishable by the output policy, then the program is noninterfering:
\begin{theorem}[Noninterference]\label{thm:noninterference}
	Given input policy case $\gamma_{i}$ and output policy $\Gamma_{o}$,
	if $T, pc, \tscope_{i} \vdash s : \Gamma$ and for every $\tscope_{a} , \tscope_{b} \in \Gamma$,
	such that $\mem_{1} \consistent{\tscope_{a}}{\tscope_{b}} \mem_{2}$
        it holds also that
	$\mem_1 \consistent{\tscope_{o}}{\tscope_{o}} \mem_2$ for all $\gamma_{o} \in \Gamma_{o}$, then $s$ is noninterfering wrt.
	the input policy case $\tscope_{i}$ and the output policy $\Gamma_{o}$.
\end{theorem}

Theorem~\ref{thm:noninterference} is required to be proved for \emph{every} possible pair of states.
To make the verification process feasible, we rely on the following lemma to show that this condition
can be verified by simply verifying a relation between the final state types ($\Gamma$)
and the output policy ($\Gamma_{o}$):%

  \begin{restatable}[Sufficient Condition]{lemma}{SufficientCondition}\label{lemma:sufficient_condition}
    Assume for every $\tscope_{1} , \tscope_{2} \in \Gamma$ and every $\gamma_{o} \in \Gamma_{o}$ such that $\tscope_{1} \cap \tscope_{o} \neq \bullet$ that
	\begin{enumerate}
		\item[(1)] $\tscope_{2} \cap \tscope_{o} \neq \bullet$ implies $\tscope_{1} \sqcup \tscope_{2} \sqsubseteq \tscope_{o}$, and
		\item[(2)] for every {$\lval$} either $\tscope_{2}(\lval) \subseteq \tscope_{o}(\lval)$ or $\tscope_{1} \sqcup \tscope_{2}(\lval) = {\low}$\ .
	\end{enumerate}
Then for every $\tscope_{1} , \tscope_{2} \in \Gamma$ such that $\mem_{1} \consistent{\tscope_{1}}{\tscope_{2}} \mem_{2}$,
	and every $\tscope_{o} \in \Gamma_{o}$ such that $\tscope_{o} \vdash \mem_{1}$ also  
	$\tscope_{o} \vdash \mem_{2}$ and moreover $\mem_1 \lequiv{\tscope_{o}} \mem_2$.
\end{restatable}

In the statement of Lemma~\ref{lemma:sufficient_condition} we use
$\tscope_{2}(\lval) \subseteq \tscope_{o}(\lval)$ to indicate that the interval of {$\lval$} in $\tscope_{2}$ is included in the interval specified in $\tscope_{o}$.

Intuitively, Lemma~\ref{lemma:sufficient_condition} formalizes that the least upper bound of any pair in the set of final state types ($\Gamma$)
should not be more restrictive than the output policy
(e.g. if {\high} information has flown to a variable, that variable should also be {\high} in the output policy cases)
\emph{and} the abstractions specified in the output policy cases (i.e. the intervals) are either always satisfied or do not depend on {\high} variables.

\subsection{Revisiting the basic congestion program}
\label{sec:type:example}
We revisit Program \ref{prg:basic_congestion} to illustrate
some key aspects of our typing rules.
Here, we only consider the first policy case of the input policy (\ref{par:policy}) of Section \ref{sec:background}, where the input packet is IPv4 and it is coming from the internal network.

For the initial state derived from this input policy case,
since (1) the parser's transitions depend on {\low}
variables, (2) the type system does not merge state types, and (3) the type system prunes the unreachable transition to \verb|accept| from \verb|parse_ethernet|,
then the parser terminates in a single state type where both \texttt{hdr.eth} and \texttt{hdr.ipv4} are valid, and their respective headers include the slices, intervals, and labels defined by the initial state type.

After the parsing stage is finished, the program's control flow reaches the \texttt{MyCtrl} control block. Since \texttt{hdr.ipv4} is valid and $pc$ is {\low}, pruning empty states allows us to ignore the \texttt{else} branch on line 62. Afterwards, the nested if statement at line 54 entails two possible scenarios. 
First scenario, when the destination address is in range \texttt{192.168.*.*}, as described in Example~\ref{example:cond}, the two state types resulting from
the \verb|if| at line 56 have \texttt{hdr.ipv4.ecn} set to {\high}.
As in Example~\ref{example:table}, these state types satisfy only the first condition
of the table's contract, which results in
assigning the type $\tbvn{1,9}{{\low}}{9}$ to \texttt{egress\_spec} and producing the state types
$\gamma^1_{int}$
and
$\gamma^2_{int}$.

Second scenario, when the destination address on line 54 does not match  \texttt{192.168.*.*}, the state type is refined for the \texttt{else} branch, producing one state type under condition $\texttt{ipv4.dstAddr} \geq \texttt{192.169.0.0}$ and one under $\texttt{ipv4.dstAddr} < \texttt{192.168.0.0}$.
For both of these state types, \texttt{hdr.ipv4.ecn} is set to $\tbvn{0}{{\low}}{2}$ by assignment on line 59.
Since in this case all branch conditions were {\low}, there is no {\high} field left in the headers.
The first of these two refined state types only satisfies the first condition of the table contract, resulting in one single (after pruning empty states) final state type, $\gamma^3_{int}$, where the packet has been forwarded to the internal network and \texttt{egress\_spec} is set to $\tbvn{1,9}{{\low}}{9}$.
The second refined state type however satisfies all the conditions of the table contract, resulting in three final state types $\gamma^4_{int}$, $\gamma^1_{ext}$, $\gamma^1_{drop}$ with \texttt{egress\_spec} being set to $\tbvn{1,9}{{\low}}{9}$, $\tbvn{10,20}{{\low}}{9}$, and $\tbvn{0}{{\low}}{9}$, respectively.
Notice that among these states, only $\gamma^3_{int}$ and $\gamma^4_{int}$ contains a {\high} fields (i.e. $\texttt{ipv4.dstAddr}$) due to the first argument returned by the table being $\tbvn{*}{{\high}}{48}$.

We finally check the sufficient condition for the output policy (\ref{par:policy_out}) and its only output policy case $\gamma_{o}$, which states that when \texttt{egress\_spec} is $\tbvn{10,20}{{\low}}{9}$ (i.e. the packet leaves the internal network) all header fields are {\low}.
Only state type $\gamma^1_{ext}$ matches the output policy case (i.e. $\cap \gamma_{o} \not= \bullet$), and this state type satisfies $\gamma^1_{ext} \sqcup \gamma^1_{ext} \sqsubseteq \gamma_{o}$ since all header fields and \texttt{egress\_spec} are {\low} in $\gamma^1_{ext}$.
All other state types (i.e. $\gamma^1_{int}$, $\gamma^2_{int}$, $\gamma^3_{int}$, $\gamma^4_{int}$, and $\gamma^1_{drop}$)
do not match the output policy condition (i.e. $\cap \gamma_{o} = \bullet$), since they do not correspond to packets sent to the external network 
(i.e. their $\texttt{egress\_spec}$ is not in range $\tbvn{10,20}{}{}$).
Therefore, we conclude that for this specific input policy case, Program \ref{prg:basic_congestion} is non-interfering wrt. the output policy case $\gamma_{o}$.

Our analysis can also detect bugs. Assume a bug on line 54 of Program \ref{prg:basic_congestion}.
To illustrate this, assume that the program is buggy and instead of checking the 8 most significant bits (i.e. [31:24]) of the \texttt{hdr.ipv4.dstAddr}, it checks the least significant bits (i.e. [7:0]).
This means that IPv4 packets with destination address is in range \texttt{*.168.*.192} would satisfy the condition of the \texttt{if} statement on line 54.
Similar to the non-buggy program, the \texttt{if} at line 56 would produce two state types with \texttt{hdr.ipv4.ecn} set to {\high}.
These state types satisfy all the conditions of the table contract.
For presentation purposes, let us focus on only one of these state types. 
Applying the table on line 61 would produce three final state types $\gamma^1_{int}$, $\gamma^1_{ext}$, $\gamma^1_{drop}$ with \texttt{egress\_spec} being set to $\tbvn{1,9}{{\low}}{9}$, $\tbvn{10,20}{{\low}}{9}$, and $\tbvn{0}{{\low}}{9}$, respectively.
Note that in all these final state types, \texttt{hdr.ipv4.ecn} is {\high}.
When checking the sufficient condition, state type $\gamma^1_{ext}$ matches the output policy case (i.e. $\cap \gamma_{o} \not= \bullet$) but it does not satisfy $\gamma^1_{ext} \sqcup \gamma^1_{ext} \sqsubseteq \gamma_{o}$, because the \texttt{hdr.ipv4.ecn} field {\high} in $\gamma^1_{ext}$ and {\low} $\gamma_{o}$.
Hence, this buggy program will be marked as interfering, highlighting the fact that some of the packets destined for the external network contain congestion information and unintentionally leak sensitive information.

The benefit of value- and path-sensitivity of our approach can also be demonstrated here.
For all other input policy cases that describe non-IPv4 packets, the \texttt{parse\_ipv4} state is not going to be visited.
A path-insensitive analysis, which merges the results of the parser transitions, would lose the information about the validity of the \texttt{hdr.ip4} header. 
This would then lead to the rejection of the program as insecure because an execution where the \texttt{parse\_ipv4} state has not been visited, yet the \texttt{if} branch on line 53 has been taken, will be considered feasible.

Our analysis, on the other hand, identifies that any execution that has not visited \texttt{parse\_ipv4} results in an invalid \texttt{hdr.ip4} header. Consequently, for all such executions, it produces a final state type where the packet is dropped, and \texttt{egress\_spec} is set to $\tbvn{0}{{\low}}{9}$.
This state type satisfies the sufficient condition, since \texttt{egress\_spec} does not intersect $\gamma_{o}(\texttt{egress\_spec})$ and is {\low}.

\section{Implementation and Evaluation}\label{sec:implementation}

To evaluate our approach we developed {\toolname}\cite{tap4sTool}, a prototype tool which implements the security type system of Section \ref{sec:type_system}. 
{\toolname} is developed in Python and uses the lark parser library \cite{larkTool} to parse P4 programs. 

{\toolname} takes as input a P4 program, an input policy, and an output policy. Initially, it parses the P4 program, generates an AST, and relies on this AST and the input policy $\tscope_{i}$ to determine the initial type of input packet fields and the standard metadata. Because the input policy is data-dependent, the result of this step can generate multiple state types ($\tscope_1, \hdots, \tscope_m$), one state type for each input interval.
{\toolname} uses each of these state types as input for implementing the type inference on the program. During this process {\toolname} occasionally interacts with a user-defined contract file to retrieve the contracts of the tables and externs. Finally, {\toolname} yields a set of final state types ($\tscope'_1, \hdots, \tscope'_n$) which are checked against an output policy, following the condition in Lemma \ref{lemma:sufficient_condition}. If this check is successful the program is deemed secure wrt. the output policy, otherwise the program is rejected as insecure.

\tightpar{Test suite} To validate our implementation we rely on a functional test suite of {\tests} programs. These programs are P4 code snippets designed to validate the support for specific functionalities of our implementation, such as extern calls, refinement, and table application. %

\tightpar{Use cases}
We evaluate {\toolname} on {\usecases} use cases, representing different real-world scenarios. The results of this evaluation are summarized in Table \ref{tab:results}. Due to space constraints, detailed descriptions of these use cases are provided in {\iffull Appendix \ref{sec:use_cases}. \else the full version of the paper \cite{fullreport}. \fi} 
We also implement and evaluate the use cases from P4BID \cite{grewal2022p4bid}. 
These use cases are described in {\iffull Appendix \ref{sec:p4bid_usecases}, \else the full version of the paper  \cite{fullreport},\fi} and their corresponding evaluation results are included in Table \ref{tab:results}.
They serve as a baseline for comparing the feasibility of {\toolname} with P4BID. On average, P4BID takes $30$ ms to analyze these programs, whereas {\toolname} takes $246$ ms. Despite the increased time, this demonstrates that {\toolname} performs the analysis with an acceptable overhead.
On the other hand, due to the data-dependent nature of our use cases and their reliance on P4-specific features such as slicing and externs, P4BID cannot reliably check these scenarios, leading to their outright rejection in all cases.

\begin{table}[t]
	\centering
	\small	
	\caption{Evaluation results}
	\label{tab:results}
	\rowcolors{4}{}{gray!10}
	\begin{tabular}{l L{24pt} L{24pt} L{24pt} M{39pt}}
		& \multicolumn{3}{c}{\textbf{Time (ms)}} & \\
		\cmidrule(ll){2-4}
		& \footnotesize\textbf{Total} & \footnotesize\textbf{Typing}  & \footnotesize\textbf{Security Check} & \footnotesize\textbf{Number of Final $\gamma$s} \\   
		\midrule[1px]
		\footnotesize Basic Congestion & \footnotesize$5930$ & \footnotesize$966$ & \footnotesize$4794$ & \footnotesize$97$ \\ 
		\footnotesize Basic Tunneling & \footnotesize$610$ & \footnotesize$157$ & \footnotesize$290$ & \footnotesize$15$ \\ 
		\footnotesize Multicast & \footnotesize$199$ & \footnotesize$16$ & \footnotesize$23$ & \footnotesize$6$ \\
		\footnotesize Firewall & \footnotesize$4560$ & \footnotesize$1015$ & \footnotesize$3378$ & \footnotesize$44$ \\
		\footnotesize MRI & \footnotesize$7646$ & \footnotesize$523$ & \footnotesize$6957$ & \footnotesize$23$ \\
		\midrule[0.5px]
		\footnotesize Data-plane Routing & \footnotesize$274$ & \footnotesize$109$ & \footnotesize$8$ & \footnotesize$12$ \\
		\footnotesize In-Network Caching & \footnotesize$261$ & \footnotesize$91$ & \footnotesize$14$ & \footnotesize$6$ \\
		\footnotesize Resource Allocation & \footnotesize$256$ & \footnotesize$87$ & \footnotesize$10$ & \footnotesize$9$ \\
		\footnotesize Network Isolation - Alice & \footnotesize$243$ & \footnotesize$27$ & \footnotesize$62$ & \footnotesize$3$ \\
		\footnotesize Network Isolation - Top & \footnotesize$242$ & \footnotesize$23$ & \footnotesize$63$ & \footnotesize$3$ \\
		\footnotesize Topology & \footnotesize$202$ & \footnotesize$40$ & \footnotesize$4$ & \footnotesize$3$ \\   
	\end{tabular}
	\vspace*{-10pt}
\end{table}

\section{Related Work}\label{sec:related_work}

\tightpar{IFC for P4}
Our work draws inspiration from P4BID \cite{grewal2022p4bid}, which adapts and implements a security type system \cite{volpano} for P4, ensuring that well-typed programs satisfy noninterference. By contrast, we show that security policies are inherently data-dependent, thus motivating the need for combining security types with interval-based abstractions.
This is essential enforcing IFC in real-world P4 programs without code modifications, as demonstrated by our 5 use cases.
Moreover, our analysis handles P4 features such as slicing and externs, while supporting the different stages of the P4 pipeline, beyond a single control block of the match-action stage. %

\tightpar{IFC policy enforcement}
Initial attempts at enforcing data-dependent policies \cite{giffin2012hails,stefan2017flexible,yang2012language,GuarnieriBSBS19,parker2019lweb} used dynamic information flow control. %
The programmer declaratively specifies data-dependent policies and delegates the enforcement to a security-enhanced runtime, thus separating the  policy specification from the code implementation.

Our approach shares similarities with static enforcement of data-dependent IFC policies such as \emph{dependent information flow types} and \emph{refinement information flow types}.
Dependent information flow types \cite{lourencco2015dependent} rely on dependent type theory and propose a dependent security type system, in which the security level of a type may depend on its runtime value.
Eichholz et al. \cite{eichholz2022dependently} introduced a dependent type system for the P4 language, called $\Pi4$, which ensures properties such as preventing the forwarding of expired packets and invalid header accesses.
Value-dependent security labels \cite{lourencco2013information} partition the security levels by indexing their labels with values, resulting in partitions that classify data at a specific level, depending on the value.
Dependent information flow types provide a natural way to express data-centric policies where the security level of a data structure's field may depend on values stored in other fields.

Later approaches have focussed on trade-offs between automation and decidability of the analysis. Liquid types \cite{vazou2013abstract,vazou2014refinement} are an expressive yet decidable refinement type system \cite{jhala2021refinement} to statically express and enforce data-dependent information flow polices.
\textsc{Lifty} \cite{polikarpova2020liquid} provides tool support for specifying data-dependent policies and uses Haskell's liquid type-checker \cite{vazou2014refinement} to verify and repair the program against these policies. %
\textsc{Storm} \cite{lehmann2021storm} is a web framework that relies on liquid types to build MVC web applications that can statically enforce data-dependent policies on databases using liquid types. %

Our interval-based security types can be seen as instantiations of refinement types and dependent types. Our  simple interval analysis appears to precisely capture the key ingredients of P4 programs, while avoiding challenges with more expressive analysis.
By contrast, the compositionality of analysis based on refinement and dependent types can result in precision loss and is too restrictive for our intended purposes (as shown in Section~\ref{sec:type:example}), due to merging types of different execution paths.
We solve this challenges by proposing a global path-sensitive analysis that avoids merging abstract state in conditionals.
We show that our simple yet tractable abstraction is sufficient to enforce the data-dependent policies while precisely modeling P4-specific constructs such as slicing, extract, and emit.

Other works use abstract interpretation in combination with IFC. De Francesco and Martini \cite{de2007instruction} implement information-flow analysis for stack-based languages like Java. They analyze the instructions an intermediate language by using abstract interpretation to abstractly execute a program on a domain of security levels. Their method is flow-sensitive but not path-sensitive.
Cortesi and Halder \cite{cortesi2014information} study information leakage in databases interacting with Hibernate Query Language (HQL). Their method uses a symbolic domain of positive propositional formulae that encodes the variable dependencies of database attributes to check information leaks.
Amtoft and Banerjee formulate termination-insensitive information-flow analysis by combining abstract interpretation and Hoare logic%
 \cite{amtoft2004information}.
They also show how this logic can be extended to form a security type system that is used to encode noninterference. This work was later extended to handle object-oriented languages in~\cite{banerjeelogic}.

\tightpar{Analysis and verification of network properties} Existing works on network analysis and verification do not focus on information flow properties.
Symbolic execution is widely used for P4 program debugging, enabling tools to explore execution paths, find bugs, and generate test cases. 
Vera~\cite{stoenescu2018debugging} uses symbolic execution to explore all possible execution paths in a P4 program, using symbolic input packets and table entries.
Vera catches bugs such as accesses fields of invalid headers and checking that the \texttt{egress\_spec} is zero for dropped packets. Additionally, it allows users to specify policies, such as; ensuring that the NAT table translates packets before reaching the output ports, and the NAT drops all packets if its entries are empty.
Recently, Scaver \cite{yao2024scaver} uses symbolic execution to verify forwarding properties of P4 programs. To address the path explosion problem, they propose multiple pruning strategies to reduce the number of explored paths.
ASSERT-P4 \cite{freire2018uncovering} combines symbolic execution with assertion checking to find bugs in P4 programs, for example, that the packets with TTL value of zero are not dropped and catching invalid fields accesses.
Tools like P4Testgen~\cite{ruffy2023p4testgen} and p4pktgen~\cite{notzli2018p4pktgen} use symbolic execution to automatically generate test packets. 
This approach supports test-driven development and guarantees the correct handling of packets by synthesizing table entries for thorough testing of P4 programs.

Abstract interpretation has also been used to verify functional properties such as packet reachability and isolation. While these properties ensure that packets reach their intended destinations, they do not address the flow of information within the network.
Alpernas et al. \cite{alpernas2018abstract} introduce an abstract interpretation algorithm for networks with stateful middleboxes (such as firewalls and load balancers). Their method abstracts the order and cardinality of packet on channels, and the correlation between middleboxes states, allowing for efficient and sound analysis.
Beckett et al. \cite{beckett2019abstract} develop ShapeShifter, which uses abstract interpretation to abstract routing algebras to verify reachability in distributed network control-planes, including  objects such as path vectors and IP addresses and methods such as path lengths, regular expressions, intervals, and ternary abstractions.

\section{Conclusion}\label{sec:conclusion}
This paper introduced a novel type system that combines security types with interval analysis to ensure noninterference in P4 programs. Our approach effectively prevents information leakages and security violations by statically analyzing data-dependent flows in the data-plane. 
The type system is both expressive and precise, minimizing overapproximation while simplifying policy specification for developers.
Additionally, our type system successfully abstracts complex elements like match-action blocks, tables, and external functions, providing a robust framework for practical security verification in programmable networks.
Our implementation, {\toolname}, demonstrated the applicability of the security type system on real-world P4 use cases without losing precision due to overapproximations. 
Future research includes adding support for declassification, advanced functionalities such as cryptographic constructs, and extending the type system to account for side channels.

\section*{Acknowledgment}

Thanks are due to  the anonymous reviewers for their insightful comments and feedback.
This work was partially supported by the Wallenberg AI, Autonomous Systems and Software Program (WASP) funded by the Knut and Alice Wallenberg Foundation, the KTH Digital Futures research program, and the Swedish Research Council (VR).

\bibliographystyle{IEEEtran}
\bibliography{bibliography_cleaned}

\iffull
\appendices

\section{Use cases}\label{sec:use_cases}

\subsection{Basic Tunneling}
Our first use case, shown in Program~\ref{prg:basic_tunneling}, outlines procedures for handling standard IPv4 packets and encapsulated tunneling packets.
The parser \texttt{MyParser} starts by extracting the Ethernet header on line 6, 
and %
for \texttt{etherType} \hexnumber{1212} (tunneled packet), 
it transitions to \texttt{parse\_myTunnel} state (line 14), 
extracts the tunnel header, checks the \texttt{proto\_id} field, 
and transitions to \texttt{parse\_ipv4} state (line 21) if an IPv4 packet is indicated.
For \texttt{etherType} \hexnumber{0800} (IPv4 packet), 
it directly transitions to \texttt{parse\_ipv4} and extracts the IPv4 header.
Once the headers are parsed, the pipeline proceeds to the \texttt{MyCtrl} control block,
starting from the apply block on line 42 which contains two if statements:
If only the IPv4 header is valid (line 43), the \texttt{ipv4\_lpm} table is applied which forward or drop the packet based on a longest prefix match (lpm) on the destination IPv4 address.
If the tunnel header is valid (line 46), the \texttt{myTunnel\_exact} table forwards the packet based on an exact match of the \texttt{myTunnel} header's \texttt{dst\_id}, using the \texttt{myTunnel\_forward} action.

\begin{figure*}[ht]
	\minipage[b]{0.5\textwidth}
\begin{lstlisting}[
	language=P4,
	caption={},
	firstnumber=1]
struct headers {
	ethernet_t	eth; myTunnel_t	myTunnel; ipv4_t	ipv4;
}
parser MyParser(/* omitted */) {
	state start { transition parse_ethernet; }  
	state parse_ethernet {
		packet.extract(hdr.eth);
		transition select(hdr.eth.etherType) {
			0x1212: parse_myTunnel;
			0x0800: parse_ipv4;
			default: accept;
		}
	}
	state parse_myTunnel {
		packet.extract(hdr.myTunnel);
		transition select(hdr.myTunnel.proto_id) {
			0x0800: parse_ipv4;
			default: accept;
		}
	}
	state parse_ipv4 {
		packet.extract(hdr.ipv4);
		transition accept;
	}
}
\end{lstlisting}
	\endminipage\hfill
	\minipage[b]{0.5\textwidth}
\begin{lstlisting}[
	language=P4,
	caption={},
	firstnumber=last]
control MyCtrl(/* omitted */) {
	action drop() 
		{ /* drops the packet */ }
	action ipv4_forward(bit<48> dstAddr, bit<9> port)
		{ /* basic forward */ }
	table ipv4_lpm 
		{ /* omitted */}
	
	action myTunnel_forward(bit<19> port) {
		standard_metadata.egress_spec = port;
	}
	table myTunnel_exact {
		key = {hdr.myTunnel.dst_id: exact;}
		actions = { myTunnel_forward; drop;}
		default_action = drop();
	}
	apply {
		if (hdr.ipv4.isValid() && !hdr.myTunnel.isValid()) {
			ipv4_lpm.apply(); // Process non-tunneled packets
		}
		if (hdr.myTunnel.isValid()) { 
			myTunnel_exact.apply(); // Process tunneled packets
		}
	}
}
\end{lstlisting}
	\endminipage\hfill
	\captionof{lstlisting}{Basic tunneling}
	\label{prg:basic_tunneling}
\end{figure*}

Since the header fields of the tunneled packets are not modified while they are forwarded (Lines 34-36), 
to keep the source MAC address of the packets within the internal networks private,
the program should not forward tunneled packets to an external network.
The input policy in this use case indicates that if the input packet the packet is tunneled (i.e., its \texttt{etherType} is \hexnumber{1212}) then the packet's \texttt{srcAddr} is {\high}.
The output policy relies on the output port the packet is sent to, and ensures that
``if the \texttt{egress\_spec} is between $10$-$511$ then the packet has left the internal network, therefore all field of the packet's headers should be {\low}."

The general behavior of table \texttt{ipv4\_lpm} is reflected in its contract 
as it makes sure the packets with \texttt{ipv4} destination address \texttt{198.*.*.*} are forwarded to the ports connected to the internal network, 
while all the other destination addresses are forwarded to ports connected to the external network.

The contract of table \texttt{myTunnel\_exact} plays a crucial role in the security of Program~\ref{prg:basic_tunneling}. 
A \texttt{correct} behavior for this table only forwards the tunneled packets to ports connected to the internal network. 
The evaluation reported in Table~\ref{tab:results} is performed under a contract that reflected this behavior, 
which results in {\toolname} accepting the program as secure.
If this table is somehow misconfigured and forwards the tunneled packet to any port connected to the external network, 
{\toolname} can capture this and flag the program as insecure.

\subsection{Multicast}
Our next use case is Program~\ref{prg:multicast} which is capable of multicasting packets to a group of ports.
Upon receiving a packet, the switch looks up its destination MAC address \texttt{dstAddr}, if it is destined to any of the hosts connected to the switch, the packet is forwarded to its destination (line 11), otherwise the switch broadcasts the packet on ports belonging to a multicast group by setting the \texttt{standard\_metadata.mcast\_grp} to $1$ (line 9). Fig. \ref{fig:multicast} illustrates the network schema of this scenario.

\begin{figure}[ht]
\begin{lstlisting}[
	label=prg:multicast,
	language=P4, 
	caption={Multicast},
	firstnumber=1]
struct headers {
	ethernet_t   ethernet;
}
control MyCtrl(/* omitted */) {
	action drop() {
		mark_to_drop(standard_metadata);
	}
	action multicast() {
		standard_metadata.mcast_grp = 1;
	}
	action mac_forward(bit<9> port) {
		standard_metadata.egress_spec = port;
	}
	table mac_lookup {
		key = { hdr.ethernet.dstAddr : exact; }
		actions = { multicast; mac_forward; drop; }
	}
	apply {
		if (hdr.ethernet.isValid())
			mac_lookup.apply();
	}
}
\end{lstlisting}
\end{figure}
\begin{figure}[ht]
	\centering
	\begin{tikzpicture}[xscale=1, yscale=1]
		\node[inner sep=0] (h1) at (-2,1) {\includegraphics[width=6mm]{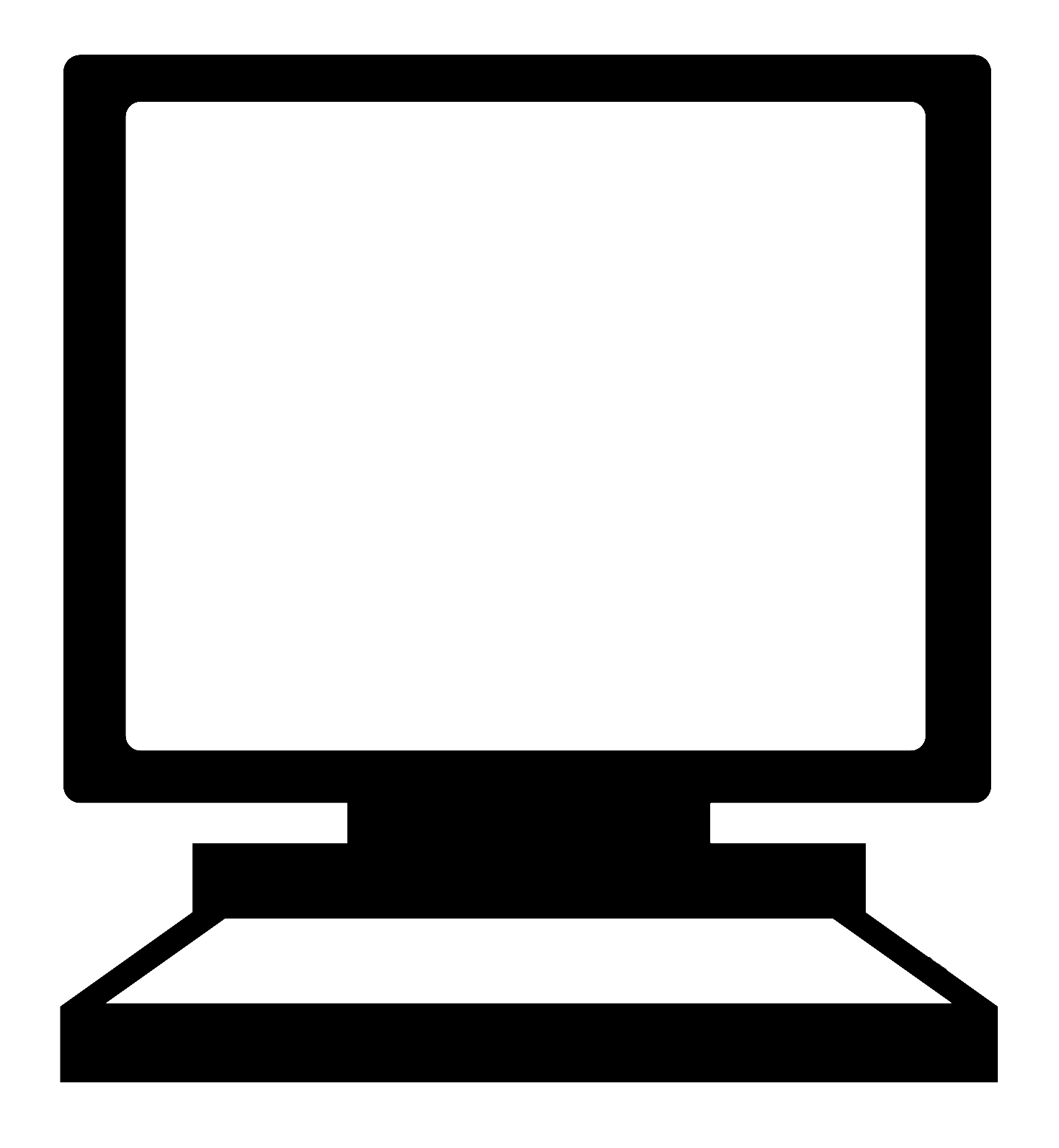}};
		\node[inner sep=0] (h2) at (-2,-1) {\includegraphics[width=6mm]{./figures/host_icon.png}};
		\node[inner sep=0] (h3) at (2,1) {\includegraphics[width=6mm]{./figures/host_icon.png}};
		\node[inner sep=0] (h4) at (2,-1) {\includegraphics[width=6mm]{./figures/host_icon.png}};
		
		\node[inner sep=3pt] (switch) at (0,0) {\includegraphics[width=10mm]{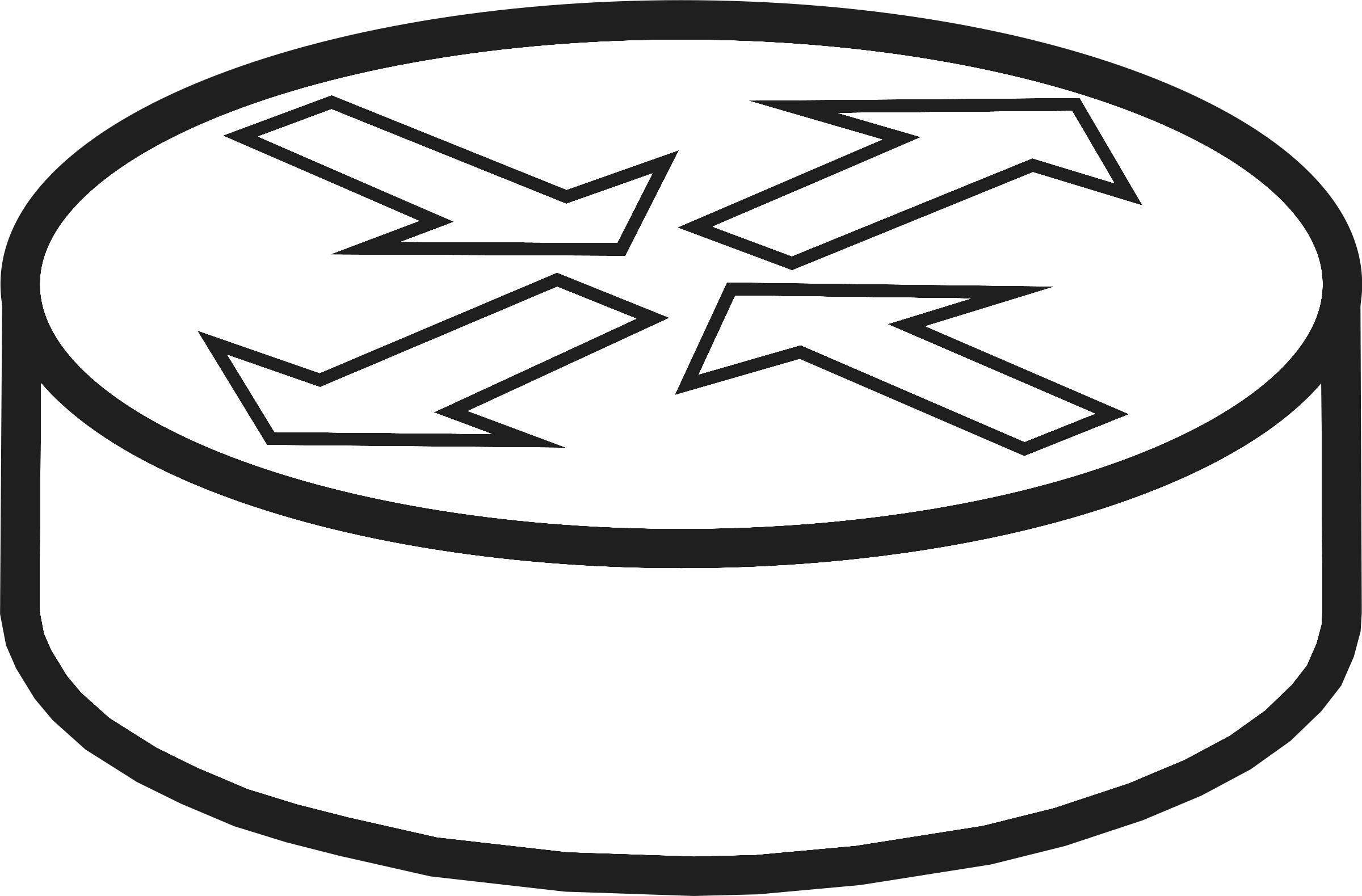}};

		\draw[<->] (h1) -- (switch) node[near end, above] {\scriptsize $1$};
		\draw[<->] (h2) -- (switch) node[near end, above] {\scriptsize $2$};
		\draw[<->] (h3) -- (switch) node[near end, above] {\scriptsize $3$};
		\draw[<->] (h4) -- (switch) node[near end, above] {\scriptsize $4$};
	\end{tikzpicture}
	\caption{Multicast schema}
	\label{fig:multicast}
\end{figure}

To implement this functionality, the program utilizes the table \texttt{mac\_lookup} which is populated by the control-plane, and contains the mac addresses and the port information needed to forward non-multicast packets.

While the broadcast packets are sent to all of the multicast ports, it is desirable to ensure the packets that are not supposed to be broadcast are indeed not broadcasted.
Our input security policy in this scenario sets the packets destined to any of the hosts connected to the switch as {\high}, while labeling the broadcast packets {\low}.
The contract of the table \texttt{mac\_lookup}'s needs to capture the essence of this use case, that is, packets with the \texttt{dstAddr} of any of the hosts need to be forwarded by invoking the \texttt{mac\_forward} action (line 11), and all the other packets need to be broadcast by invoking the \texttt{multicast} action (line 8).
To ensure the program behaves desirably, the output policy checks that all the packets send to the multicast ports (which have their \texttt{mcast\_grp} set to 1 according to line 9) are {\low}.

As illustrated in Table~\ref{tab:results}, under these policies and contracts, Program~\ref{prg:multicast} is secure. It results in $6$ final state types ($\tscope$), and takes approximately $220$ milliseconds to verify the program is policy compliance.

\subsection{Firewall}
This use case models a scenario where the switch is running the firewall Program~\ref{prg:firewall} which allows it to monitor the connections between an internal and an external network. The network schema of this scenario is presented in Fig. \ref{fig:firewall}. 

\begin{figure}[ht]
\begin{lstlisting}[
	label=prg:firewall,
	language=P4, 
	caption={Firewall},
	firstnumber=1]
header tcp_t{
	bit<16> srcPort;
	// omitted
}
struct headers {
	ethernet_t	ethernet; ipv4_t	ipv4; tcp_t	tcp;
}
control MyCtrl(/* omitted */) {
	action drop() 
		{ /* drops the packet */ }
	action ipv4_forward(bit<48> dstAddr, bit<9> port) 
		{ /* basic forward */ }
	table ipv4_lpm 
		{ /* omitted */}

	action set_direction(bit<1> dir) {
		direction = dir;
	}
	table check_ports {
		key = { standard_metadata.ingress_port: exact; }
		actions = { set_direction; NoAction; }
	}

	apply {
		if (hdr.ipv4.isValid()) {
			ipv4_lpm.apply();
			if (hdr.tcp.isValid()) {
				check_ports.apply();
				if (direction == 1) {// Packet is from outside
					// only allow ssh conncetions from outside
					if (hdr.tcp.srcPort != 22) { drop(); }
				}
			}
		}
	}
}
\end{lstlisting}
\end{figure}
\begin{figure}[ht]
	\centering
	\begin{tikzpicture}[xscale=1, yscale=1]
		\node[inner sep=0] (h1) at (-2,1) {\includegraphics[width=6mm]{./figures/host_icon.png}};
		\node[inner sep=0] (h2) at (-2,0) {\includegraphics[width=6mm]{./figures/host_icon.png}};
		\node[inner sep=0] (h3) at (-2,-1) {\includegraphics[width=6mm]{./figures/host_icon.png}};
		
		\node[inner sep=3pt] (switch) at (0.3,0) {\includegraphics[width=10mm]{./figures/router_icon.png}};
		
		\node[inner sep=3pt] (internet) at (3,0) {\includegraphics[width=27mm]{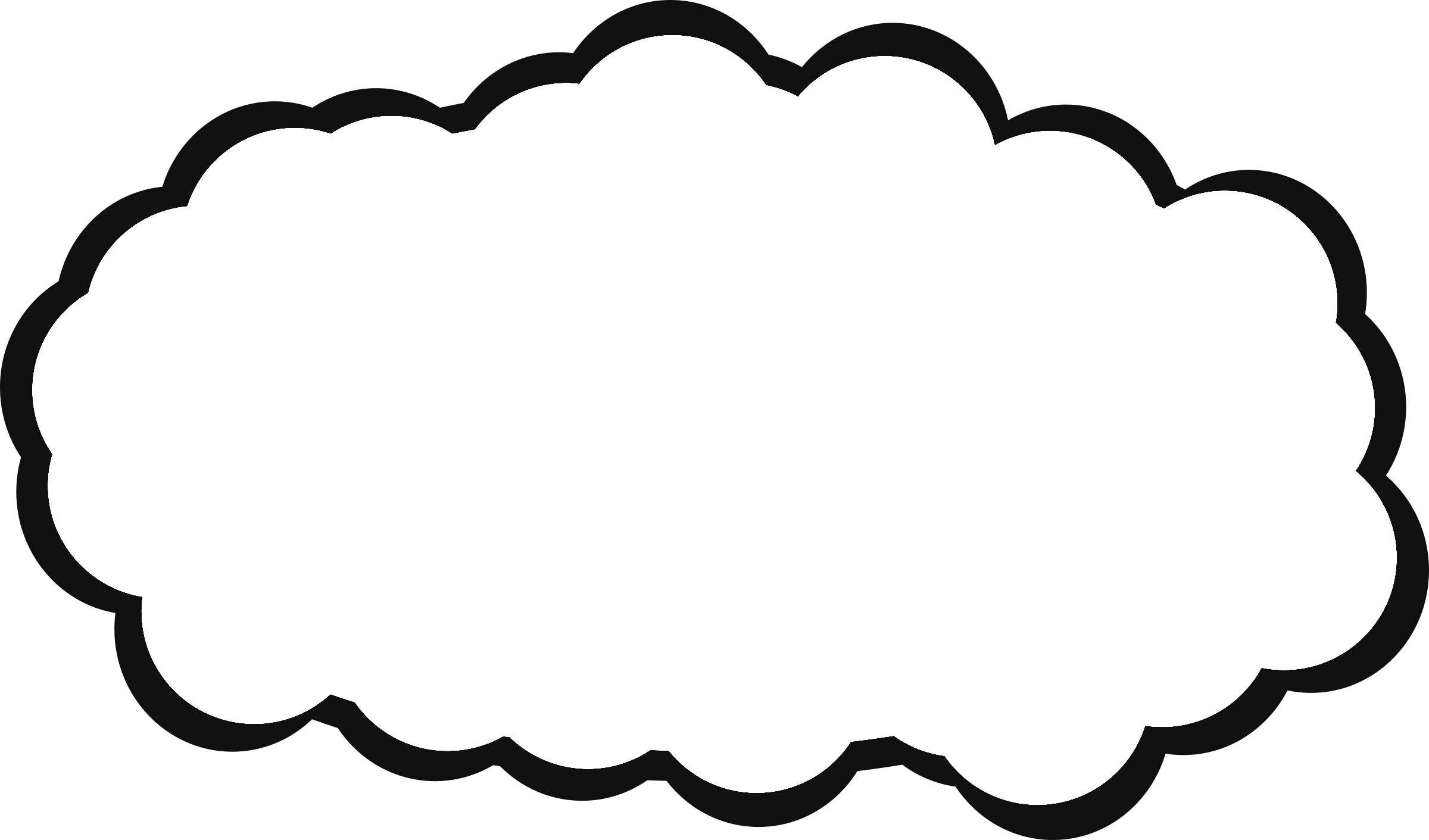}};
		\node[inner sep=0] (internetText) at (3,0) {\scriptsize External Network};
		
		\node (outPol) at (-2,1.8) {\tiny Internal Network};
		\node[draw, dashed, inner sep=6pt, rounded corners, fit=(h1) (h3)] (internalNet) {};
		
		\draw[<->] (h1) -- (switch) node[near end, above] {\scriptsize $1$};
		\draw[<->] (h2) -- (switch) node[near end, above] {\scriptsize $2$};
		\draw[<->] (h3) -- (switch) node[near end, above] {\scriptsize $3$};
		\draw[<->] (switch) -- (internet) node[near start, above] {\scriptsize $4$};
	\end{tikzpicture}
	\caption{Firewall schema}
	\label{fig:firewall}
\end{figure}

After parsing an input packet, the switch applies the \texttt{ipv4\_lpm} table (line 26), which based on the packet's IPv4 destination address forwards or drops the packet. Next, it applies the \texttt{check\_ports} table, which based on the input port number (\texttt{ingress\_port}) identifies whether the packet is coming from the external or the internal network. As depicted in Fig. \ref{fig:firewall} port $4$ is connected to the external network and ports $1-3$ are connected to the hosts of the internal network, therefor if the standard metadata's \texttt{ingress\_port} was $4$, the \texttt{check\_ports} table sets the direction to $1$ which indicates the packet is coming from the external network.

The policy of the firewall is that the hosts in the internal network are allowed to communicate with the outside networks, but the hosts in the external network are only allowed to \texttt{ssh} to the internal hosts. To this end, for all the packets with direction $1$, the program will drop all the packets whose tcp port (\texttt{hdr.tcp.srcPort}) is not $22$ (line 31).

To enforce such policy, we rely on integrity labels (instead of confidentiality) to designate which packets are allowed (trusted), and which packets are not.
The input security policy in this scenario sets the packets coming the internal network, identified by their \texttt{ingress\_port} as trusted ({\low}), while any packet coming from the external network (with \texttt{ingress\_port} $4$) is untrusted ({\high}) except when its TCP source port \texttt{srcPort} is $22$.

The contract of the \texttt{ipv4\_lpm} captures the behavior of this table by making sure the packets with \texttt{ipv4} destination address \texttt{198.*.*.*} are forwarded to the ports connected to the internal network (ports $1$ to $3$), and all the other destination addresses are forwarded to port $4$.
The contract of table \texttt{check\_ports} updates the direction by checking the \texttt{ingress\_port} of the incoming packet, setting the direction to $1$ if the \texttt{ingress\_port} was $4$.

The output policy checks that all the packets leaving the switch are trusted and {\low}.
Table~\ref{tab:results} depicts the results of {\toolname} for this use case. Under these policies and contracts Program~\ref{prg:multicast} is deemed secure. {\toolname} produces $44$ final state types, and takes approximately $6$ seconds to verify the security of the program.

\subsection{Multi-Hop Route Inspection}
Program~\ref{prg:mri} implements a simplified version of In-Band Network Telemetry, called Multi-Hop Route Inspection (MRI).
The purpose of MRI is to let the users to track the path and the length of queues that every packet travels through. To do this, the P4 program adds an ID and queue length to the header stack of every packet (line 10). Upon reaching the destination, the sequence of switch IDs shows the path the packet took, and each ID is followed by the queue length at that switch.

\begin{figure}[ht]
\begin{lstlisting}[
	label=prg:mri,
	language=P4, 
	caption={Multi-Hop Route Inspection},
	firstnumber=1]
header switch_t {
	switchID_t	swid;
	qdepth_t	qdepth;
}
struct headers {
	ethernet_t	ethernet;
	ipv4_t	ipv4;
	ipv4_option_t	ipv4_option;
	mri_t	mri;
	switch_t[9]	swtraces;
}
control MyIngress(/* omitted */) {
	action drop() 
		{ /* drops the packet */ }
	action ipv4_forward(bit<48> dstAddr, bit<9> port)
		{ /* basic forward */ }
	table ipv4_lpm 
		{ /* omitted */}
	apply {
		if (hdr.ipv4.isValid()) { ipv4_lpm.apply(); }
	}
}
control MyEgress(/* omitted */) {
	action add_swtrace(switchID_t swid) {
		// updates the swtraces header
	}
	table swtrace {
		key = { standard_metadata.egress_spec: exact; }
		actions = { add_swtrace; NoAction;}
	}
	apply {
		if (hdr.mri.isValid()) {
			swtrace.apply();
		}
	}
}
\end{lstlisting}
\end{figure}

After parsing the packet, the program applies the \texttt{ipv4\_lpm} table to forward the packet based on its IPv4 destination address. Afterwards, in the \emph{MyEgress} control block, the \texttt{swtrace} table (line 27), based on the port information specified in the \texttt{egress\_spec}, decides whether to add the queue length data to \texttt{swtraces} header or not. 

While it makes sense to add this information for packets that are traveling within a local network, similar to the basic congestion example (Program \ref{prg:basic_congestion}) the id of the switches and their queue length can give an external adversary information about the state of the local network. Therefore it is desirable to protect the local network by making sure that Program~\ref{prg:mri} only adds this data to the packets being forwarded within the local network.

The input policy of this scenario labels the input packet as {\low} and only marks the \texttt{deq\_qdepth} of the standard metadata as {\high}.
The contract of the \texttt{ipv4\_lpm} table forwards the packets with \texttt{ipv4} destination address \texttt{198.*.*.*} to the ports connected to the internal network, while all the other destination addresses are forwarded to ports connected to the external network.
If the \texttt{egress\_spec} indicates ports connected to the internal network, the contract of the \texttt{swtrace} table invokes the \texttt{add\_swtrace} action, adding the queue length data to \texttt{swtraces} header, otherwise \texttt{NoAction} takes place.

The output policy ensures that in all of the packets going to the external network, identified by their \texttt{hdr.ipv4.dstAddr} being anything other than \texttt{198.*.*.*}, have {\low} \texttt{switch\_t} header.

Table~\ref{tab:results} depicts the results of type checking this program with {\toolname}. It generates $23$ final state types and takes approximately $4$ seconds to verify the security of the program.
\texttt{swtrace} table is crucial for the security of this Program and if it is misconfigured and calls the \texttt{add\_swtrace} action on outgoing packets, the program will be rejected by {\toolname}.

\section{P4BID use cases}\label{sec:p4bid_usecases}
We implemented the use cases of P4BID~\cite{grewal2022p4bid} in {\toolname} to ensure that it can correctly evaluate all of their use cases and that its verdict is inline with the results reported in~\cite{grewal2022p4bid}.
These use cases and their corresponding policies are simpler than our own use cases because P4BID does not support data-dependent policies and hence only labels program variables without taking the value of the packet header fields into account. 
The results of this evaluation is depicted in Table~\ref{tab:results}.

\textbf{Dataplane Routing}
\emph{Routing} is the process of determining how to send a packet from its source to its destination. In traditional networks the control-plane is responsible for routing, but recently, Subramanian et al. \cite{subramanian2021d2r} proposed an approach to implement the routing in the data plane. 
Their approach uses pre-loaded information about the network topology and link failures to perform a breadth-first search (BFS) and find a path to the destination. 

In this scenario we do not care about the details of this BFS search algorithm, but we want to make sure that the sensitive information about the private network (such as the number of hops in the network) do not leak to an external network. 
Similar to P4BID~\cite{grewal2022p4bid} labeling the number of hops as $\high$ will result in the program being rejected by {\toolname} because the forwarding action uses this information to update the packet's priority field, which results in an indirect leakage of sensitive information.

\textbf{In-Network Caching}
In order to enable the fast retrieval of popular items, switches keep track of the frequently requested items in a cache and only query the controller when an item cannot be found in the cache. 
Similar to any cache system, the result of a query is the same regardless of where the item is stored. 
However, from a security perspective, an observer can potentially detect variations in item retrieval time. This timing side-channel can potentially allow an adversary to learn about the state of the system.

To model the cache in this scenario, we mark the request query as sensitive, because whether this query is a \emph{hit} or a \emph{miss} leaks information about the internal state of the switch.
The variable marking the state of the result in the cache (\texttt{response.hit}) is not sensitive because it is considered observable by the adversary.
Similar to P4BID~\cite{grewal2022p4bid} this labeling will result in the program being rejected by {\toolname} because a sensitive query can indirectly affect the value of \texttt{response.hit}, resulting in the leakage of sensitive information.

\textbf{Resource Allocation}
This use case models a simple resource allocation program, where the switch increases the priority of the packets belonging to latency-sensitive applications.
The application ID in the packet's header will indicate which application the packet belongs to. A table will matches on this application ID and
sets the packet's priority by modifying the \texttt{priority} field of the \texttt{ipv4} header.

The problem is that a malicious client can manipulate the application ID to increase the priority of their packets. 
We rely on integrity labels (instead of confidentiality) to address this issue, that is, the application ID will be labeled as untrusted ({\high}) and the \texttt{ipv4} \texttt{priority} field will be labeled as trusted ({\low}). 
Since the program sets the \texttt{priority} field based on the value of the application ID, the \texttt{priority} field will also be labeled untrusted by the type system, which results in the rejection of the program.

\textbf{Network Isolation}
This use case models a private network used by two clients, Alice and Bob. Each client runs its own P4 program, but the packets sent between these two clients have a shared header with separate fields for Alice and Bob. In this scenario we want to make sure that Alice does not touch Bob’s fields, and vice versa.

The isolation property in this example can be modeled by a four-point lattice with labels $\{A, B, \top, \bot \}$, where $A$ is the label of Alice's data, $B$ is for Bob's data, $\top$ is the top element confidential to both Alice and Bob, and $\bot$ is public. 
By IFC, data from level $\ell$ can flow to $\ell'$ if and only if $\ell \sqsubseteq \ell'$.

In this use case, we consider Alice's program in which she updates the fields belonging to herself. Additionally, we use label $\top$ to label the telemetry data which can be updated by Alice's program, but she cannot leak information from $\top$-labeled data into her own fields.

Since {\toolname} only support simple lattice with two levels, we type check this program twice, with two policies. First where Alice is {\high} and everything else is {\low}, and a second time where $\top$ is {\high} and everything else is {\low}. The same process can be repeated for Bob's program as well.
This program is accepted by {\toolname}, and the results for both cases are reported in Table~\ref{tab:results}.

\textbf{Topology}
This use case is a P4 program which processes packets as they enters a local network.
The incoming packets refers to a virtual address which needs to be translated to a physical address as the packet is routed in the local network.

Our security policy dictates that the routing details of this local network should not leak into fields that are visible when the packet leaves the network.
As such, the program relies on a separate header to store the local information, and as long as the packet is inside the local network, the switches do not modify the \texttt{ipv4} and Ethernet headers, instead, they parse, use, and update this local header with the routing information.

As explained in P4BID~\cite{grewal2022p4bid}, this program has a bug where it incorrectly stores the local \texttt{ttl} in the \texttt{ipv4} header instead of the local header. Marking the local fields as {\high}, {\toolname} flags this program as insecure and facilitates the process of catching and fixing these types of errors.

\section{State Type Operations}\label{sec:type_stm_op_app}
A type $\type'$ is considered an overapproximation of type $\type$ (written as ${\type \lesstype \type'}$) iff for every value ${v : \type}$ it holds that $v : \type'$ and $\getlabel(\type) \sqsubseteq \getlabel(\type')$.
We denote two non-overlapping lvalues by $\lval \hash \lval'$, which means if $\lval$ is a record $\lval'$ is not one of its fields, and if $\lval$ is a bitvector $\lval'$ is not one of its sub-slices.

We present the properties that operators over the state types must guarantee:
\begin{itemize}
	\item $\tscope[\lval \mapsto \type] = \tscope'$ indicates updating the
	type of $\lval$, which can be a part of a variable, in state type $\tscope$.
	This operator guarantees that
	$\tscope' \vdash \lval : \type$ and
	for every lvalue $\lval' \hash \lval$ such that
	$\tscope \vdash \lval' : \type'$
	and $\tscope' \vdash \lval' : \type''$ then
	$\type' \lesstype \type''$.
	
	\item $\tscope \concat \tscope'$ updates $\tscope$ such that for every variable in the domain of $\tscope'$, the type of that variable in $\tscope$ is updated to match $\tscope'$.
	
	\item $\refine{\tscope}{e} = \tscope'$
	returns an overapproximation of states that satisfy the abstraction of
	$\tscope$ and the predicate $e$.
	It guarantees that if $\tscope \vdash \mem$ and $e$ evaluates
	to \emph{true} in $\mem$ then
	$\tscope' \vdash \mem$ and for every $\lval$,
	if $\tscope \vdash \lval : \type$ and $\tscope' \vdash \lval : \type'$
	then $\getlabel(\type) \sqsubseteq \getlabel(\type')$
	
	\item $\join(\tscope_{1}, \tscope_{2}) = \tscope_{3}$
	returns an overapproximation of $\tscope_{1}$, whose labels are at least as restrictive as $\tscope_{1}$ \emph{and} $\tscope_{2}$.
	This operator guarantees that if $\tscope_{1} \vdash \mem$ then $\tscope_{3} \vdash \mem$ and for every $\lval$, then $\getlabel(\tscope_{1}(\lval)) \sqcup \getlabel(\tscope_{2}(\lval)) \sqsubseteq \getlabel(\tscope_{3}(\lval))$.
\end{itemize}

\section{Proofs and guarantees}\label{sec:proofs_guarantees}

We use $T \vdash E$ to represent the abstraction and labeling soundness guarantees for externs and tables, and in the following we assume that this condition holds.

\subsection{Sufficient condition proof}

\SufficientCondition*

\begin{proof}
	First, we prove $\tscope_{o} \vdash \mem_{2}$. 
	By definition, it is sufficient to prove that for every $\lval$, $\mem_{2}(\lval) : \tscope_{o}(\lval)$.\\
	From the definition of $\mem_{1}\consistent{\tscope_{1}}{\tscope_{2}} \mem_{2}$ we know that  
	$\tscope_{1} \vdash \mem_{1}$ and $\tscope_{2} \vdash \mem_{2}$ hold. 
	Since $\tscope_{o} \vdash \mem_{1}$ and $\tscope_{1} \vdash \mem_{1}$, then trivially
	$\tscope_{1} \cap \tscope_{o} \neq \bullet$ holds.
	From the second hypothesis of the sufficient condition, two cases are possible.
	\begin{enumerate}
		\item ${\tscope_{2}(\lval) \subseteq \tscope_{o}(\lval)}$.
	Since $\tscope_{2} \vdash \mem_{2}$ hold, 
	then by definition $\mem_{2}(\lval) : \tscope_{2}(\lval)$ holds, 
	therefore trivially $\mem_{2}(\lval) : \tscope_{o}(\lval)$ holds.
	
		\item ${\tscope_{1} \sqcup \tscope_{2}(\lval) = {\low}}$.
	Since $\mem_{1} \consistent{\tscope_{1}}{\tscope_{2}} \mem_{2}$ then, 
	indeed $\mem_{1}(\lval) = \mem_{2}(\lval)$ holds.
	By definition of $\tscope_o \vdash \mem_{1}$, 
	we know that $\mem_{1}(\lval) : \tscope_{o}(\lval)$ also holds. 
	Therefore, we can trivially show that
	$\mem_{2}(\lval) : \tscope_{o}(\lval)$.
	\end{enumerate}

	Second, we prove $\mem_1 \lequiv{\tscope_{o}} \mem_2$.
	Previously, we showed that $\tscope_{o} \vdash \mem_{2}$ holds, 
	and since $\tscope_{2} \vdash \mem_{2}$, 
	then trivially $\tscope_{2} \cap \tscope_{o} \neq \bullet$ holds. 
	From the first hypothesis of the sufficient condition, 
	we can show that $\tscope_{1} \sqcup \tscope_{2} \sqsubseteq \tscope_{o}$. 
	By definition of $\mem_{1}\consistent{\tscope_{1}}{\tscope_{2}} \mem_{2}$,
	we know that $\mem_1 \lequiv{\tscope_{1}\sqcup\tscope_{2}} \mem_2$ holds.
	Therefore, trivially $\mem_1 \lequiv{\tscope_{o}} \mem_2$. 
\end{proof}

\subsection{Hypothesis for refine}

\begin{hypothesis}[Interval typedness - boolean expressions' refinement]\label{hyp:interval_typedness_boolean_expressions_refinement}
	\begin{align*}
		\tscope \vdash e : \type \band \tscope &\vdash \mem \ \implies \\
		\Big( \expsem{\mem}{e}&={true} \ \implies \\
		&(\refine{\tscope}{e}) \vdash \mem   \ \band \ \expsem{\mem}{e}={false} \ \implies \\
		&(\refine{\tscope}{\neg e}) \vdash \mem \Big)
	\end{align*}
\end{hypothesis}

\begin{hypothesis}[Interval typedness - select expressions' refinement]\label{hyp:interval_typedness_select_expressions_refinement}
	\begin{align*}
		\tscope \vdash e:\type &\band \expsem{\mem}{e}={\val} \band \tscope \vdash \mem \\
		\band i = & \ \text{min}\{ i .\ \val = \val_i \lor i = n + 1\} \implies \\
		&\listindexed{\tscope}{n}, \tscope_{\text{n}+1} = (\refine{\tscope}{e=\val_i}) \ \implies \\
		&\tscope_i \vdash \mem
	\end{align*}
\end{hypothesis}

\begin{hypothesis}[Interval typedness - externs and tables refinement]\label{hyp:interval_typedness_externs_and_tables_refinement}
	\begin{align*}
		\tscope \vdash \mem  \band  \phi(\mem) \ \implies \ \refine{\tscope}{\phi}\vdash \mem 
	\end{align*}
\end{hypothesis}

\begin{hypothesis}[Label typedness - boolean expressions' refinement]\label{hyp:label_typedness_boolean_expressions_refinement}
	\begin{align*}
		\tscope \vdash e : \type_1 & \band \tscope \vdash e : \type_2 \\
		\band \getlabel(\type_1) = {\low} &\band \getlabel(\type_2) = {\low} \\ 
		&\band \mem_1 \lequiv{\tscope_1\sqcup\tscope_2} \mem_2 \ \implies \ \\
		\Big( (\refine{\tscope_1}{e} = \tscope_1' &\band \refine{\tscope_2}{e} = \tscope_2' \\
		\band \expsem{\mem_1}{e}={true} &\band \expsem{\mem_2}{e}={true} \implies \mem_1 \lequiv{\tscope1'\sqcup\tscope2'} \mem_2 ) \\
		\band \big(\refine{\tscope_1}{\neg e} = \tscope_1' &\band \refine{\tscope_2}{\neg e} = \tscope_2' \\
		\band \expsem{\mem_1}{\neg e}={true} &\band \expsem{\mem_2}{\neg e}={true} \implies \\
		&\mem_1 \lequiv{\tscope1'\sqcup\tscope2'} \mem_2 \big) \Big)
	\end{align*}
\end{hypothesis}

\begin{hypothesis}[Label typedness - select expressions' refinement]
	\begin{align*}
		\tscope \vdash e: \type \band \getlabel(\type) = {\low} &\band \mem_1 \lequiv{\tscope} \mem_2 \ \implies \ \\
		\Big( {\listindexed{\tscope}{n}, \tscope_{\text{n}+1}} = {\refine{\tscope}{e=\val_i}} &\band \expsem{\mem_1}{e}={\val} \implies \\
		&\forall j \leq n+1 . \ {\mem_1 \lequiv{\tscope_j} \mem_2} \Big)
	\end{align*}
\end{hypothesis}

\subsection{Lemmas}

\begin{lemma}[Expression reduction preserves the type]\label{lemma:expression_reduction_preserves_type}
	\begin{align*}
		\tscope \vdash \mem \band \tscope \vdash e : \type  \band \expsem{\mem}{e}={\val} \ \implies \ \val : \type
	\end{align*}
\end{lemma}

\begin{lemma}[lvalue updates preserves the type]\label{lemma:lval updates_preserves_type}
	\begin{align*}
		\tscope \vdash \mem  \band  \val: \type \ \implies \ \tscope[\lval \mapsto \type] \vdash \mem[\lval \mapsto \val]
	\end{align*}
\end{lemma}

\begin{lemma}[Expressions have types same as their values]\label{lemma:expressions_have_types_same_as_their_values}
	\begin{align*}
		\tscope \vdash \mem  \band  \tscope \vdash e : \type \ \implies \ \exists \val . \ \expsem{\mem}{e}={\val} \band \val : \type
	\end{align*}
\end{lemma}

\begin{lemma}[Join does not modify the intervals]\label{lemma:join_does_not_modify_the_intervals}
	\begin{align*}
		\tscope \vdash \mem  \band  \tscope \in \Gamma_1 & \implies \\
		\forall \Gamma_2 ... \Gamma_n . \ & \exists \tscope' \in \text{join} (\Gamma_1 \cup \Gamma_2 \cup ... \cup \Gamma_n) . \\
		&\tscope' \vdash \mem \band \tscope \sqsubseteq \tscope'
	\end{align*}
\end{lemma}

\begin{lemma}[Expression evaluation of consistent states]\label{lemma:expression_evaluation_of_consistent_states}
	\begin{align*}
		\mem_1 \lequiv{\tscope_1\sqcup\tscope_2} \mem_2 \band & \tscope_1 \vdash \mem_1 \band \tscope_2 \vdash \mem_2 \band \\
		\tscope_1 \vdash e : \type_1 \band& \tscope_2 \vdash e : \type_2 \band \\
		\getlabel(\type_1)& = {\low} \band \getlabel(\type_2) = {\low}  \implies \\
		&(\expsem{\mem_1}{e}={\val} \band \expsem{\mem_2}{e}={\val})
	\end{align*}
\end{lemma}

\begin{lemma}[State equivalence preservation]\label{lemma:state_equivalence_preservation}
	\begin{align*}
		\tscope' \sqsubseteq \tscope \band \mem_1 \lequiv{\tscope'} \mem_2 \ \implies \ \mem_1 \lequiv{\tscope} \mem_2
	\end{align*}
\end{lemma}

\begin{lemma}[Branch on high - state preservation]\label{lemma:branch_on_high_state_preservation}
	\begin{align*}
		{\stmtctx}: \sem{\mem}{\stmt}{\mem'} \band \tc, {\high}, \tscope \vdash \stmt : \Gamma \band \tscope' \in \Gamma &\implies \\
		(\tscope'(\lval) = \type \band \getlabel(\type) = {\low}) &\implies \\
		\mem(\lval) & = \mem'(\lval)
	\end{align*}
\end{lemma}

\begin{lemma}[Join's low label implication]\label{lemma:joins_low_label_implication}
	\begin{align*}
		\Gamma = \text{join} &(\Gamma_1 \cup ...   \cup\Gamma_n) \band \tscope \in \Gamma \band \\
		\tscope(\lval) & = \type \band \getlabel(\type) = {\low}  \implies \\
		&\forall \tscope' \in \Gamma_1, ... , \Gamma_n . \ \tscope'(\lval) = \type' \band \getlabel(\type') = {\low}
	\end{align*}
\end{lemma}

\begin{lemma}[High program's final types]\label{lemma:high_program_final_types_subset_of_initial}
	\begin{align*}
		{\tc, {\high}, \tscope \vdash \stmt : \Gamma} \ \implies \ \forall \tscope' \in \Gamma . \ \tscope \sqsubseteq \tscope'
	\end{align*}
\end{lemma}

\begin{lemma}[Branch on high - state lemma]\label{lemma:branch_on_high_state_lemma}
	\begin{align*}
		\tc, {\high}, \tscope \vdash \stmt : \Gamma &\band \tscope' \in \Gamma \band \\
		\tscope'(\lval) = \type' &\band \getlabel(\type') = {\low} \implies \\
		\forall \tscope'' \tscope''' . \ &  {\tc, pc, \tscope'' \vdash \stmt : \Gamma'} \band {\tscope''' \in \Gamma' } \implies \ \\
		\Big( &\tscope''(\lval) = \type'' \band \tscope'''(\lval) = \type''' \\
		& \implies \getlabel (\type'') \sqsubseteq \getlabel (\type''') \Big)
	\end{align*}
\end{lemma}

\begin{lemma}[Branch on high is never empty]\label{lemma:branch_on_high_is_never_empty}
	\begin{align*}
		{\tc, {\high}, \tscope \vdash \stmt : \Gamma} \implies \Gamma \neq \varnothing
	\end{align*}
\end{lemma}

\begin{lemma}[Low equivalence distribution]\label{lemma:low_equivalence_distribution}
	\begin{align*}
		(\mem_1,\mem_1') & \lequiv{(\tscope_a,\tscope_a') \sqcup (\tscope_b,\tscope_b')} (\mem_2,\mem_2') \Leftrightarrow \\
		&(\mem_1 \lequiv{\tscope_a\sqcup \tscope_b} \mem_2 \band \mem_1' \lequiv{\tscope_a'\sqcup \tscope_b'} \mem_2' )
	\end{align*}
\end{lemma}

\begin{lemma}[Low equivalence update]\label{lemma:low_equivalence_update}
	\begin{align*}
		&\mem_1 \lequiv{\tscope_a\sqcup \tscope_b} \mem_2 \band \mem_3 \lequiv{\tscope_c\sqcup \tscope_d} \mem_4 \implies \\
		&\mem_1[\lval \mapsto \mem_3(e)] \lequiv{\tscope_a[\lval \mapsto \tscope_c(e)] \sqcup \tscope_b[\lval \mapsto \tscope_d(e)]} \mem_2[\lval \mapsto \mem_4(e)]
	\end{align*}
\end{lemma}

\subsection{Soundness of abstraction}
\begin{theorem}
	\begin{align*}
		\forall \, \stmt \, \tc \, \mem \, \tscope \, pc \, \Gamma. \ & \tc, pc, \tscope \vdash \stmt : \Gamma
		\ \implies \\
		&\tc  \vdash \stmtctx
		\band
		\tscope \vdash \mem
		\implies \\
		&\exists \mem' . {\stmtctx}: \sem{\mem}{\stmt}{\mem'}
		\band
		\exists \tscope' \in \Gamma. \ \tscope' \vdash \mem'
	\end{align*}
\end{theorem}

\begin{proof}
In this proof we assume that the program is well typed and does not get stuck according to HOL4P4 type system.
By induction on the typing tree of the program $stmt$, i.e. $\stmt$.
Note that, initially $\tscope = (\tscope_g,\tscope_m)$ and $\mem=(m_g,m_l)$.
In the following proof, we know that $\tscope \vdash \mem$ holds in all subcases.

\casesItem{Case assignment}: Here $stmt$ is $\lval := e$. 
From assignment typing rule we know:
\begin{enumerate}
    \item $\tscope \vdash {e} : \type$
    \item $\type' = \raisel(\type, pc)$
    \item $\tscope' = \tscope[\lval \mapsto \type']$
    \item $\Gamma = \{\tscope'\}$
\end{enumerate}

We need to prove $\exists \mem' . {\stmtctx}: \sem{\mem}{\lval := e}{\mem'}$ and $\exists \tscope'' \in \Gamma. \tscope'' \vdash \mem'$. 
From the assignment reduction definition, we can rewrite the goal's conjunctions to:
\begin{enumerate}
    \item $\expsem{\mem}{e}={\val}$ \hspace{69pt} (from assignment reduction)
    \item $\exists \mem' . \  \mem'= \mem[\lval \mapsto \val]$ \hspace{8pt}  (from assignment reduction)
    \item $\exists \tscope'' \in \{\tscope'\}. \tscope'' \vdash \mem'$
\end{enumerate}

\textbf{Goal 1}: we know that the initial variable map is typed using the state type, i.e. $\tscope\vdash \mem$ from assumptions. Using $\tscope\vdash \mem$ and 1 from typing rule, then we can use Lemma \ref{lemma:expressions_have_types_same_as_their_values} to
directly infer that the expression's reduction indeed keeps the type, i.e. $\exists \val . \ \expsem{\mem}{e}={\val} \band \val : \type $ and this resolves goal 1.

\textbf{Goal 2}: trivial, as we can instantiate $\mem'$ to be $\mem[\lval \mapsto \val]$.

\textbf{Goal 3}: we know that assignment typing rule produces a singleton set from 4 of typing rule, thus we can instantiate $\tscope''$ to be $\tscope[\lval \mapsto \type'] $, thus allows us to rewrite the goal to $\tscope[\lval \mapsto \type'] \vdash \mem[\lval \mapsto \val]$.

Given $\tscope\vdash \mem$, this entails (by definition) that $\tscope$ and $\mem$ they contain the same
variable name in the domain, and also the values are well-typed, i.e.
$domain(\tscope) = domain(\mem) \band \forall x \in domain(\mem). \ \mem(x) : \tscope(x)$.

Using Lemma \ref{lemma:expression_reduction_preserves_type}, we know that the assigned value $\val$ can be typed with $\type$, i.e. $\val:\type$.
The assignment typing rule raises the labels using the function raise, however by definition we know that the abstraction is unaffected,
so the interval of $\type'$ and $\type$ are the same, thus we can infer that $\val:\type'$.

Using Lemma \ref{lemma:lval updates_preserves_type}, we can see that the update preserves the type, thus goal proved.

\casesItem{Case condition}: Here $stmt$ is $\texttt{if } e \texttt{ then } \stmt_1 \texttt{ else } \stmt_2$. 
From conditional typing rule we know:
\begin{enumerate}
    \item $\tscope \vdash e : \type$
    \item $\lbl = \getlabel(\type)$
    \item $pc' = pc \sqcup \lbl$
    \item $\tc, pc', (\refine{\tscope}{e}) \vdash \stmt_1 : \Gamma_1$
    \item $\tc, pc', (\refine{\tscope}{\neg e}) \vdash \stmt_2 : \Gamma_2$
    \item $\Gamma' = \Gamma_1 \cup \Gamma_2 $
    \item $\Gamma'' =  {\begin{cases}
                  \text{join}(\Gamma') \hspace{15pt} \text{if } \lbl = {\high} \\
                  \Gamma'   \hspace{38pt}  \text{otherwise}
              \end{cases}}$
\end{enumerate}

We need to prove $\exists \mem' . {\stmtctx}: \sem{\mem}{\texttt{if } e \texttt{ then } \stmt_1 \texttt{ else } \stmt_2}{\mem'}$ and $\exists \tscope'' \in \Gamma''. \tscope'' \vdash \mem'$.

In the goal, the boolean guard $e$ evaluates to $true$ or $false$.
So we will get two goal cases with similar proofs. This only solves for $e$, while case $\neg e$ follows the same proof strategy.
We can rewrite the goal to:
\begin{enumerate}
    \item $\exists \mem' . \ {\stmtctx} : \sem{\mem}{\stmt_1}{\mem'}$
    \item $\exists \tscope'' \in \Gamma''. \tscope'' \vdash \mem'$
\end{enumerate}

In this proof, we will get two induction hypotheses for statement: for $\stmt_1$ call it \textsc{IH1} and for $\stmt_2$ call it \textsc{IH2}.

\textsc{IH1} is the following (note that \textsc{IH2} is the same, but instantiated for $\stmt_2$): 
\begin{align*}
	\forall \tc \, \mem \, \tscope \, pc \, \Gamma . \ \tc, pc, \tscope \vdash \stmt_1 : \Gamma \ &\implies \\
	\tc  \vdash \stmtctx \band \tscope \vdash \mem &\implies \\
	\exists \mem' . \  {\stmtctx}: \sem{\mem}{\stmt_1}{\mem'} &\band \exists \tscope' \in \Gamma. \ \tscope' \vdash \mem'
\end{align*}

We first prove that the set of refined state types using $e$ can still type the staring concrete memory $\mem$.
This we can show, because we know initially $\tscope \vdash \mem$, and we know that $e$ is typed as a boolean (assumed to be well-typed), and we know that $e$ reduces to $true$ from the reduction rule,
these allow us to infer $(\refine{\tscope}{e}) \vdash \mem$ using Hyp \ref{hyp:interval_typedness_boolean_expressions_refinement}.

Now we instantiate the induction
hypothesis \textsc{IH1} using the following $(\tc ,\mem, \refine{\tscope}{e}, pc', \Gamma_1)$,
to show that  exists $\mem'$ such that ${\stmtctx}: \sem{\mem}{\stmt_1}{\mem'}$, i.e. there indeed exists a transition to a final configuration in the semantics to $\mem'$ (which resolves goal 1). Additionally, we can show from \textsc{IH1} that exists $\tscope' \in \Gamma_1$ such that $\tscope' \vdash \mem'$.

Now we implement cases on the expression's label $\lbl$ being {\high} or {\low}:

\begin{enumerate}[label=\textbf{case}]
    \item $\getlabel(\type) = {\high} $:
          we need to prove $\exists \tscope'' \in \text{join } \{\Gamma_1 \cup \Gamma_2 \}. \tscope'' \vdash \mem'$. Then it is easy to deduct that the goal holds; because we showed that there is a state type $\tscope' \in \Gamma_1$ such that it is a sound abstraction of the final state $\tscope' \vdash \mem'$, and we know that join operation does not change the abstraction, it just modifies the security label. Hence, indeed there exists a state type $\tscope''$ in $\text{join } \{\Gamma_1 \cup \Gamma_2 \}$ such that it is also a sound abstraction of final state $\tscope'' \vdash \mem'$.

    \item $\getlabel(\type) = {\low} $: we need to prove $\exists \tscope'' \in \{\Gamma_1 \cup \Gamma_2 \} . \tscope'' \vdash \mem'$, which is trivially true.
\end{enumerate}

For the negation case, use \textsc{IH2} and follow the same steps.

\casesItem{Case sequence}: Here $stmt$ is $\stmt_1;\stmt_2$. 
From sequence typing rule we know:
\begin{enumerate}
    \item $\tc, pc, \tscope \vdash {\stmt_1} : \Gamma_1 $
    \item $\forall \tscope_1 \in \Gamma_1 . \ \tc, pc,  \tscope_1 \vdash \stmt_2 : \Gamma_2^{\tscope_1}$
    \item $\Gamma' =  \bigcup_{\tscope_1 \in \Gamma_1} \Gamma_2^{\tscope_1}$
\end{enumerate}

In this proof, we will get two induction hypotheses for statement: for $\stmt_1$ call it \textsc{IH1} and for $\stmt_2$ call it \textsc{IH2}. \\
\textsc{IH1} is (note that \textsc{IH2} is the same, but instantiated for $\stmt_2$):
\begin{align*}
	\forall \tc \mem \tscope \, pc \, \Gamma . \ \tc, pc, \tscope \vdash \stmt_1 : \Gamma &\implies \\
	\tc  \vdash \stmtctx \band \tscope \vdash \mem &\implies \\
	\exists \mem' . {\stmtctx}&: \sem{\mem}{\stmt_1}{\mem'} \band \exists \tscope' \in \Gamma. \tscope' \vdash \mem'
\end{align*}

We need to prove $\exists \mem'' . \ {\stmtctx}: \sem{\mem}{\stmt_1;\stmt_2}{\mem''}$ and $\exists \tscope'' \in \Gamma'. \ \tscope'' \vdash \mem''$. \\

We can rewrite the goal using the definition of sequence case to:
\begin{enumerate}
    \item ${\stmtctx}: \sem{\mem}{\stmt_1}{\mem'}$
    \item $\exists \mem'' . {\stmtctx}: \sem{\mem'}{\stmt_2}{\mem''}$
    \item $\exists \tscope'' \in \Gamma'. \tscope'' \vdash \mem''$
\end{enumerate}

\textbf{Goal 1}: We can instantiate the \textsc{IH1} with $(\tc, \mem, \tscope , pc , \Gamma_1)$ to
infer that exists $\mem_1'$ such that ${\stmtctx}: \sem{\mem}{\stmt_1}{\mem_1'}$, and also exists $\tscope' \in \Gamma_1$ such that $\tscope' \vdash \mem_1'$. Since the semantics are deterministic, $\mem_1'$ and $\mem'$ are equivalent, thus it holds ${\stmtctx}: \sem{\mem}{\stmt_1}{\mem'}$ and $\tscope' \vdash \mem'$.

\textbf{Goal 2 and 3} : We showed from \textsc{IH1} that $\tscope' \vdash \mem'$, now we can instantiate 2 from sequence typing rule with $\tscope'$.
Now we can to instantiate \textsc{IH2} with $(\tc, \mem', \tscope' , pc , \Gamma_2^{\tscope'})$, and infer that exists $\mem_2'$ such that ${\stmtctx}: \sem{\mem'}{\stmt_2}{\mem_2'}$, and also exists $\tscope''' \in \Gamma_2^{\tscope'}$ such that $\tscope''' \vdash \mem_2'$. Since the semantics are deterministic, $\mem_2'$ and $\mem''$ are equivalent, thus it holds ${\stmtctx}: \sem{\mem}{\stmt_2}{\mem''}$ and $\tscope''' \vdash \mem'$.

From 3 in sequence typing rule, we know that $\Gamma'$ is the union of all resulted state type sets, such that it can type $\stmt_2$,
since we know that $\tscope''' \in \Gamma_2^{\tscope'}$ will be in $\Gamma'$, then we prove the goal $\exists \tscope'' \in \Gamma'. \tscope'' \vdash \mem''$.

\casesItem{Case function call}: Here $stmt$ is $f(e_{1} \cons e_{n})$. 
From call typing rule we know (note that here we explicitly write the global and local state type):
\begin{enumerate}
    \item $(\tscope_g, \tscope_l) \vdash  e_{i} : \type_{i} $
    \item $(\stmt,(x_{1},d_{1}) \cons (x_{n}, d_{n})) = \tcv(f) $
    \item $\tscope_f =\{x_{i} \mapsto \type_{i}\}$
    \item $ \tcv, pc, (\tscope_g, \tscope_f) \vdash \stmt : \Gamma' $
    \item $\Gamma'' =
              \{(\tscope'_{g}, \tscope_{l})  [e_{i} \mapsto \tscope'_{f}(x_{i}) \mid \isOut(d_{i})]
              \mid (\tscope'_{g}, \tscope'_{f}) \in \Gamma' \}$
\end{enumerate}

In the following proof, we know that $(\tscope_g,\tscope_l) \vdash (\mem_g,\mem_l)$ and $\tcv  \vdash \stmtctxv$ hold.
Additionally, we get an induction hypothesis \textsc{IH} for the body of the function $\stmt$.
\begin{align*}
	\forall \tcv \ \mem \tscope \, pc \, \Gamma . \ &\tcv, pc, \tscope \vdash \stmt : \Gamma \implies \\
	\tcv  \vdash &\stmtctxv \band \tscope \vdash \mem \implies \\
	\exists \mem' . &\stmtctxv: \sem{\mem}{\stmt}{\mem'} \band \exists \tscope' \in \Gamma. \tscope' \vdash \mem'
\end{align*}

We need to prove $\exists \mem' . \ {\stmtctxv}: \sem{(\mem_g,\mem_l)}{f(e_{1} \cons e_{n})}{\mem'}
    \band
    \exists \tscope' \in \Gamma''. \ \tscope' \vdash \mem'$

From call reduction rule we can rewrite the goal to:
\begin{enumerate}
    \item  $\exists \stmt \  (x_{1},d_{1}) \cons (x_{n}, d_{n}) . \ \ (\stmt,(x_{1},d_{1}) \cons (x_{n}, d_{n})) = \stmtctxv(f)$
    \item  $\exists \mem_f  . \  \mem_f =\{x_{i} \mapsto \expsem{(\mem_{g}, \mem_{l})}{e_{i}}\}$
    \item  $\exists (\mem'_{g}, \mem'_{f}) . \ {\stmtctxv}: \sem{(\mem_{g}, \mem_{f})}{\stmt}{(\mem'_{g}, \mem'_{f})} $
    \item  $\exists \mem'' . \mem'' = (\mem'_{g}, \mem_{l})[e_{i} \mapsto \mem'_{f}(x_{i}) \mid \isOut(d_{i})]$
    \item  $\exists \tscope' \in \Gamma''. \tscope' \vdash \mem''$
\end{enumerate}

\textbf{Goal 1}:
From $\tcv  \vdash \stmtctxv$ we know indeed the function's body and signature found in the semantics is the same found in the typing rule.

\textbf{Goal 2}:
Trivial, the existence can be instantiated with $\{x_{i} \mapsto \expsem{(\mem_{g}, \mem_{l})}{e_{i}}\}$.

\textbf{Goal 3}:
To prove that there exists a state where the body of the function reduces to, we need to use the \textsc{IH}.
Thus, we first
need to show that the resulted copy-in map is also well-typed i.e. $\tscope_f\vdash\mem_f$, more specifically $\forall i . \ \{x_{i} \mapsto \type_{i}\} \vdash \{x_{i} \mapsto \expsem{(\mem_{g}, \mem_{l})}{e_{i}}\}$. Given that initially the state type can type the state $(\tscope_g,\tscope_l) \vdash (\mem_g,\mem_l)$ from assumptions and given that the expressions $e_i$ have a type $\type_i$ from typing rule 1, now we can use Lemma \ref{lemma:expression_reduction_preserves_type} to infer that $\forall \val_i : \type_i$ such that $\val_i$ is the evaluation of $\expsem{(\mem_{g}, \mem_{l})}{e_{i}}$. This leads us to trivially infer that $\forall i . \ \{x_{i} \mapsto \type_{i}\} \vdash \{x_{i} \mapsto \val_i \}$ holds.

Now we can use the \textsc{IH} by instantiating it to $(\tcv,(\mem_g,\mem_f), (\tscope_g,\tscope_f), pc, \Gamma')$ in order to infer that exists $(\mem'_{g}, \mem'_{f})$ such that ${\stmtctxv}: \sem{(\mem_{g}, \mem_{f})}{\stmt}{(\mem'_{g}, \mem'_{f})}$ and exists $(\tscope_g'',\tscope_f'') \in \Gamma'$ such that $(\tscope_g'',\tscope_f'') \vdash (\mem'_g,\mem'_f)$. 

\textbf{Goal 4}:
Trivial, we can instantiate the existence by $(\mem'_{g}, \mem_{l})[e_{i} \mapsto \mem'_{f}(x_{i}) \mid \isOut(d_{i})]$.

\textbf{Goal 5}:
The goal is to prove copy-out operation to be well-typed, thus we can rewrite the goal to
$\exists \tscope' \in \{(\tscope'_{g}, \tscope_{l}) [e_{i} \mapsto \tscope'_{f}(x_{i}) \mid \isOut(d_{i})]
    \mid (\tscope'_{g}, \tscope'_{f}) \in \Gamma' \}$
such that
$\tscope '
    \vdash
    (\mem'_{g}, \mem_{l})[e_{i} \mapsto \mem'_{f}(x_{i}) \mid \isOut(d_{i})]$.

We can choose $\tscope'$ to be $(\tscope_{g}'', \tscope_{l}) [e_{i} \mapsto \tscope_{f}''(x_{i}) \mid \isOut(d_{i})]$.

Since we are able to choose a state type that types the final state, we can rewrite the goal to prove again to be
$(\tscope_{g}'', \tscope_{l}) [e_{i} \mapsto \tscope_{f}''(x_{i}) \mid \isOut(d_{i})]
    \ \vdash \
    (\mem'_{g}, \mem_{l})[e_{i} \mapsto \mem'_{f}(x_{i}) \mid \isOut(d_{i})]$.

First, we know that $\tscope_l\vdash\mem_l$ holds from the assumptions. Also, we showed in (Goal 3) that $(\tscope_g'',\tscope_f'') \vdash (\mem'_g,\mem'_f)$, thus trivially $\tscope_g''\vdash \mem'_g$ and $\tscope_f'' \vdash \mem'_f$ hold. Thus, we can deduct that $(\tscope_{g}'', \tscope_{l}) \vdash (\mem_{g}'', \mem_{l})$, and $\mem'_{f}(x_{i}):\tscope_{f}''(x_{i})$, and then we can use Lemma \ref{lemma:lval updates_preserves_type} to deduct that the update preserves the well-typedness, i.e. $(\tscope_{g}'', \tscope_{l}) [e_{i} \mapsto \tscope_{f}''(x_{i})]
    \ \vdash \
    (\mem'_{g}, \mem_{l})[e_{i} \mapsto \mem'_{f}(x_{i})]$, which proves the goal.

\casesItem{Case extern}: Here $stmt$ is $f(e_{1}\cons e_{n})$. From extern typing rule we know (note that here we explicitly write the global and local state type):

\begin{enumerate}
    \item $(\tscope_g, \tscope_l) \vdash  e_{i} : \type_{i}$
    \item $ (\text{Cont}_{\mathrm{E}},(x_{1},d_{1}) \cons (x_{n}, d_{n})) = \tcv(f) $
    \item $ \tscope_f =\{x_{i} \mapsto \type_{i}\}$
    \item $\forall (\tscope_i, \phi,\tscope_t) \in \text{Cont}_{\mathrm{E}}. \ (\tscope_g, \tscope_l) \sqsubseteq \tscope_i$
    \item $ \Gamma' = \{ \tscope' \concat \raisel(\tscope_t, pc) \ | \ (\tscope_i,\phi,\tscope_t)   \in \text{Cont}_{\mathrm{E}} \band \refine{(\tscope_{g}, \tscope_f)}{\phi} = \tscope' \neq \bullet  \} $
    \item   $ \Gamma'' =
              \{(\tscope'_{g}, \tscope_{l}) [e_{i} \mapsto \tscope'_{f}(x_{i}) \mid \isOut(d_{i})]
              \mid (\tscope'_{g}, \tscope'_{f}) \in \Gamma'\}$
\end{enumerate}

In the following proof, we know that $(\tscope_g,\tscope_l) \vdash (\mem_g,\mem_l)$ and $\tcv  \vdash \stmtctxv$ hold.

We need to prove $\exists \mem' . \ {\stmtctxv}: \sem{(\mem_g,\mem_l)}{f(e_{1} \cons e_{n})}{\mem'}
    \band
    \exists \tscope'' \in \Gamma''. \ \tscope'' \vdash \mem'$ \\

From extern reduction rule we can rewrite the goal to:
\begin{enumerate}
    \item $\exists (sem_{f},(x_{1},d_{1}) \cons (x_{n}, d_{n})) . \ (sem_{f},(x_{1},d_{1}) \cons (x_{n}, d_{n})) = \stmtctxv(f) $
    \item $\exists \mem_f . \ \mem_f =\{x_{i} \mapsto \expsem{(\mem_{g}, \mem_{l})}{e_{i}}\}$
    \item $\exists (\mem'_{g}, \mem'_{f}) . \ (\mem'_{g}, \mem'_{f}) = sem_{f}(\mem_{g}, \mem_{f})$
    \item $\exists \mem'' . \mem'' = (\mem'_{g}, \mem_{l})[e_{i} \mapsto \mem'_{f}(x_{i}) \mid \isOut(d_{i})]$
    \item $\exists \tscope'' \in \Gamma'' . \tscope'' \vdash \mem''$
\end{enumerate}

\textbf{Goal 1}:
From the environment's well-typedness $\tcv  \vdash \stmtctxv$, we know $\text{domain}(C) \cap \text{domain}(F) = \emptyset$ and $\text{extWT } C \ X$ holds. This mean that indeed the extern is defined only in $C$. Additionally, from well-typedness $\tcv  \vdash \stmtctxv$, that if $C(f)= (sem_f,\overline{(x,d)})$ then $X(f) =(\text{Cont}_E,\overline{(x,d)})$, thus indeed exist $\text{Cont}_E$ and signature $(x_{1},d_{1}) \cons (x_{n}, d_{n})$. 

\textbf{Goal 2}:
Trivial, by instantiating $\mem_f$ to be $\{x_{i} \mapsto \expsem{(\mem_{g}, \mem_{l})}{e_{i}}\}$.

We can here also prove that the resulted copy-in map is also well-typed i.e. $\tscope_f\vdash\mem_f$, more specifically $\forall i . \ \{x_{i} \mapsto \type_{i}\} \vdash \{x_{i} \mapsto \expsem{(\mem_{g}, \mem_{l})}{e_{i}}\}$.
Given that initially the typing state types the state $(\tscope_g,\tscope_l) \vdash (\mem_g,\mem_l)$ from assumptions and given that the expressions $e_i$ have a type $\type_i$ from typing rule 1, now we can use Lemma \ref{lemma:expression_reduction_preserves_type} to infer that $\forall \val_i : \type_i$ such that $\val_i$ is the evaluation of $\expsem{(\mem_{g}, \mem_{l})}{e_{i}}$. This leads us to trivially infer that $\forall i . \ \{x_{i} \mapsto \type_{i}\} \vdash \{x_{i} \mapsto \val_i \}$ holds.

\textbf{Goal 3}:
From $\tcv  \vdash \stmtctxv$, we know that $\text{extWT } C \ X$ holds, and from its definition, we know that indeed exists $(\tscope_i,\phi,\tscope_t)$ such that $\phi(\mem_g,\mem_f)$, this means that indeed there exists a contract's predicate satisfied by the values in the initial concrete input state, i.e.  $\phi(\mem_g,\mem_f)$.
This implies, from the definition of $\text{extWT } C \ X$, that indeed exists $(\mem_g',\mem_f')$ such that $(\mem_g',\mem_f') =  sem_f(\mem_g,\mem_f)$.

\textbf{Goal 4}: Trivial, by instantiating $\mem''$ to $(\mem'_{g}, \mem_{l})[e_{i} \mapsto \mem'_{f}(x_{i}) \mid \isOut(d_{i})]$. 

\textbf{Goal 5}: We can rewrite the goal to exists $\tscope'' \in \{(\tscope'_{g}, \tscope_{l}) [e_{i} \mapsto \tscope'_{f}(x_{i}) \mid \isOut(d_{i})]
    \mid (\tscope'_{g}, \tscope'_{f}) \in \Gamma'\} $ such that $\tscope'' \vdash (\mem'_{g}, \mem_{l})[e_{i} \mapsto \mem'_{f}(x_{i}) \mid \isOut(d_{i})]$. 

We can prove this goal
by first find a $\tscope_A \in \Gamma'$ such that it types the final states of the extern's semantic $(\mem_g',\mem_f')$. Second, we find the $\tscope'' \in \Gamma''$ such that it can type the final state $\mem''$ after copying out the extern. \\

We previously established $\tscope_f\vdash\mem_f$ (from goal 2), also given that $\tscope_g\vdash\mem_g$ from assumptions we can trivially $(\tscope_g,\tscope_f)\vdash(\mem_g,\mem_f)$.
Since we previously showed that $\phi(\mem_g,\mem_f)$ (from goal 3), now we can use Hyp \ref{hyp:interval_typedness_externs_and_tables_refinement} in order to infer that the refined state $\refine{(\tscope_g,\tscope_f)}{\phi}$ can also type $(\mem_g,\mem_f)$ also we infer that it is not empty, i.e. $\refine{(\tscope_g,\tscope_f)}{\phi}\vdash(\mem_g,\mem_f)$ and $\refine{(\tscope_g,\tscope_f)}{\phi} \neq \bullet$.

The definition of $\text{extWT } C \ X$ states that
$\tscope_t \vdash (\mem_g',\mem_f')$ holds for the set of variables that the extern's semantics has changed i.e. $\{x. \ (\mem_g,\mem_f)(x) \neq (\mem_g',\mem_f')(x) \} \subseteq domain(\tscope_t)$.

Consequently, we can further prove that indeed exists a $\tscope_A$ in $\Gamma'$ (from 4 in typing rule of extern) that can type the output of the extern's semantics $(\mem_g',\mem_f')$ including the unchanged variables. Thus, we can make cases on extern's semantics input and output as following:

\begin{enumerate}[label=\textbf{case}]
    \item $(\mem_g,\mem_f)(x) = (\mem_g',\mem_f')(x)$:
          This means that variable $x$ not in the domain of $\tscope_t$, thus it is unchanged, therefore it is typed by the refined state $(\mem_g',\mem_f')(x):(\refine{(\tscope_g,\tscope_f)}{\phi})(x)$. Trivially, we can also infer that $(\mem_g',\mem_f')(x):(\refine{(\tscope_g,\tscope_f)}{\phi}\concat \raisel(\tscope_t, pc))(x)$.

    \item $(\mem_g,\mem_f)(x) \neq (\mem_g',\mem_f')(x)$:
          This means that variable $x$ is in the domain of $\tscope_t$, thus it is changed, and the new type of it is in $\tscope_t$. Therefore, we can trivially conclude that $(\mem_g,\mem_f)(x):\tscope_t(x)$, Consequently, since we know that raise does not change the abstraction, and just change labels, we can therefore conclude that $(\mem_g',\mem_f')(x):(\refine{(\tscope_g,\tscope_f)}{\phi}\concat \raisel(\tscope_t, pc))(x)$
\end{enumerate}

These cases show
that we can select $\tscope_A$ such that $\tscope_A \in \Gamma'$ to be $(\refine{(\tscope_g,\tscope_f)}{\phi}\concat \raisel(\tscope_t, pc))$, because it can indeed type $(\mem_g',\mem_f')$. i.e. $(\refine{(\tscope_g,\tscope_f)}{\phi}\concat \raisel(\tscope_t, pc)) \vdash (\mem_g',\mem_f')$. For simplicity in the rest of the proof, let us rewrite $\tscope_A = (\tscope_{A_g},\tscope_{A_f})$, where $\tscope_{A_g}$ is the global part of the pair $(\refine{(\tscope_g,\tscope_f)}{\phi}\concat \raisel(\tscope_t, pc))$ and $\tscope_{A_f}$ is the local part of the pair.
Thus, we can say that $\tscope_{A_g}\vdash\mem_g'$ and $\tscope_{A_f}\vdash\mem_f'$.

Now we need to find $\tscope'' \in \Gamma''$ such that it types $(\mem'_{g}, \mem_{l})[e_{i} \mapsto \mem'_{f}(x_{i}) \mid \isOut(d_{i})]$ order to prove Goal 5.

Given from assumptions $\tscope_l\vdash\mem_l$, and we showed that $\tscope_{A_g}\vdash\mem_g'$ and $\tscope_{A_f}\vdash\mem_f'$. In 5 of the typing rules, we pick $(\tscope_g',\tscope_f')$ such that it is in  $\Gamma'$ to be $(\tscope_{A_g},\tscope_{A_f})$,
Now we conduct cases on the direction of the parameter $\isOut(d_i)$ being \texttt{out} or not.
\begin{enumerate}[label=case]
    \item $\neg \isOut(d_i)$:
          Then $(\tscope_{A_g},\tscope_l)$ are unchanged, similarly in the semantics $(\mem'_{g}, \mem_{l})$ are also unchanged. Thus they can still be typed as $(\tscope_{A_g},\tscope_l)\vdash(\mem'_{g}, \mem_{l})$
    \item $\isOut(d_i)$:
          Then $(\tscope_{A_g},\tscope_l)$ are updated with $e_i \mapsto \tscope_{A_f}(x)$, similarly the semantics state $(\mem_g',\mem_l)$ is updated with $e_i \mapsto \mem_f'(x)$. Previously we showed that $\tscope_{A_f}\vdash\mem_f'$, this entails by the definition of state typedness $\mem_f'(x):\tscope_{A_f}(x)$. This leads to the point that the modifications of $e_i$'s type in state type $(\tscope_{A_g},\tscope_l)$ and value in state $(\mem_g',\mem_l)$ keeps them well typed. Therefore, $(\tscope_{A_g},\tscope_l)\vdash(\mem'_{g}, \mem_{l})$ holds.
\end{enumerate}

We can finally conclude that goal 5 can be resolved by picking the $\tscope''$ to be $(\tscope_{A_g},\tscope_l)$, the goal is now proven.

\casesItem{Case table application}: Here $stmt$ is $\texttt{apply} \ tbl $. 
From table typing rule we know:
\begin{enumerate}
    \item $(\overline{e},\text{Cont}_{\mathrm{tbl}})= \tcv(tbl)$
    \item $\tscope \vdash {{e}_i} : {\type}_i$
    \item $\lbl = \bigsqcup_{i} \getlabel(\type_i)$
    \item $pc' = pc \sqcup \lbl $
    \item  $\forall (\phi_j, (a_j, \overline{\tau}_j)) \in \text{Cont}_{\mathrm{tbl}}. \
              (\tscope_{g_j},\tscope_{l_j})= \refine{\tscope}{\phi_j} \band \\
              (\stmt_j,(x_{j_1},\text{none}) \cons (x_{j_n}, \text{none})) = \tcv(a_j) \band \\
              \tscope_{a_j} =\{x_{j_i} \mapsto \type_{j_i}\} \band
              T, pc', (\tscope_{g_j},\tscope_{a_j}) \vdash \stmt_j : \Gamma_j$
    \item   $\Gamma' = \cup_j \{ (\tscope_{g_j}' , \tscope_{l_j}) | (\tscope_{g_j}',\tscope_{a_j}') \in \Gamma_j \}$
    \item   $\Gamma'' =  {\begin{cases}
                  \join(\Gamma') \hspace{15pt} \text{if } \lbl = {\high} \\
                  \Gamma'   \hspace{38pt}  \text{otherwise}
              \end{cases}}$
\end{enumerate}

In the following proof, we know that $(\tscope_g,\tscope_l) \vdash (\mem_g,\mem_l)$ and $\tcv  \vdash \stmtctxv$ hold.

We need to prove $\exists \mem' . \  {\stmtctx}: \sem{\mem}{\texttt{apply} \ tbl }{\mem'}$ and $\exists \tscope'' \in \Gamma''. \tscope'' \vdash \mem'$.

And from table reduction rule we can rewrite the goal to:
\begin{enumerate}
    \item  $ \exists \overline{e} \ sem_{tbl}. \  (\overline{e}, sem_{tbl}) = \stmtctxv(tbl) $
    \item  $ \exists (a, \overline{\val}) . \ sem_{tbl}(\expsem{(\mem_{g}, \mem_{l})}{e_{1}} \cons \expsem{(\mem_{g}, \mem_{l})}{e_{n}}) = (a, \overline{\val})$
    \item  $ \exists \stmt \  (x_{1} \cons x_{n}). \ (\stmt,(x_{1},\mathrm{none}) \cons (x_{n}, \mathrm{none})) = E(a)$
    \item  $ \exists \mem_a . \ \mem_a = \{x_{i} \mapsto {\val_i}\} $
    \item  $ \exists (\mem_{g'}, \mem_{a'}) . \  E:\sem{(\mem_{g}, \mem_{a})}{\stmt}{(\mem_{g'}, \mem_{a'})}$
    \item  $ \exists \tscope'' \in \Gamma''. \tscope'' \vdash (\mem_{g'}, \mem_{l})$
\end{enumerate}

In this proof, we will get an induction hypothesis for action call $a (\overline{\val})$ (formalized as a function call), call it \textsc{IH}.
\begin{align*}
	\forall \tcv \ \mem \ \tscope \ pc \ \Gamma . \ &\tcv, pc, \tscope \vdash a (\overline{\val}) : \Gamma \implies \\
	\tcv  \vdash &\stmtctxv \band \tscope \vdash \mem \implies \\
	\exists \mem' . &{\stmtctxv}: \sem{\mem}{a(\overline{\val})}{\mem'} \band \exists \tscope' \in \Gamma. \tscope' \vdash \mem'
\end{align*}

\textbf{Goal 1}:
Trivial, by the well-typedness condition $\tcv  \vdash \stmtctxv$, we know that if the table has a contract (from 1 in typing rule), then indeed there is semantics for it $sem_{tbl}$ and a key list $\overline{e}$ that matches the one in the table typing rule.

\textbf{Goal 2}:
From $\tcv  \vdash \stmtctxv$, we can deduct from condition $\text{tblWT } C \ X $ that indeed exists an action and value list pair $(a,\overline{\val})$ in the contract $\text{Cont}_\mathrm{tbl}$ correlated to a $\phi_j(\mem)$ that holds.

\textbf{Goal 3}:
From $\tcv  \vdash \stmtctxv$ we know indeed the actions's $a$ body and signature found in the semantics is the same found in the typing rule.

\textbf{Goal 4}:
Trivial, by setting $\mem_a$ to be $\{x_{i} \mapsto {\val_i}\}$.

\textbf{Goal 5}:
To prove this goal, we need to use \textsc{IH}. And in order to use \textsc{IH}, we must first show that $(\tscope_{g_j}, \tscope_{a_j}) \vdash (\mem_g,\mem_a)$ by proving $\tscope_{g_j} \vdash \mem_g$ where $(\tscope_{g_j},\tscope_{l_j})= \refine{\tscope}{\phi_j}$ and the copied in is well typed $\tscope_{a_j} \vdash \mem_a$.

Prove $\tscope_{g_j} \vdash \mem_g$:
Since initially given that $\tscope \vdash \mem$ i.e. $(\tscope_g,\tscope_l)\vdash(\mem_g,\mem_l)$,
also we know from $\text{tblWT } C \ X $ that there is $j$ such that the predicate $\phi_j$ is satisfied in $(\mem_g,\mem_l)$ i.e. $\phi_j(\mem_g,\mem_l)$, thus we can use Hyp \ref{hyp:interval_typedness_externs_and_tables_refinement} to deduct that $\refine{\tscope}{\phi_j} \vdash (\mem_g,\mem_l)$, i.e. we can infer that the refined state type is able to type the initial state. Given 5 in the typing rule, we know that $(\tscope_{g_j},\tscope_{l_j})= \refine{\tscope}{\phi_j}$ thus $(\tscope_{g_j},\tscope_{l_j}) \vdash (\mem_g,\mem_l)$, therefore $\tscope_{g_j} \vdash \mem_g$ holds.

Prove $\tscope_{a_j} \vdash \mem_a$:
This goal can be rewritten as $\{x_{i} \mapsto \type_{i}\} \vdash \{x_{i} \mapsto {\val_i}\}$. The proof is trivial by WF definition of tables we know that $\val_i:\type_i$, thus the variable $x_i$ is well typed.

Now, we can instantiate
\textsc{IH} to using $(\tcv,(\mem_g,\mem_a),(\tscope_{g_j}, \tscope_{a_j}), pc', \Gamma_j)$, so we can infer that exists $\mem'$ such that ${\stmtctxv}: \sem{\mem}{\stmt}{\mem'} $ (thus goal 5 is resolved).

\textbf{Goal 6}:
We can also infer from \text{IH} that exists $ \tscope' \in \Gamma_j $ such that $ \tscope' \vdash \mem'$. Let $(\tscope_{g_j}',\tscope_{a_j}') = \tscope'$ and $(\mem_g',\mem_{a}') = \mem'$, thus indeed $(\tscope_{g_j}',\tscope_{a_j}') \vdash (\mem_g',\mem_{a}')$ holds trivially. 

Line 6 of the typing rule iterates over each final state type set and collects the
modified global state and the refined local state, thus $(\tscope_{g_j}',\tscope_{l_j})$ is indeed in $\Gamma'$. Since we proved that $\tscope_{g_j}' \vdash \mem_g'$ holds in the previous step, and also proved $\tscope_{l_j} \vdash \mem_l$ in goal 5, therefore, $(\tscope_{g_j}',\tscope_{l_j}) \vdash (\mem_g',\mem_l)$.

Line 7 of the typing rule changes the labels but not the abstraction, thus the abstraction of the state type $(\tscope_{g_j}',\tscope_{l_j})$ in $\Gamma'$ indeed exists in $\Gamma''$ with labels changed so goal 6 holds.

\end{proof}

\subsection{Soundness of labeling}

\begin{theorem}
	\begin{align*}
		\forall \stmt \ \tc \ pc \ \mem_1 \ \mem_2 \ \mem_1' \ \mem_2' & \ \tscope_1 \ \tscope_2 \ \Gamma_1 \ \Gamma_2 \ E_1 \ E_2. \\
		\tc , pc, \tscope_1 \vdash \stmt : \Gamma_1 &\band \ \tc, pc,  \tscope_2 \vdash \stmt : \Gamma_2 \implies \\
		\tc \vdash \stmtctx_1 & \band \tc \vdash \stmtctx_2 \band \stmtctx_1 \lequiv{\tc} \stmtctx_2 \band \\
		\tscope_1 \vdash \mem_1 & \band \tscope_2 \vdash \mem_2 \band \mem_1 \lequiv{\tscope_1\sqcup\tscope_2} \mem_2 \band \\
		{\stmtctx_1}: \sem{\mem_1}{\stmt}{\mem_1'} & \band {\stmtctx_2}: \sem{\mem_2}{\stmt}{\mem_2'} \implies \\
		\big( \exists \tscope_1' \in \Gamma_1 &\band \tscope_2' \in \Gamma_2. \ \tscope_1' \vdash \mem_1' \\
		&\band \tscope_2' \vdash \mem_2' \band \mem_1' \lequiv{\tscope_1'\sqcup\tscope_2'} \mem_2' \big)
	\end{align*}
\end{theorem}

\begin{proof}
In this proof we assume that the program is well typed and does not get stuck according to HOL4P4 type system.
by induction on the typing tree of the program $stmt$ i.e. $\stmt$. 
Note that, initially $\tscope = (\tscope_g,\tscope_m)$ and $\mem=(m_g,m_l)$. 
In the following proof, we know that $\mem_1\lequiv{\tscope_1\sqcup\tscope_2}\mem_2$, we also know that $\tscope_1\vdash\mem_1$ and $\tscope_2\vdash\mem_2$.

\casesItem{Case assignment}: Here $stmt$ is $\lval := e$. 
From assignment typing rule we know:
\begin{enumerate}
    \item $\tscope_1 \vdash {e} : \type_1$
    \item $\type_1' = \raisel(\type_1, pc)$
    \item $\tscope_1' = \tscope_1[\lval \mapsto \type_1']$
    \item $\Gamma_1 = \{\tscope_1'\}$
    \item $\tscope_2 \vdash {e} : \type_2$
    \item $\type_2' = \raisel(\type_2, pc)$
    \item $\tscope_2' = \tscope_2[\lval \mapsto \type_2']$
    \item $\Gamma_2 = \{\tscope_2'\}$
\end{enumerate}

And from assignment reduction rule we know:
\begin{enumerate}
    \item $\expsem{\mem_1}{e}={\val_1}$
    \item $\mem_1'= \mem_1[\lval \mapsto \val_1]$
    \item $\expsem{\mem_2}{e}={\val_2}$
    \item $\mem_2'= \mem_2[\lval \mapsto \val_2]$
\end{enumerate}

Prove $\exists \tscope_1' \in \Gamma_1 \band \tscope_2' \in \Gamma_2. \
    \tscope_1' \vdash \mem_1' \band \tscope_2' \vdash \mem_2' \band \mem_1' \lequiv{\tscope_1'\sqcup\tscope_2'} \mem_2'$.
Since that assignment typing rule produces one state type (from 3,4,7,and 8), then from \textsc{Soundness of Abstraction}, we can infer $\tscope_1'\vdash\mem_1'$ and $\tscope_2'\vdash\mem_2'$, therefore the first two conjunctions of the goal holds.

Now, the final remaining goal to prove is that $\mem_1' \lequiv{\tscope_1'\sqcup\tscope_2'} \mem_2'$. This entails proving that all $\lval'$ in both $\mem_1'$ and $\mem_2'$ are low equivalent with respect to the least upper bound of the final state types that types the final states. 

Now we do cases on the label of $\lval'$ being {\high} or {\low} in $\tscope_1'\sqcup\tscope_2'$. 

\begin{enumerate}[label=\textbf{case}]
    \item label of $\lval'$ is {\high}:
          If the label is {\high}, i.e. $\tscope_1'\sqcup\tscope_2' \vdash \lval' : \type \band  \getlabel(\type) = {\high} $, then the property of labeling soundness holds after the assignment trivially. That's because soundness property checks the equality of the state type's low ranges only. 

    \item label of $\lval'$ is {\low}:
          If the label is {\low}, i.e., $\tscope_1' \sqcup \tscope_2' \vdash \lval' : \type \band \getlabel(\type) = {\low}$, this implies that in each state type $\tscope_1'$ and $\tscope_2'$ individually, the typing label of $\lval'$ is {\low}. 

          In this section, we conduct a case analysis on possible sub-cases relations between of $\lval$ being assigned and $\lval'$, thus we will have the following subcases: $\lval' \subsetneqq \lval$, $\lval \subsetneqq \lval'$, $\lval' = \lval$, $\lval' \subseteqq \lval$, and $\lval \subseteqq \lval'$. 

          \begin{enumerate}[label=\textbf{case}]
              \item $\lval' \subsetneqq \lval$ and $\lval \subsetneqq \lval'$: \\
                    Now for all $\lval'$ that are not equal to the $\lval$ we assign to, or not sub-lval of it:
                    our goal is to show $\mem_1'(\lval') = \mem_2'(\lval')$ by demonstrating that initially $\mem_1(\lval') = \mem_2(\lval')$ also holds.
                    We know that the assignment doesn't alter those parts of the states. Thus, the semantic update should keep the values of $\lval'$ the same, so $\mem_1(\lval') = \mem_1'(\lval')$ and $\mem_2(\lval') = \mem_2'(\lval')$. Likewise, the typing update should keep the type of $\lval'$ unchanged, so $\tscope_1'(\lval')  = \tscope_1(\lval')$ and $\tscope_2'(\lval')  = \tscope_2(\lval')$. This indeed mean that the labels of $\lval'$ were also {\low} in both $\tscope_1$ and $\tscope_2$, and given the assumption that $\mem_1\lequiv{\tscope_1\sqcup\tscope_2}\mem_2$, thus indeed $\mem_1(\lval')=\mem_2(\lval')$. 

              \item $\lval' = \lval$: \\
                    For $\lval' = \lval$, we know that in $\tscope_1'\sqcup\tscope_2'$ we type the $\lval$ as {\low}, i.e. $\tscope_1'\sqcup\tscope_2' (\lval) = \type \band \getlabel(\type) = {\low} $ thus the same property holds in individual state types $\tscope_1' (\lval) = \type' \band \getlabel(\type') = {\low} $ and also $\tscope_2' (\lval) = \type'' \band \getlabel(\type'') = {\low} $. We need to prove $\mem_1'(\lval) = \mem_2'(\lval)$.
                    For $\lval$ to be {\low} after the update function in either $\tscope_1'$ or $\tscope_2'$, it is necessary for the types of the expression $e$ to be {\low} in both initial state types $\tscope_1$ in 1 and $\tscope_2$ in 5 in typing rules, i.e., ($\getlabel(\type_1) = {\low}$ and $\getlabel(\type_2) = {\low}$).
                    This condition holds because otherwise, if the typing labels of $e$ were {\high}, then the goal would be trivially true (as the update would make $\lval$ {\high} in $\tscope_1'$ and $\tscope_2'$, contradicting the assumptions).
                    Since the typing label of $e$ is {\low} in 1 and 5 of typing rule, when reduced to a value in 1 and 3 of the semantics rule, this indicates that they reduce to the same value $\val_1=\val_2$ (using Lemma \ref{lemma:expression_evaluation_of_consistent_states}).  Since we update $\lval$ in $\mem_1$ and $\mem_2$ with the same value such that we produce $\mem_1'$ and $\mem_2'$ respectively, it follows that $\mem_1'(\lval) = \mem_2'(\lval)$. 

              \item $\lval' \subseteqq \lval$: \\
                    For $\lval' \subseteqq \lval$, we know that $\lval'$ can be a shorter variation of the $\lval$. The proof is the same as the previous case.

              \item $\lval \subseteqq \lval'$: \\
                    For $\lval \subseteqq \lval'$, we know that $\lval$ can be a shorter variation of the $\lval'$, thus the update of $\lval$ affects part of $\lval'$ type while the rest of it stays unchanged. Therefore, the proof is straightforward by conducing the same steps of the first two cases. 

          \end{enumerate}

          Given the last four subcases, we can now show that $\mem_1' \lequiv{\tscope'} \mem_2'$.

\end{enumerate}

\casesItem{Case condition}: Here $stmt$ is $\texttt{if } e \texttt{ then } \stmt_1 \texttt{ else } \stmt_2$. 
From conditional typing rule we know:
\begin{enumerate}
    \item $\tscope \vdash e : \type_1$
    \item $\lbl_1 = \getlabel(\type_1)$
    \item $pc_1 = pc \sqcup \lbl_1$
    \item $\tc, pc_1, (\refine{\tscope_1}{e}     ) \vdash \stmt_1 : \Gamma_1$
    \item $\tc, pc_1, (\refine{\tscope_1}{\neg e}) \vdash \stmt_2 : \Gamma_2$
    \item $\Gamma_3 =  {\begin{cases}
                  \text{join}(\Gamma_1 \cup \Gamma_2) \hspace{15pt} \text{if } \lbl_1 = {\high} \\
                  \Gamma_1 \cup \Gamma_2  \hspace{38pt}  \text{otherwise}
              \end{cases}}$
    \item $\tscope \vdash e : \type_2$
    \item $\lbl_2 = \getlabel(\type_2)$
    \item $pc_2 = pc \sqcup \lbl_2$
    \item $\tc, pc_2, (\refine{\tscope_2}{e}     ) \vdash \stmt_1 : \Gamma_4$
    \item $\tc, pc_2, (\refine{\tscope_2}{\neg e}) \vdash \stmt_2 : \Gamma_5$
    \item $\Gamma_6 =  {\begin{cases}
                  \text{join}(\Gamma_4 \cup \Gamma_5) \hspace{15pt} \text{if } \lbl_2 = {\high} \\
                  \Gamma_4 \cup \Gamma_5  \hspace{38pt}  \text{otherwise}
              \end{cases}}$
\end{enumerate}

We also know that both initial states are $\tscope_1 \vdash \mem_1$ and $\tscope_2 \vdash \mem_2$ and also $\mem_1 \lequiv{\tscope_1\sqcup\tscope_2} \mem_2$. And we know that the conditional statement is executed with $\mem_1$ and $\mem_2$ resulting $\mem_1'$ and $\mem_2'$ consequently.

In addition to that, we get induction hypothesis for $\stmt_1$ \textsc{IH1} and $\stmt_2$ \textsc{IH2} (we only show \textsc{IH1}):
\begin{align*}
	\forall \ \tc \ pc \ \mem_a \ \mem_b \ \mem_a' \ & \mem_b' \ \tscope_a \ \tscope_b \ \Gamma_a \ \Gamma_b \ \stmtctx_a \ \stmtctx_b. \\
	\ \tc, pc, \tscope_a \vdash \stmt_1 : \Gamma_a &\band \ \tc, pc, \tscope_b \vdash \stmt_1 : \Gamma_b
	\implies \\
	\ \tc \vdash \stmtctx_a \band \tc \vdash \stmtctx_b &\band \stmtctx_a \lequiv{\tc} \stmtctx_b
	\band \\
	\tscope_a \vdash \mem_a
	\band
	\tscope_b \vdash \mem_b
	&\band
	\mem_a \lequiv{\tscope_a\sqcup\tscope_b} \mem_b
	\band \\
	{\stmtctx_a}: \sem{\mem_a}{\stmt_1}{\mem_a'}
	&\band
	{\stmtctx_b}: \sem{\mem_b}{\stmt_1}{\mem_b'} \\
	&\implies \\
	\big(
	\exists \tscope_a' \in \Gamma_a \band \tscope_b' \in \Gamma_b. \ &
	\tscope_a' \vdash \mem_a' \band \tscope_b' \vdash \mem_b' \band \mem_a' \lequiv{\tscope_a'\sqcup\tscope_b'} \mem_b' \big)
\end{align*}

We start by cases on labels $\lbl_1$ and $\lbl_2$ of $e$.

\begin{enumerate}[label=\textbf{case}]
    \item $\lbl_1 = \lbl_2 = {\low}$: We need to prove that $\exists \tscope_1'\in\Gamma_3$ and $\exists \tscope_2'\in\Gamma_6$ it holds $\mem_1' \lequiv{\tscope_1'\sqcup\tscope_2'} \mem_2'$. We can directly use Lemma \ref{lemma:expression_evaluation_of_consistent_states} to infer that e is evaluation is indistinguishable in the states, and since $e$ is assumed to be typed as boolean, thus we get two subcases where: $\expsem{\mem_1}{e}={true}$ and $\expsem{\mem_2}{e}={true}$ , or $\expsem{\mem_1}{e}={false}$ and $\expsem{\mem_2}{e}={false}$.

          \begin{enumerate}[label=\textbf{case}]
              \item $\expsem{\mem_1}{e}={true}$ and $\expsem{\mem_2}{e}={true}$:
                    when looking into the reduction rule
                    of the both if statements in the assumption, we only reduce the first branch of each. Hence, we have ${\stmtctx_1}: \sem{\mem_1}{\stmt_1}{\mem_1'}$ and ${\stmtctx_2}: \sem{\mem_2}{\stmt_1}{\mem_2'}$.

                    From Hyp \ref{hyp:label_typedness_boolean_expressions_refinement}, we can show that $\mem_1 \lequiv{(\refine{\tscope_1}{e})\sqcup(\refine{\tscope_2}{e})} \mem_2$.
                    Additionally, we can infer that $(\refine{\tscope_1}{e})\vdash\mem_1$ and $(\refine{\tscope_2}{e})\vdash\mem_2$ using Hyp \ref{hyp:interval_typedness_boolean_expressions_refinement}.

                    Now we can directly instantiate and apply \textsc{IH1} using the following $(\tc ,pc\sqcup {\low}, \mem_1, \mem_2, \mem_1', \mem_2',$ $(\refine{\tscope_1}{e}), (\refine{\tscope_2}{e}),$ $\Gamma_1, \Gamma_4, \stmtctx_1, \stmtctx_2)$
                    to infer that the states after executing $\stmt_1$ are low equivalent, i.e.
                    exists $\tscope_1'' \in \Gamma_1$ and exists $\tscope_2'' \in \Gamma_4$ such that $\tscope_1'' \vdash \mem_1'$ and $\tscope_2'' \vdash \mem_2'$ and  $\mem_1' \lequiv{\tscope_1''\sqcup\tscope_2''} \mem_2'$. Since $\Gamma_1\subseteq\Gamma_3$ and $\Gamma_4\subseteq\Gamma_6$, thus the goal holds.

              \item $\expsem{\mem_1}{e}={false}$ and $\expsem{\mem_2}{e}={false}$ same proof as the previous case.
          \end{enumerate}

    \item $\lbl_2 = {\high}$:
          We initiate the proof by fixing $\lbl_2$ to be {\high}, and the value of $e$ to be reduced to $false$, thus it executes $\stmt_2$ (starting from configuration $\mem_2$, and yields $\mem_2'$, note that if $e$ reduces to $true$ the proof is identical as this case). In this proof, we refer to these as the second configuration.

          Now, consider the following scenario where we start from $\mem_1$ in the
          semantics rule and $\tscope_1$ in typing rule, we generalize the proof for any boolean expression $e_i$ such that $i$ ranges over $true$ and $false$, where $e_{true}$ is $e$, $e_{false}$ is $\neg e$, $\stmt_{true}$ is the first branch $\stmt_1$, and $\stmt_{false}$ is the second branch $\stmt_2$. In this sub-case of the proof, $\lbl_1$ denotes the label associated $e_i$'s typing label, and let the refinement of the initial typing scope $\tscope_1$ to be represented as $\refine{\tscope_1}{e_i}$. Suppose the executed branch is $\stmt_i$, yielding a final set of state types denoted as $\Gamma_i$. In this proof, we refer to these as the first configuration.

          Given the previous generalizations, we can rewrite the assumptions to:
          \begin{enumerate}[label=(\alph*)]
              \item $\tc, \lbl_1, (\refine{\tscope_1}{e_i}) \vdash \stmt_i : \Gamma_i$
              \item ${\stmtctx_1}: \sem{\mem_1}{\stmt_i}{\mem_1'}$
              \item $\tc, {\high}, (\refine{\tscope_2}{\neg e}) \vdash \stmt_2 : \Gamma_5$
              \item $\tc, {\high}, (\refine{\tscope_2}{e}) \vdash \stmt_1 : \Gamma_4$
              \item ${\stmtctx_2}: \sem{\mem_2}{\stmt_2}{\mem_2'}$
          \end{enumerate}

          What we aim to prove is the existence of $\tscope_1' \in \Gamma_3$
          such that $\tscope_1' \vdash \mem_1'$. Additionally, we need to establish the existence of $\tscope_2' \in \Gamma_6$ where $\Gamma_6 = \text{join}(\Gamma_4 \cup \Gamma_5)$, such that $\tscope_2' \vdash \mem_2'$. Furthermore, we must prove that $\mem_1' \lequiv{\tscope_1' \sqcup \tscope_2'} \mem_2'$.

          From the \textsc{Soundness of Abstraction}, we know that for the second configuration
          indeed exists $\tscope_2''\in\Gamma_5$ such that it types $\mem_2'$ (i.e. $\tscope_2''\vdash\mem_2'$).

          Given that $\tscope_2'' \in \Gamma_5$ and
          $\Gamma_6 = \text{join}(\Gamma_4 \cup \Gamma_5)$, we can deduce (by Lemma \ref{lemma:join_does_not_modify_the_intervals}) the existence of $\overline{\tscope_2''} \in \text{join}(\Gamma_4 \cup \Gamma_5)$ such that it is more restrictive than $\tscope_2''$, denoted as $\tscope_2'' \sqsubseteq \overline{\tscope_2''}$. Using the same lemma, we conclude that $\overline{\tscope_2''} \vdash \mem_2'$. Now on, we choose $\overline{\tscope_2''}$ to be used in the proof and resolve the second conjunction of the goal.

          From the \textsc{Soundness of Abstraction}, we know that for the first configuration
          indeed exists $\tscope_i''\in\Gamma_i$ such that it types $\mem_1'$ (i.e. $\tscope_i''\vdash\mem_1'$).

          In the first configuration, we generalized the proof according to the evaluation of
          $e_i$. Consequently, the final state type set $\Gamma_3$ can be either a union (if
          $\lbl_1 = {\low}$) or a join (if $\lbl_1 = {\high}$) of all final state type sets $\forall i \leq 1 . \ \Gamma_i$ resulting from typing their corresponding $\stmt_i$.
          In either case (union or join), we can establish the existence of
          $\overline {\tscope_i''} \in \Gamma_3$ such that $\tscope_i'' \sqsubseteq \overline{\tscope_i''}$ and indeed $\overline{\tscope_i''} \vdash \mem_1'$. Note that if $\Gamma_3$ resulted from a join, we infer this using Lemma \ref{lemma:join_does_not_modify_the_intervals}; otherwise, if it resulted from a union, it is trivially true. In fact, we can directly choose $\tscope_i'' = \overline{\tscope_i''}$ when union the final state types sets.

          Next, we proceed to implement cases based on whether an $\lval$'s type label is {\high} or {\low}.

          \begin{enumerate}[label=\textbf{case}]
              \item $(\overline{\tscope_i''}\sqcup\overline{\tscope_2''}(\lval) = \type) \band \getlabel(\type) =  {\high}$: the goal holds trivially.

              \item $(\overline{\tscope_i''}\sqcup\overline{\tscope_2''}(\lval) = \type) \band \getlabel(\type) =  {\low} $:
                    this case entails that each state type individually holds $\overline{\tscope_i''}(\lval) = \type_1' \band \getlabel(\type_1') =  {\low}$ and also $\overline{\tscope_2''}(\lval) = \type_2'' \band \getlabel(\type_2'') =  {\low}$. \\

                    Given that the $\lval$'s type is {\low} in $\overline{\tscope_2''}$, and
                    considering $\overline{\tscope_2''} \in \text{join}(\Gamma_4 \cup \Gamma_5)$, it follows that the $\lval$ is also {\low} in any state type within $\text{join}(\Gamma_4 \cup \Gamma_5)$. Consequently, the $\lval$ is {\low} in the state types of both $\Gamma_5$ and $\Gamma_4$ (if not empty) individually before the join operation. Hence, since $\tscope_2'' \in \Gamma_5$, it implies that $\lval$ is also {\low} in $\tscope_2''$  expressed as $\tscope_2''(\lval) = \type_2''' \band \getlabel(\type_2''') = {\low}$ (using Lemma \ref{lemma:joins_low_label_implication}). We also know that typing a statement in a {\high} context in (d) entails that the final $\Gamma_4$ is not empty (using Lemma \ref{lemma:branch_on_high_is_never_empty}), thus $\lval$'s type label is also {\low} in all the state types in $\Gamma_4$.

                    In the second configuration, we type the statement $\stmt_2$ with a
                    {\high} context in (c), and $\stmt_2$ reduces to $\mem_2'$ in (e). Furthermore, from the previous step, we inferred that the $\lval$'s type is {\low} in $\tscope_2''$. Hence, we can use Lemma \ref{lemma:branch_on_high_state_preservation} to infer that the initial and final states remain unchanged for {\low} lvalues, which means $\mem_2(\lval) = \mem_2'(\lval)$.
                    Then we can use Lemma \ref{lemma:high_program_final_types_subset_of_initial} to infer that $\refine{\tscope_2}{\neg e} \sqsubseteq \tscope_2''$, this entails that the $\lval$'s type label is indeed {\low} in the refined state $\refine{\tscope_2}{\neg e}$.
                    It is easy to see that $\tscope_2\sqsubseteq\refine{\tscope_2}{\neg e}$, thus $\lval$'s type is also {\low} in the initial state type $\tscope_2$, i.e. $\tscope_2(\lval) = \type_2' \band \getlabel(\type_2') = {\low}$.

                    For the first configuration, in this sub-case, we have
                    $\overline{\tscope_i''}(\lval)= \type_1' \band \getlabel(\type_1') = {\low}$, and $\overline{\tscope_i''} \in \Gamma_3$.
                    We previously showed $\tscope_i''\in\Gamma_i$ and $\tscope_i'' \sqsubseteq \overline{\tscope_i''}$, where $\Gamma_i$ such that  is the final state type of typing $\stmt_i$ according to the evaluation of $e_i$.
                    Since $\lval$'s typing label is {\low} in $\overline{\tscope_i''}$ in $\Gamma_3$ and we know that the state types in $\overline{\tscope_i''}$ are more restrictive than the state types in $\tscope_i''$, we can conclude that $\tscope_i'' \in \Gamma_i$ also types $\lval$ as {\low} $\tscope_i''(\lval) = \type_1'' \band \getlabel(\type_1'') = {\low}$.

                    Previously, we demonstrated that the $\lval$'s typing label is
                    {\low} in all final state types in $\Gamma_4$ and $\Gamma_5$, and we showed that $\Gamma_4$ is not empty.
                    Consequently, neither $\stmt_1$ nor $\stmt_2$ can modify $\lval$. Considering the assumptions (a) and (b) (related to the first configuration), where $\stmt_i$ can be either $\stmt_1$ or $\stmt_2$, we conclude that the $\lval$ remains unchanged there as well. Here, we can apply Lemma \ref{lemma:branch_on_high_state_lemma} to deduce that the $\lval$'s typing label in
                    $\tscope_i'' \in \Gamma_i$ are more restrictive than the one we find in the refined state $\refine{\tscope_1}{e_i}$, thus $\refine{\tscope_1}{e_i}(\lval) = \type_i' \band \getlabel(\type_i') = {\low}$.
                    Now we can apply Lemma \ref{lemma:branch_on_high_state_preservation} for any $\stmt_i$ to show that $\mem_1(\lval)=\mem_1'(\lval)$.

                    Finally, since the typing label of $\lval$ in $\refine{\tscope_1}{e_i}$ is {\low}, then trivially we know that the typing label of $\lval$ in ${\tscope_1}$ is also {\low} because $\tscope_1 \sqsubseteq \refine{\tscope_1}{e_i}$. We previously showed that $\lval$'s typing label is {\low} in $\tscope_2$,
                    thus we now can show that $\tscope_1\sqcup\tscope_2(\lval) = \type' \band \getlabel(\type') = {\low}$.
                    Now, we can deduce that $\mem_1(\lval)=\mem_2(\lval)$ from the definition of the assumption $\mem_1\lequiv{\tscope_1\sqcup\tscope_2}\mem_2$. Thus, the goal holds.
          \end{enumerate}

    \item $\lbl_1 = {\high}$: when fixing the first configuration, we implement same proof as previous case.
\end{enumerate}

\casesItem{Case sequence}: Here $stmt$ is $\stmt_1;\stmt_2$. 
From sequence typing rule we know:
\begin{enumerate}
    \item $\tc, pc, \tscope_1 \vdash {\stmt_1} : \Gamma_1 $
    \item $\forall \tscope_1' \in \Gamma_1 . \ \tc, pc, \tscope_1' \vdash \stmt_2 :\Gamma_2^{\tscope_1'}$
    \item $\Gamma' =  \bigcup_{\tscope_1' \in \Gamma_1} \Gamma_2^{\tscope_1'} $
    \item $\tc, pc, \tscope_2 \vdash {\stmt_1} : \Gamma_3 $
    \item $\forall \tscope_2' \in \Gamma_3 . \ \tc, pc, \tscope_2' \vdash \stmt_2 :\Gamma_4^{\tscope_2'}$
    \item $\Gamma'' =  \bigcup_{\tscope_2' \in \Gamma_3} \Gamma_4^{\tscope_2'}$
\end{enumerate}

And from sequence reduction rule we know:
\begin{enumerate}
    \item ${\stmtctx}_1: \sem{\mem_1}{\stmt_1}{\mem_1'}$
    \item ${\stmtctx}_1: \sem{\mem_1'}{\stmt_2}{\mem_1''}$
    \item ${\stmtctx}_2: \sem{\mem_2}{\stmt_1}{\mem_2'}$
    \item ${\stmtctx}_2: \sem{\mem_2'}{\stmt_2}{\mem_2''}$
\end{enumerate}

We initially know that : $\tscope_1 \vdash \mem_1$, $\tscope_2 \vdash \mem_2$, and $\mem_1 \lequiv{\tscope_1\sqcup\tscope_2} \mem_2$. 

In addition to that, we get induction hypothesis for $\stmt_1$ \textsc{IH1} and $\stmt_2$ \textsc{IH2} (we only show \textsc{IH1}):
\begin{align*}
	\forall \ \tc \ pc \ \mem_a \ \mem_b \ \mem_a' \ & \mem_b' \ \tscope_a \ \tscope_b \ \Gamma_a \ \Gamma_b \ \stmtctx_a \ \stmtctx_b. \\
	\ \tc, pc, \tscope_a \vdash \stmt_1 : \Gamma_a &\band \ \tc, pc, \tscope_b \vdash \stmt_1 : \Gamma_b
	\ \implies \  \\
	\ \tc \vdash \stmtctx_a &\band \tc \vdash \stmtctx_b \band \stmtctx_a \lequiv{\tc} \stmtctx_b
	\band \\
	\tscope_a \vdash \mem_a
	&\band
	\tscope_b \vdash \mem_b
	\band
	\mem_a \lequiv{\tscope_a\sqcup\tscope_b} \mem_b
	\band \\
	{\stmtctx_a}: \sem{\mem_a}{\stmt_1}{\mem_a'}
	&\band
	{\stmtctx_b}: \sem{\mem_b}{\stmt_1}{\mem_b'} \\
	&\implies \\
	\big( \exists \tscope_a' \in \Gamma_a \band \tscope_b' \in \Gamma_b. &\
	\tscope_a' \vdash \mem_a' \band \tscope_b' \vdash \mem_b' \band \mem_a' \lequiv{\tscope_a'\sqcup\tscope_b'} \mem_b' \big)
\end{align*}

We need to prove there are two state types $\tscope_1''\in\Gamma'$ and $\tscope_2''\in\Gamma''$
such they type the final states $\tscope_1''\vdash\mem_1''$ and $\tscope_2''\vdash\mem_2''$ and indeed $\mem_1'' \lequiv{\tscope_1''\sqcup\tscope_2''} \mem_2''$ holds. 

We start by using \textsc{IH1}, and instantiating it with
$(\tc , pc, \mem_1,\mem_2,\mem_1',\mem_2',\tscope_1,\tscope_2,\Gamma_1,\Gamma_3, \stmtctx_1, \stmtctx_2)$ to infer that there exists $\tscope_1'\in\Gamma_1$ and $\tscope_2'\in\Gamma_3$ such that $\tscope_1'\vdash\mem_1' \band \tscope_2'\vdash\mem_2'$ and also $\mem_1'\lequiv{\tscope_1'\sqcup\tscope_2'}\mem_2'$. 

Then, in the typing rule, we instantiate 2 with $\tscope_1'$ and 5 with $\tscope_2'$.
Now we can use \textsc{IH2}, and instantiating it with $(\tc , pc, \mem_1',\mem_2',\mem_1'',\mem_2'',\tscope_1',\tscope_2',\Gamma_2^{\tscope_1'},\Gamma_4^{\tscope_2'}, \stmtctx_1, \stmtctx_2)$. From that we can infer that indeed there exists state types $\tscope_1''\in\Gamma_2^{\tscope_1'}$ and $\tscope_2''\in\Gamma_4^{\tscope_2'}$ such that they type the final states $\tscope_1''\vdash\mem_1''$ and $\tscope_2''\vdash\mem_2''$, where they keep the states low equivalent as $\mem_1''\lequiv{\tscope_1''\sqcup\tscope_2''}\mem_2''$.

We know that the final set of state type of interest is simply the union of all state types that can type the second statement in 3 and 6.  It is easy to see that since $\tscope_1''\in\Gamma_2^{\tscope_1'}$ then $\tscope_1''\in\Gamma'$. Similarly, $\tscope_2''\in\Gamma_4^{\tscope_2'}$ then $\tscope_2''\in\Gamma''$. Thus, the goal is proven.

\casesItem{Case function call}: Here $stmt$ is $f(e_{1} \cons e_{n})$. 
From call typing rule we know (note that here we explicitly write the global and local state type):

\begin{enumerate}
    \item $(\tscope_{g1}, \tscope_{l1}) \vdash  e_{i} : \type_{i1} $
    \item $(\stmt,(x_{1},d_{1}) \cons (x_{n}, d_{n})) = \tcv(f) $
    \item $\tscope_{f1} =\{x_{i} \mapsto \type_{i1}\}$
    \item $ \tcv, pc, (\tscope_{g1}, \tscope_{f1}) \vdash \stmt : \Gamma_1' $
    \item $\Gamma_1'' =
              \{(\tscope'_{g1}, \tscope_{l1})  [e_{i} \mapsto \tscope'_{f1}(x_{i}) \mid \isOut(d_{i})]
              \mid (\tscope'_{g1}, \tscope'_{f1}) \in \Gamma_1' \}$
    \item $(\tscope_{g2}, \tscope_{l2}) \vdash  e_{i} : \type_{i2} $
    \item $\tscope_{f2} =\{x_{i} \mapsto \type_{i2}\}$
    \item $ \tcv, pc, (\tscope_{g2}, \tscope_{f2}) \vdash \stmt : \Gamma_2' $
    \item $\Gamma_2'' =
              \{(\tscope'_{g2}, \tscope_{l2})  [e_{i} \mapsto \tscope'_{f2}(x_{i}) \mid \isOut(d_{i})]
              \mid (\tscope'_{g2}, \tscope'_{f2}) \in \Gamma_2' \}$
\end{enumerate}

And from call reduction rule we know:
\begin{enumerate}
    \item $(\stmt,(x_{1},d_{1}) \cons (x_{n}, d_{n})) = \stmtctxvn{1}(f) $
    \item $\mem_{f1} =\{x_{i} \mapsto \expsem{(\mem_{g1}, \mem_{l1})}{e_{i}}\}$
    \item ${\stmtctxvn{1}}: \sem{(\mem_{g1}, \mem_{f1})}{s}{(\mem'_{g1}, \mem'_{f1})}$
    \item $ \mem_1'' = (\mem'_{g1}, \mem_{l1})[e_{i} \mapsto \mem'_{f1}(x_{i}) \mid \isOut(d_{i})]$
    \item $(\stmt,(x_{1},d_{1}) \cons (x_{n}, d_{n})) = \stmtctxvn{2}(f) $
    \item $\mem_{f2} =\{x_{i} \mapsto \expsem{(\mem_{g2}, \mem_{l2})}{e_{i}}\}$
    \item ${\stmtctxvn{2}}: \sem{(\mem_{g2}, \mem_{f2})}{s}{(\mem'_{g2}, \mem'_{f2})}$
    \item $ \mem_2'' = (\mem'_{g2}, \mem_{l2})[e_{i} \mapsto \mem'_{f2}(x_{i}) \mid \isOut(d_{i})]$
\end{enumerate}

Let $\mem_1= (\mem_{g1}, \mem_{l1})$, $\mem_2= (\mem_{g2}, \mem_{l2})$, let  $\tscope_1 = (\tscope_{g1}, \tscope_{l1})$, and  $\tscope_1 = (\tscope_{g2}, \tscope_{l2})$.
We initially know that: $\tscope_1 \vdash \mem_1$, $\tscope_2 \vdash \mem_2$, and $\mem_1 \lequiv{\tscope_1\sqcup\tscope_2} \mem_2$.

In addition to that, we get induction hypothesis for $\stmt$ \textsc{IH}:
\begin{align*}
	\forall \ \tc \ pc \  \mem_a \ \mem_b \ \mem_a' \ &\mem_b' \ \tscope_a \ \tscope_b \ \Gamma_a \ \Gamma_b \ \stmtctx_a \ \stmtctx_b . \\
	\ \tc, pc,  \tscope_a \vdash \stmt : \Gamma_a &\band \ \tc, pc, \tscope_b \vdash \stmt : \Gamma_b
	\implies \\
	\tc \vdash \stmtctx_a \band \tc \vdash \stmtctx_b &\band \stmtctx_a \lequiv{\tc} \stmtctx_b
	\band \\
	\tscope_a \vdash \mem_a
	\band
	\tscope_b \vdash \mem_b
	&\band
	\mem_a \lequiv{\tscope_a\sqcup\tscope_b} \mem_b
	\band \\
	{\stmtctx_a}: \sem{\mem_a}{\stmt}{\mem_a'}
	&\band
	{\stmtctx_b}: \sem{\mem_b}{\stmt}{\mem_b'} \\
	&\implies \\
	\big( \exists \tscope_a' \in \Gamma_a \band \tscope_b' \in \Gamma_b. \
	\tscope_a' \vdash \mem_a' &\band \tscope_b' \vdash \mem_b' \band \mem_a' \lequiv{\tscope_a'\sqcup\tscope_b'} \mem_b' \big)
\end{align*}

We need to prove there are two state types $\tscope_1''\in\Gamma_1''$ and $\tscope_2''\in\Gamma_2''$
such they type the final states $\tscope_1''\vdash\mem_1''$ and $\tscope_2''\vdash\mem_2''$ and indeed $\mem_1'' \lequiv{\tscope_1''\sqcup\tscope_2''} \mem_2''$ holds. 

First we need to prove that the resulted copy-in map is also well-typed i.e. $\tscope_{f1}\vdash\mem_{f1}$. Note that the same proof applies to prove $\tscope_{f2}\vdash\mem_{f2}$:

Given that initially
the typing state types the state $(\tscope_{g1},\tscope_{l1}) \vdash (\mem_{g1},\mem_{l1})$ from assumptions and given that the expressions $e_i$ have a type $\type_{i1}$ from typing rule 1, now we can use Lemma \ref{lemma:expression_reduction_preserves_type} to infer that for all $\val_i : \type_i$ such that $\val_i$ is the evaluation of $\expsem{(\mem_{g1}, \mem_{l1})}{e_{i}}$. This leads us to trivially infer that $\forall i . \ \{x_{i} \mapsto \type_{i}\} \vdash \{x_{i} \mapsto \val_i \}$ holds, thus $\tscope_{f1}\vdash\mem_{f1}$. \\

Now we want to show that $(\mem_{g1},\mem_{f1})\lequiv{(\tscope_{g1},\tscope_{f1})\sqcup(\tscope_{g2},\tscope_{f2})}(\mem_{g2},\mem_{f2})$.
This can be broken down and rewritten into two goals according to Lemma \ref{lemma:low_equivalence_distribution}:
\begin{enumerate}[label=\textbf{Goal}]
    \item \textbf{1.} $\mem_{g1}\lequiv{\tscope_{g1} \sqcup \tscope_{g2}} \mem_{g2}$: \\
          We know that the domains of $\mem_{g1}$ and $\mem_{l1}$
          are distinct and do not intersect (similarly for $\mem_{g2}$ and $\mem_{l2}$),
          and given that they are initially low equivalent with respect to
          $(\tscope_{g1}, \tscope_{l1})\sqcup(\tscope_{g2}, \tscope_{l2})$ as
          $(\mem_{g1},\mem_{l1}) \lequiv{(\tscope_{g1}, \tscope_{l1})\sqcup(\tscope_{g2}, \tscope_{l2})} (\mem_{g2},\mem_{l2})$.
          Then we can use Lemma \ref{lemma:low_equivalence_distribution} to show that the goal holds.
    \item \textbf{2.} $\mem_{f1}\lequiv{\tscope_{f1} \sqcup \tscope_{f2}} \mem_{f2}$: \\
          First we start by rewriting the goal as following: \\
          $\{x_{i} \mapsto \expsem{(\mem_{g1}, \mem_{l1})}{e_{i}}\}
              \lequiv{\{x_{i} \mapsto (\tscope_{g1}, \tscope_{l1})(e_i)\} \sqcup \{x_{i} \mapsto (\tscope_{g2}, \tscope_{l2})(e_i)\}}
              \{x_{i} \mapsto \expsem{(\mem_{g2}, \mem_{l2})}{e_{i}}\}$. It is easy to see that the empty map is low equivalent as: $\{\} \lequiv{\{\} \sqcup \{\} } \{\}$.
          Now we can use Lemma \ref{lemma:low_equivalence_update} to show that the goal holds.
\end{enumerate}

Now we can use \textsc{IH}, and instantiate it with:
$(\tcv,pc, (\mem_{g1},\mem_{f1}) ,(\mem_{g2},\mem_{f2}),(\mem_{g1}',\mem_{f1}'),(\mem_{g2}',\mem_{f2}'),$ $(\tscope_{g1},\tscope_{f1}), (\tscope_{g2},\tscope_{f2}), \Gamma_1', \Gamma_2', \stmtctxvn{1}, \stmtctxvn{2})$,
so that we can deduct that exists $\overline{\tscope_1'} \in \Gamma_1'$ such that $\overline{\tscope_1'}\vdash(\mem_{g1}',\mem_{f1}')$, also exists $\overline{\tscope_2'} \in \Gamma_2'$ such that $\overline{\tscope_2'}\vdash(\mem_{g2}',\mem_{f2}')$.
We can also infer that $(\mem_{g1}',\mem_{f1}')\lequiv{\overline{\tscope_1'}\sqcup\overline{\tscope_2'}}(\mem_{g2}',\mem_{f2}')$.

Let $\overline{\tscope_1'} = (\overline{\tscope_{g1}'}, \overline{\tscope_{f1}'} )$
and $\overline{\tscope_2'} = (\overline{\tscope_{g2}'}, \overline{\tscope_{f2}'} )$ then we can also conclude from the previous:
\begin{enumerate}[label=(\alph*)]
    \item $\overline{\tscope_{g1}'} \vdash \mem_{g1}'$
    \item $\overline{\tscope_{g2}'} \vdash \mem_{g2}'$
    \item $\overline{\tscope_{f1}'} \vdash \mem_{f1}'$
    \item $\overline{\tscope_{f2}'} \vdash \mem_{f2}'$
    \item $\mem_{g1}'\lequiv{\overline{\tscope_{g1}'}\sqcup\overline{\tscope_{g2}'}}\mem_{g2}'$
    \item $\mem_{f1}'\lequiv{\overline{\tscope_{f1}'}\sqcup\overline{\tscope_{f2}'}}\mem_{f2}'$
\end{enumerate}

Since the final goal is to prove there are two state types $\tscope_1''\in\Gamma_1''$ and $\tscope_2''\in\Gamma_2''$
such they type the final states $\tscope_1''\vdash\mem_1''$ and $\tscope_2''\vdash\mem_2''$ and indeed
$\mem_1'' \lequiv{\tscope_1''\sqcup\tscope_2''} \mem_2''$ holds, then we can select $\tscope_1''$ to be
$(\overline{\tscope_{g1}'}, \overline{\tscope_{f1}'} )$, i.e. the copied-out map is
$(\overline{\tscope_{g1}'}, \tscope_{l1})  [e_{i} \mapsto \overline{\tscope_{f1}'}(x_{i}) \mid \isOut(d_{i})]$.

Similarly, we can select $\tscope_2''$ to be
$(\overline{\tscope_{g2}'}, \overline{\tscope_{f2}'} )$,  i.e. the copied-out map is
$(\overline{\tscope_{g2}'}, \tscope_{l2})  [e_{i} \mapsto \overline{\tscope_{f2}'}(x_{i}) \mid \isOut(d_{i})]$. \\

\begin{enumerate}[label=\textbf{Goal}]
    \item \textbf{1:} 
    \begin{align*}
    	(\overline{\tscope_{g1}'}, \tscope_{l1}) [e_{i} \mapsto \overline{\tscope_{f1}'}(x_{i})] \vdash (\mem'_{g1}, \mem_{l1})[e_{i} \mapsto \mem'_{f1}(x_{i})  ]
    \end{align*}
          Since we know $\overline{\tscope_{g1}'} \vdash \mem_{g1}'$ ((a) from previous step) and ${\tscope_{l1}} \vdash \mem_{l1}$ from assumptions rewrites,
          this entails that $(\overline{\tscope_{g1}'}, \tscope_{l1}) \vdash (\mem'_{g1}, \mem_{l1})$.
          We also proved that $\overline{\tscope_{f1}'} \vdash \mem_{f1}'$ ((c) from previous step), thus for all $x \in \text{domain}(\tscope_{f1}')$ holds $\mem_{f1}'(x) : \overline{\tscope_{f1}'}(x)$. Now we can use Lemma \ref{lemma:lval updates_preserves_type} to show that the goal holds.

    \item \textbf{2:} 
    \begin{align*}
    	(\overline{\tscope_{g2}'}, \tscope_{l2})  [e_{i} \mapsto \overline{\tscope_{f2}'}(x_{i})  ] \vdash (\mem'_{g2}, \mem_{l2})[e_{i} \mapsto \mem'_{f2}(x_{i})  ]
    \end{align*}
          Same proof as the previous sub goal, using (b) and (d) the previous step, and the initial assumption $\tscope_{l2}\vdash\mem_{l2}$.

    \item \textbf{3:} 
    \begin{align*}
    	&(\mem'_{g1}, \mem_{l1})[e_{i} \mapsto \mem'_{f1}(x_{i})  ] \\
    	&\lequiv{((\overline{\tscope'_{g1}}, \tscope_{l1})[e_{i} \mapsto \tscope'_{f1}(x_{i})  ]) \sqcup ((\overline{\tscope_{g2}'}, \tscope_{l2})  [e_{i} \mapsto \overline{\tscope_{f2}'}(x_{i})  ])} \\
    	&(\mem'_{g2}, \mem_{l2})[e_{i} \mapsto \mem'_{f2}(x_{i})  ] 
    \end{align*}

          We know that $\mem_{g1}' \lequiv{\overline{\tscope_{g1}'}\sqcup\overline{\tscope_{g2}'}}\mem_{g2}'$ (from (e) of previous step),
          we also know that $\mem_{l1}' \lequiv{\overline{\tscope_{l1}}\sqcup\overline{\tscope_{l2}}}\mem_{l2}$ (from assumptions), now using Lemma \ref{lemma:low_equivalence_distribution}
          we can combine them to infer an equivalence before copying out or (no copy-out because there are not out directed parameters), i.e. :

          ${(\mem'_{g1}, \mem_{l1})}
              \lequiv{(\overline{\tscope_{g1}'}, \tscope_{l1}) \sqcup (\overline{\tscope_{g2}'}, \tscope_{l2})}
              {(\mem'_{g2}, \mem_{l2})}$.

          We previously proved that $\mem_{f1}\lequiv{\tscope_{f1} \sqcup \tscope_{f2}} \mem_{f2}$,
          now we can prove this goal directly using Lemma \ref{lemma:low_equivalence_update}.
\end{enumerate}

\casesItem{Case extern}: Here $stmt$ is $f(e_{1} \cons e_{n})$. 
From call typing rule we know (Note that here we explicitly write the global and local state type):

\begin{enumerate}
    \item $(\tscope_{g1}, \tscope_{l1}) \vdash  e_{i} : \type_{i1} $
    \item $(\text{Cont}_{\mathrm{E}},(x_{1},d_{1}) \cons (x_{n}, d_{n})) = \tcv(f) $
    \item $\tscope_{f1} =\{x_{i} \mapsto \type_{i1}\}$
    \item $\forall (\tscope_i, \phi,\tscope_t) \in \text{Cont}_{\mathrm{E}}. \ (\tscope_{g1}, \tscope_{l1}) \sqsubseteq \tscope_i$
    \item $\Gamma_1' =
              \{ \tscope_1' \concat \raisel(\tscope_t, pc) \ | \ (\tscope_i,\phi,\tscope_t) \in \text{Cont}_{\mathrm{E}} \band \refine{(\tscope_{g1}, \tscope_{f1})}{\phi} = \tscope_1' \neq \bullet  \}$
    \item $\Gamma_1'' = \{(\tscope'_{g}, \tscope_{l1}) [e_{i} \mapsto \tscope'_{f}(x_{i}) \mid \isOut(d_{i})]
              \mid (\tscope'_{g}, \tscope'_{f}) \in \Gamma_1'\}$
    \item $(\tscope_{g2}, \tscope_{l2}) \vdash  e_{i} : \type_{i2} $
    \item $\tscope_{f2} =\{x_{i} \mapsto \type_{i2}\}$
    \item $\forall (\tscope_i, \phi,\tscope_t) \in \text{Cont}_{\mathrm{E}}. \ (\tscope_{g2}, \tscope_{l2}) \sqsubseteq \tscope_i$
    \item $\Gamma_2' =
              \{ \tscope_2' \concat \raisel(\tscope_t, pc) \ | \ (\tscope_i,\phi,\tscope_t) \in \text{Cont}_{\mathrm{E}} \band \refine{(\tscope_{g2}, \tscope_{f2})}{\phi} = \tscope_2' \neq \bullet  \}$
    \item $\Gamma_2'' = \{(\tscope'_{g}, \tscope_{l2}) [e_{i} \mapsto \tscope'_{f}(x_{i}) \mid \isOut(d_{i})]
              \mid (\tscope'_{g}, \tscope'_{f}) \in \Gamma_2'\}$
\end{enumerate}

And from extern reduction rule we know:
\begin{enumerate}
    \item $(sem_{f},(x_{1},d_{1}) \cons (x_{n}, d_{n})) = \stmtctxvn{1}(f) $
    \item $\mem_{f1} =\{x_{i} \mapsto \expsem{(\mem_{g1}, \mem_{l1})}{e_{i}}\}$
    \item $(\mem'_{g1}, \mem'_{f1}) = sem_{f}(\mem_{g1}, \mem_{f1})$
    \item $ \mem_1'' = (\mem'_{g1}, \mem_{l1})[e_{i} \mapsto \mem'_{f1}(x_{i}) \mid \isOut(d_{i})]$
    \item $(sem_{f},(x_{1},d_{1}) \cons (x_{n}, d_{n})) = \stmtctxvn{2}(f) $
    \item $\mem_{f2} =\{x_{i} \mapsto \expsem{(\mem_{g2}, \mem_{l2})}{e_{i}}\}$
    \item $(\mem'_{g2}, \mem'_{f2}) = sem_{f}(\mem_{g2}, \mem_{f2})$
    \item $ \mem_2'' = (\mem'_{g2}, \mem_{l2})[e_{i} \mapsto \mem'_{f2}(x_{i}) \mid \isOut(d_{i})]$
\end{enumerate}

Let $\mem_1= (\mem_{g1}, \mem_{l1})$, $\mem_2= (\mem_{g2}, \mem_{l2})$, let  $\tscope_1 = (\tscope_{g1}, \tscope_{l1})$, and  $\tscope_1 = (\tscope_{g2}, \tscope_{l2})$.
We initially know that : $\tscope_1 \vdash \mem_1$, $\tscope_2 \vdash \mem_2$, and $\mem_1 \lequiv{\tscope_1\sqcup\tscope_2} \mem_2$.

We need to prove there are two state types $\tscope_1''\in\Gamma_1''$ and $\tscope_2''\in\Gamma_2''$
such they type the final states $\tscope_1''\vdash\mem_1''$ and $\tscope_2''\vdash\mem_2''$ and indeed $\mem_1'' \lequiv{\tscope_1''\sqcup\tscope_2''} \mem_2''$ holds. 

First we need to prove that the resulted copy-in map is also well-typed i.e. $\tscope_{f1}\vdash\mem_{f1}$ and $\tscope_{f2}\vdash\mem_{f2}$: same proof as the function call case. 

Now we want to show that $(\mem_{g1},\mem_{f1})\lequiv{(\tscope_{g1},\tscope_{f1})\sqcup(\tscope_{g2},\tscope_{f2})}(\mem_{g2},\mem_{f2})$: same proof as the function call case. 

We know from $\text{extWT } C \ X_1$ and $\text{extWT } C \ X_2$ relation in the well-typedness $\stmtctxvn{1}\vdash\tcv$ and $\stmtctxvn{2}\vdash\tcv$ that indeed there exists an input state type, condition and output state type i.e. $(\tscope_i,\phi,\tscope_t)$ in the contract of the extern such that the condition satisfies the input state $\phi(\mem_{g1},\mem_{f1})$.

From \textsc{Soundness of Abstraction} proof, we know that $ (\tscope_{g1},\tscope_{f1}) \vdash (\mem_{g1},\mem_{f1})$.
Using 4 of extern's typing rule we know that $(\tscope_{g1},\tscope_{f1})$ is no more restrictive than $\tscope_i$, i.e. $(\tscope_{g1},\tscope_{f1}) \sqsubseteq \tscope_i$ holds. Using the same strategy and 9 of extern's typing rule, we can also prove $(\tscope_{g2},\tscope_{f2}) \sqsubseteq \tscope_i$. We can also deduct that $(\mem_{g1},\mem_{f1}) \lequiv{\tscope_i} (\mem_{g2},\mem_{f2})$ using Lemma \ref{lemma:state_equivalence_preservation}.

Let the variables $\{x_1 \cons x_n\}$ be the ones used in the condition $\phi$. Using the definition of extern's well-typedness again, we know that the least upper bound of typing label of $\{x_1 \cons x_n\}$ in $\tscope_i$ is no more restrictive that the lower bound of the output state type $\tscope_t$. The entails that since $\phi$ holds on $\mem_1$ (i.e. $\phi(\mem_1)$)then it indeed holds for $\mem_2$ (i.e. $\phi(\mem_2)$).

Now we split the proof into two cases:

\textbf{A.} For the changed variables by extern in the state: given 3 and 7 of the extern reduction rule, and using $\phi(\mem_1)$
we can use the definition of extern's well-typedness again to infer that the variables that are changed by the semantics are a subset of the domain of $\tscope_t$ and low equivalent with respect to $\tscope_t$: 
\begin{align*}
	((\mem'_{g1}, \mem'_{f1}) \setminus (\mem_{g1}, \mem_{f1}))
	\lequiv{\tscope_t}
	((\mem'_{g2}, \mem'_{f2}) \setminus(\mem_{g2}, \mem_{f2}))
\end{align*}

\textbf{B.} Now for the unchanged variables by extern in the state: it is easy to see that the refined state $\refine{(\tscope_{g1},\tscope_{f1})}{\phi}$ is more restrictive than $(\tscope_{g1},\tscope_{f1})$, thus we can say $(\tscope_{g1},\tscope_{f1}) \sqsubseteq \refine{(\tscope_{g1},\tscope_{f1})}{\phi} $. Now we can use Lemma \ref{lemma:state_equivalence_preservation} to show:
\begin{align*}
	(\mem_{g1},\mem_{f1}) \lequiv{\refine{(\tscope_{g1},\tscope_{f1})}{\phi}\sqcup\refine{(\tscope_{g1},\tscope_{f1})}{\phi}}(\mem_{g2},\mem_{f2})
\end{align*}
This property also holds on the variable names $x$ that are unchanged by the behavior of the extern, i.e.,:
\begin{align*}
	(\mem'_{g1},\mem'_{f1}) \lequiv{\refine{(\tscope_{g1},\tscope_{f1})}{\phi}\sqcup\refine{(\tscope_{g1},\tscope_{f1})}{\phi}}(\mem'_{g2},\mem'_{f2})
\end{align*}

We can rename $\refine{(\tscope_{g1},\tscope_{f1})}{\phi} \concat\tscope_t$ to $\overline{\tscope_{3}}$ and $\refine{(\tscope_{g2},\tscope_{f2})}{\phi} \concat\tscope_t$ to $\overline{\tscope_{4}}$, and we can easily infer from A and B that :
$(\mem'_{g1},\mem'_{f1}) \lequiv{\overline{\tscope_{3}}\sqcup\overline{\tscope_{4}}}(\mem'_{g2},\mem'_{f2})$.

The rest of the proof is the same as the function call case.

\casesItem{Case table application}: Here $stmt$ is $\texttt{apply} \ tbl $. 
From table rule we know:

\begin{enumerate}
    \item $(\overline{e},\text{Cont}_{\mathrm{tbl}})= \tcv(tbl)$
    \item $\tscope_1 \vdash {{e}_i} : {\type}_{i1}$
    \item $\lbl_1 = \bigsqcup_{i} \getlabel(\type_{i1})$
    \item $pc'_1 = pc \sqcup \lbl_1 $
    \item  $\forall (\phi_j, (a_j, \overline{\tau}_j)) \in \text{Cont}_{\mathrm{tbl}}. \
              (\tscope_{g_j},\tscope_{l_j})= \refine{\tscope_1}{\phi_j} \band \\
              (\stmt_j,(x_{j_1},\text{none}) \cons (x_{j_n}, \text{none})) = \tcv(a_j) \band \\
              \tscope_{a_j} =\{x_{j_i} \mapsto \type_{j_i}\} \band
              \tcv, pc'_1, (\tscope_{g_j},\tscope_{a_j}) \vdash \stmt_j : \Gamma_j$
    \item   $\Gamma'_1 = \cup_j \{ (\tscope_{g_j}' , \tscope_{l_j}) | (\tscope_{g_j}',\tscope_{a_j}') \in \Gamma_j \}$
    \item   $\Gamma''_1 =  {\begin{cases}
                  \join(\Gamma'_1) \hspace{15pt} \text{if } \lbl_1 = {\high} \\
                  \Gamma'_1   \hspace{38pt}  \text{otherwise}
              \end{cases}}$
    \item $\tscope_2 \vdash {{e}_i} : {\type}_{i2}$
    \item $\lbl_2 = \bigsqcup_{i} \getlabel(\type_{i2})$
    \item $pc'_2 = pc \sqcup \lbl_2 $
    \item  $\forall (\phi_k, (a_k, \overline{\tau}_k)) \in \text{Cont}_{\mathrm{tbl}}. \
              (\tscope_{g_k},\tscope_{l_k})= \refine{\tscope_2}{\phi_k} \band \\
              (\stmt_k,(x_{k_1},\text{none}) \cons (x_{k_n}, \text{none})) = \tcv(a_k) \band \\
              \tscope_{a_k} =\{x_{k_i} \mapsto \type_{k_i}\} \band
              \tcv, pc'_2, (\tscope_{g_k},\tscope_{a_k}) \vdash \stmt_k : \Gamma_k$
    \item   $\Gamma'_2 = \cup_k \{ (\tscope_{g_k}' , \tscope_{l_k}) | (\tscope_{g_k}',\tscope_{a_k}') \in \Gamma_k \}$
    \item   $\Gamma''_2 =  {\begin{cases}
                  \join(\Gamma'_2) \hspace{15pt} \text{if } \lbl_2 = {\high} \\
                  \Gamma'_2   \hspace{38pt}  \text{otherwise}
              \end{cases}}$
\end{enumerate}

And from table reduction rule we know:
\begin{enumerate}
    \item $(\overline{e}, sem_{tbl1}) = \stmtctxvn{1}(tbl)$
    \item $sem_{tbl1}(\expsem{(\mem_{g1}, \mem_{l1})}{e_{1}} \cons \expsem{(\mem_{g1}, \mem_{l1})}{e_{n}}) = (a_1, \overline{\val})$
    \item $(\stmt_1,(x_{1},\mathrm{none}) \cons (x_{n}, \mathrm{none})) = \stmtctxvn{1}(a_1)$
    \item $\mem_{a1} =\{x_{i} \mapsto {\val_i}\}$
    \item $\stmtctxvn{1}:\sem{(\mem_{g1}, \mem_{a1})}{\stmt_1}{(\mem_{g1}', \mem_{a1}')}$
    \item $\mem_{final1} = (\mem_{g1}',\mem_{l1})$
    \item $(\overline{e}', sem_{tbl2}) = \stmtctxvn{2}(tbl)$
    \item $sem_{tbl2}(\expsem{(\mem_{g2}, \mem_{l2})}{e'_{1}} \cons \expsem{(\mem_{g2}, \mem_{l2})}{e'_{n}}) = (a_2, \overline{\val}')$
    \item $(\stmt_2,(x'_{2},\mathrm{none}) \cons (x'_{n}, \mathrm{none})) = \stmtctxvn{2}(a_2)$
    \item $\mem_{a2} =\{x'_{i} \mapsto {\val'_i}\}$
    \item $\stmtctxvn{2}:\sem{(\mem_{g2}, \mem_{a2})}{\stmt_2}{(\mem_{g2}', \mem_{a2}')}$
    \item $\mem_{final2} = (\mem_{g2}',\mem_{l2})$
\end{enumerate}

In addition to that, we get induction hypothesis for $\stmt_1$ \textsc{IH1} (similarly for $\stmt_2$ \textsc{IH2}):
\begin{align*}
	\forall \ \tc \ pc \  \mem_a \ \mem_b \ \mem_a' \ & \mem_b' \ \tscope_a \ \tscope_b \ \Gamma_a \ \Gamma_b \ \stmtctx_a \ \stmtctx_b . \\
	\ \tc, pc,  \tscope_a \vdash \stmt : \Gamma_a &\band \ \tc, pc, \tscope_b \vdash \stmt : \Gamma_b
	\implies \\
	\tc \vdash \stmtctx_a &\band \tc \vdash \stmtctx_b \band \stmtctx_a \lequiv{\tc} \stmtctx_b
	\band \\
	\tscope_a \vdash \mem_a
	&\band
	\tscope_b \vdash \mem_b
	\band
	\mem_a \lequiv{\tscope_a\sqcup\tscope_b} \mem_b
	\band \\
	{\stmtctx_a}: \sem{\mem_a}{\stmt_1}{\mem_a'}
	&\band
	{\stmtctx_b}: \sem{\mem_b}{\stmt_1}{\mem_b'} \\
	&\implies \\
	\big(
	\exists \tscope_a' \in \Gamma_a \band \tscope_b' \in \Gamma_b. &\
	\tscope_a' \vdash \mem_a' \band \tscope_b' \vdash \mem_b' \band \mem_a' \lequiv{\tscope_a'\sqcup\tscope_b'} \mem_b' \big)
\end{align*}

Let $\mem_1= (\mem_{g1}, \mem_{l1})$,
$\mem_2= (\mem_{g2}, \mem_{l2})$, let
$\tscope_1 = (\tscope_{g1}, \tscope_{l1})$, and
$\tscope_2 = (\tscope_{g2}, \tscope_{l2})$.
We initially know that :
$\tscope_1 \vdash \mem_1$,
$\tscope_2 \vdash \mem_2$, and
$\mem_1 \lequiv{\tscope_1 \sqcup \tscope_2} \mem_2$.

The final goal is to prove there
are two state types $\tscope_1''\in\Gamma_1''$ and $\tscope_2''\in\Gamma_2''$
such they type the final states $\tscope_1''\vdash(\mem_{g1}', \mem_{l1})$ and $\tscope_2''\vdash(\mem_{g2}', \mem_{l2})$ and indeed $(\mem_{g1}', \mem_{l1}) \lequiv{\tscope_1''\sqcup\tscope_2''} (\mem_{g2}', \mem_{l2})$ holds.

Initially,
we prove expression (table keys) found in 1 in reduction rule, with 1 and 7 of typing rule are the same. We use $\stmtctxvn{1}\lequiv{\tcv}\stmtctxvn{2}$ to show that 1 and 7 in the typing rule have the same expression (i.e. $\overline{e} = \overline{e}'$). Then, using $\text{tblWT } C \ X_1$ and $\text{tblWT } C \ X_2$, we confirm that the expression in rule 1 of the reduction rule matches those in rules 1 and 7 of the typing rule. Therefore, all relevant expressions are equivalent.

First, we conduct a case analysis on $\lbl_1$ and $\lbl_2$ being equivalent:
\begin{enumerate}[label=\textbf{case}]
    \item $\lbl_1 = \lbl_2 = {\low}$:
          Given that $\lbl_1$ and $\lbl_2$ are {\low} in 3 and 9 of the typing rule, respectively, we can conclude that the evaluation of the key expressions $\overline{e}$ in states $\mem_1$ and $\mem_2$ are identical. This follows directly from Lemma \ref{lemma:expression_evaluation_of_consistent_states}, which establishes that $\mem_1(e_i)=\mem_2(e_i)$ for all $e_i$. 

          By the definition of $\stmtctxvn{1}\lequiv{\tcv}\stmtctxvn{2}$, we know that for any memory states $\mem_1$ and $\mem_2$, if $\mem_1(e_i)=\mem_2(e_i)$ for all table's keys $e_i$, then $\mem_1(\phi) \Leftrightarrow \mem_2(\phi)$. This implies that the condition for both tables to match is identical.
          Applying this definition again, we deduct that the actions $a_1$ and $a_2$ in 2 and 8 of the reduction rule are identical i.e. $a_1 = a_2 = a$, therefore their corresponding action bodies and signatures must also be the same in $\stmtctxvn{1}$ and $\stmtctxvn{2}$, thus $\stmt_1=\stmt_2=s$ and $\overline{x}=\overline{x}'$.
          Applying this definition $\stmtctxvn{1}\lequiv{\tcv}\stmtctxvn{2}$, yet again, we can also deduct that the action's values $\overline{\val}$ and $\overline{\val}'$ are low equivalent wrt. $\overline{\type}$ i.e. $\overline{\val}\lequiv{\overline{\type}}\overline{\val}'$.

          Now we can rewrite the table reduction rule to:
          \begin{enumerate}
              \item $(\overline{e}, sem_{tbl1}) = \stmtctxvn{1}(tbl)$
              \item $sem_{tbl1}(\expsem{(\mem_{g1}, \mem_{l1})}{e_{1}} \cons \expsem{(\mem_{g1}, \mem_{l1})}{e_{n}}) = (a, \overline{\val})$
              \item $(\stmt,(x_{1},\mathrm{none}) \cons (x_{n}, \mathrm{none})) = \stmtctxvn{1}(a)$
              \item $\mem_{a1} =\{x_{i} \mapsto {\val_i}\}$
              \item $\stmtctxvn{1}:\sem{(\mem_{g1}, \mem_{a1})}{\stmt}{(\mem_{g1}', \mem_{a1}')}$
              \item $\mem_{final1} = (\mem_{g1}',\mem_{l1})$
              \item $(\overline{e}, sem_{tbl2}) = \stmtctxvn{2}(tbl)$
              \item $sem_{tbl2}(\expsem{(\mem_{g2}, \mem_{l2})}{e_{1}} \cons \expsem{(\mem_{g2}, \mem_{l2})}{e_{n}}) = (a, \overline{\val}')$
              \item $(\stmt,(x_{2},\mathrm{none}) \cons (x_{n}, \mathrm{none})) = \stmtctxvn{2}(a)$
              \item $\mem_{a2} =\{x_{i} \mapsto {\val'_i}\}$
              \item $\stmtctxvn{2}:\sem{(\mem_{g2}, \mem_{a2})}{\stmt}{(\mem_{g2}', \mem_{a2}')}$
              \item $\mem_{final2} = (\mem_{g2}',\mem_{l2})$
          \end{enumerate}  

          Given $\tcv\vdash\stmtctxvn{1}$
          that indeed exists condition $\phi_j$ in $\text{Cont}_{\mathrm{tbl}}$ that satisfies the table input state $(\mem_{g1}, \mem_{l1})$, and also exists a list of types $\overline{\type}$ such that it can type that values of the table's semantics (in 2 of the typing rule) i.e. ${\val_i:\type_i}$. Similarly, from $\tcv\vdash\stmtctxvn{2}$, we know that exists $\phi_k$ that satisfies the table input state $(\mem_{g2}, \mem_{l2})$, and exists $\overline{\type}'$ such that ${\val_i':\type_i'}$. 

          We can instantiate 5 from
          the table rule with $(\phi_j, (a,\overline{\type}))$, and instantiate 11 with $(\phi_k, (a,\overline{\type}'))$.

          We need to prove $(\mem_{g1},\mem_{a1}) \lequiv{(\tscope_{g_j},\tscope_{a_j})\sqcup(\tscope_{g_k},\tscope_{a_k})} (\mem_{g2},\mem_{a2})$ using the following sub-goals:

          \begin{enumerate}[label=\textbf{Goal}]
              \item \textbf{1.} prove $\mem_{g1}\lequiv{\tscope_{g_j}\sqcup\tscope_{g_k}} \mem_{g2}$ and $\mem_{l1}\lequiv{\tscope_{l_j}\sqcup\tscope_{l_k}} \mem_{l2}$: \\
                    It is easy to see that
                    $\tscope_1 \sqsubseteq \refine{\tscope_1}{\phi_j}$ and $\tscope_2 \sqsubseteq \refine{\tscope_2}{\phi_k}$ trivially hold, and can be rewritten to $\tscope_1 \sqsubseteq (\tscope_{g_j},\tscope_{l_j})$ and $\tscope_2 \sqsubseteq (\tscope_{g_k},\tscope_{l_k})$. Since initially we know that $(\mem_{g1}, \mem_{l1})\lequiv{\tscope_1\sqcup\tscope_2} (\mem_{g2}, \mem_{l2})$, therefore using Lemma \ref{lemma:state_equivalence_preservation}
                    $(\mem_{g1}, \mem_{l1})\lequiv{(\tscope_{g_j},\tscope_{l_j})\sqcup(\tscope_{g_k},\tscope_{l_k})} (\mem_{g2}, \mem_{l2})$.
                    This entails using Lemma \ref{lemma:low_equivalence_distribution} that
                    $\mem_{g1}\lequiv{\tscope_{g_j}\sqcup\tscope_{g_k}} \mem_{g2}$ and also $\mem_{l1}\lequiv{\tscope_{l_j}\sqcup\tscope_{l_k}} \mem_{l2}$ hold. 

                    Additionally, note that Hyp \ref{hyp:interval_typedness_externs_and_tables_refinement} states also that $(\tscope_{g_j},\tscope_{l_j})$ can still type the state $(\mem_{g1}, \mem_{l1})$, i.e.
                    $(\tscope_{g_j},\tscope_{l_j}) \vdash (\mem_{g1}, \mem_{l1})$.
                    Similarly,
                    $(\tscope_{g_k},\tscope_{l_k}) \vdash (\mem_{g2}, \mem_{l2})$.

              \item \textbf{2.} prove $\mem_{a1} \lequiv{\tscope_{a_j}\sqcup\tscope_{a_k}} \mem_{a2}$: %
                    From $\stmtctxvn{1}\lequiv{\tcv}\stmtctxvn{2}$
                    we know it holds $\overline{\val} \lequiv{\overline{\type}} \overline{\val}'$, note that these are the table's semantic output (i.e. will be action's arguments).
                    Thus, whenever $\tau_i$ is {\low}, then the values of arguments are equivalent $\val_i = \val'_i$.
                    This entails that constructing states $\mem_{a1}=\{x_{i} \mapsto {\val_i}\}$ and $\mem_{a2}=\{x_{i} \mapsto {\val'_i}\}$ must be low equivalent with respect to $\tscope_{a_j} =\{x_{i} \mapsto \type_{i}\}$, i.e. $\mem_{a1} \lequiv{\tscope_{a_j}} \mem_{a2}$.
                    Similarly, from $\stmtctxvn{1}\lequiv{\tcv}\stmtctxvn{2}$ and whenever $\tau'_i$ is {\low}, we can also deduct that $\tscope_{a_k} =\{x_{i} \mapsto \type_{i}'\}$, i.e. $\mem_{a1} \lequiv{\tscope_{a_k}} \mem_{a2}$. Therefore, this entails (using Lemma \ref{lemma:low_equivalence_distribution}) that $\mem_{a1} \lequiv{\tscope_{a_j}\sqcup\tscope_{a_k}} \mem_{a2}$.

          \end{enumerate}

          From the last two sub-goals,
          we can use Lemma \ref{lemma:low_equivalence_distribution} to deduct $(\mem_{g1},\mem_{a1}) \lequiv{(\tscope_{g_j},\tscope_{a_j})\sqcup(\tscope_{g_k},\tscope_{a_k})} (\mem_{g2},\mem_{a2})$.

          Now we can use \textsc{IH} and instantiate it with
          $\big(\tcv, {pc \sqcup \lbl}, (\mem_{g1},\mem_{a1}) , (\mem_{g2},\mem_{a2}), (\mem_{g1}',\mem_{a1}'),$ $(\mem_{g2}',\mem_{a2}'), (\tscope_{g_j},\tscope_{a_j}), (\tscope_{g_k},\tscope_{a_k}), \Gamma_j,\Gamma_k, \stmtctxvn{1}, \stmtctxvn{2} \big)$
          to deduct that indeed exist $\overline{\tscope_{1}}\in\Gamma_j$ and $\overline{\tscope_{2}}\in\Gamma_k$ such that $\overline{\tscope_{1}}\vdash(\mem_{g1}', \mem_{a1}')$ and
          $\overline{\tscope_{2}}\vdash(\mem_{g2}', \mem_{a2}')$ and indeed $(\mem_{g1}', \mem_{a1}') \lequiv{\overline{\tscope_{1}}\sqcup\overline{\tscope_{2}}} (\mem_{g2}', \mem_{a2}')$.
          In the following, let $\overline{\tscope_{1}} = (\overline{\tscope_{1_g}}, \overline{\tscope_{1_l}})$ and $\overline{\tscope_{2}} = (\overline{\tscope_{2_g}}, \overline{\tscope_{2_l}})$.
          We can rewrite the \textsc{IH} results as following: exist
          $(\overline{\tscope_{1_g}}, \overline{\tscope_{1_l}})\in\Gamma_j$ and
          $(\overline{\tscope_{2_g}}, \overline{\tscope_{2_l}})\in\Gamma_k$ such that
          $(\overline{\tscope_{1_g}}, \overline{\tscope_{1_l}})\vdash(\mem_{g1}', \mem_{a1}')$ and
          $(\overline{\tscope_{2_g}}, \overline{\tscope_{2_l}})\vdash(\mem_{g2}', \mem_{a2}')$
          and indeed
          $(\mem_{g1}', \mem_{a1}') \lequiv{(\overline{\tscope_{1_g}}, \overline{\tscope_{1_l}})\sqcup(\overline{\tscope_{2_g}}, \overline{\tscope_{2_l}})} (\mem_{g2}', \mem_{a2}')$.

          Clearly, from 6 in the typing rule, and we know that $\Gamma_1'$ is the union of all the changed global state types by the action's body, with the refined starting local state type $(\tscope_{g_j},\tscope_{l_j})= \refine{\tscope_1}{\phi_j}$. Thus, we know that indeed $(\overline{\tscope_{1_g}},\tscope_{l_j}) \in \Gamma_1'$.
          Similarly, from 12 in the typing rule, we know that $(\overline{\tscope_{2_g}},\tscope_{l_k}) \in \Gamma_2'$. 

          We choose $(\overline{\tscope_{1_g}},\tscope_{l_j})$ and $(\overline{\tscope_{2_g}},\tscope_{l_k})$ to finish the proof of the goal in this subcase.

          Since $\lbl_1 = \lbl_2$ and is {\low} in this sub-case, then trivially $\Gamma_1'' = \Gamma_1'$ in 7 of the typing rule, and $\Gamma_2'' = \Gamma_2'$ in 13 of the typing rule. Since we proved that $\tscope_{l_j} \vdash \mem_{l1}$ and $\tscope_{l_k} \vdash \mem_{l2}$, therefore
          $(\overline{\tscope_{1_g}},\tscope_{l_j})\vdash(\mem_{g1}', \mem_{l1})$
          and
          $(\overline{\tscope_{2_g}},\tscope_{l_k})\vdash(\mem_{g2}', \mem_{l2})$ hold.
          Additionally, using Lemma \ref{lemma:low_equivalence_distribution}
          $(\mem_{g1}', \mem_{l1}) \lequiv{(\overline{\tscope_{1_g}},\tscope_{l_j})\sqcup(\overline{\tscope_{2_g}},\tscope_{l_k})} (\mem_{g2}', \mem_{l2})$.

    \item $\lbl_1 \neq \lbl_2$:
          This proof is similar to the conditional case.
          We initiate the proof by fixing $\lbl_2$ to be {\high} while $\lbl_1$ can be either {\high} or {\low}, thus the evaluation of $\overline{e}$ in the states $(\mem_{g1}, \mem_{l1})$ and $(\mem_{g2}, \mem_{l2})$ differs. Consequently, we (possibly) end up with two different actions and their corresponding arguments $(a_1,\overline{\val})$ and $(a_2,\overline{\val}')$.

          From $\tcv\vdash\stmtctxvn{1}$
          we know that indeed exists condition $\phi_j$ in $\text{Cont}_{\mathrm{tbl}}$ that satisfies the table input state $(\mem_{g1}, \mem_{l1})$, and also exists a list of types $\overline{\type}$ such that it can type that values of the table's semantics (in 2 of the typing rule) i.e. ${\val_i:\type_i}$.
          Similarly, from $\tcv\vdash\stmtctxvn{2}$, we know that exists $\phi_k$ that satisfies the table input state $(\mem_{g2}, \mem_{l2})$, and exists $\overline{\type}'$ such that ${\val_i':\type_i'}$.

          We can instantiate 5 (we refer to these as the first configuration) from
          the table rule with $(\phi_j, (a_1,\overline{\type}))$, and instantiate 11 (we refer to these as the second configuration) with $(\phi_k, (a_2,\overline{\type}'))$.
          Since the actions are different, then we let $\stmt_1$ be the body of action $a_1$, and $\stmt_2$ be the body of action $a_2$ (we refer to these as the first configuration).

          Using the same steps in the previous sub-case, we know that $\mem_{g1}\lequiv{\tscope_{g_j}\sqcup\tscope_{g_k}} \mem_{g2}$ and $\mem_{l1}\lequiv{\tscope_{l_j}\sqcup\tscope_{l_k}} \mem_{l2}$.

          The final goal is to prove there
          are two state types $\tscope_1''\in\Gamma_1''$ and $\tscope_2''\in\Gamma_2''$
          such they type the final states $\tscope_1''\vdash(\mem_{g1}', \mem_{l1})$ and $\tscope_2''\vdash(\mem_{g2}', \mem_{l2})$ and indeed $(\mem_{g1}', \mem_{l1}) \lequiv{\tscope_1''\sqcup\tscope_2''} (\mem_{g2}', \mem_{l2})$ holds.

          From the \textsc{Soundness of Abstraction},
          we know that for the second configuration indeed exists $\tscope_2'\in\Gamma_k$ such that it types the action's resulted state $(\mem_{g2}', \mem_{a2}')$ (i.e. $\tscope_2'\vdash(\mem_{g2}', \mem_{a2}')$).
          In the following, let $(\tscope_{g2}',\tscope_{l2}') = \tscope_2'$, thus indeed
          $(\tscope_{g2}',\tscope_{l2}')\vdash(\mem_{g2}', \mem_{a2}')$.
          Given
          that $\Gamma_2'$ in 12 of the typing rule is a union of all $\Gamma_i$; thus indeed it includes $\Gamma_k$ that has the resulted global state type $\tscope_{g2}'$, and the refined caller's local state type $\tscope_{l_k}$ while removing the callee's resulted local state type $\tscope_{a_k}$.
          In short, we know that indeed for the second configuration will find $(\tscope_{g2}',\tscope_{l_k}) \in \Gamma_2'$.
          Note that in the \textsc{Soundness of Abstraction} previously we proved that $\tscope_{l_k} \vdash \mem_{l2}$, therefore $(\tscope_{g2}',\tscope_{l_k}) \vdash (\mem_{g2}', \mem_{l2})$.
          Since $\lbl_2$ is {\high}, then
          $\Gamma_2'' = \text{join}(\Gamma_2')$, so we can deduce (by Lemma \ref{lemma:join_does_not_modify_the_intervals}) the existence of $\overline{\tscope_2'} \in \text{join}(\Gamma_2')$ such that it is more restrictive than $\tscope_2'$, denoted as $\tscope_2' \sqsubseteq \overline{\tscope_2'}$. Using the same lemma, we conclude that $\overline{\tscope_2'} \vdash (\mem_{g2}', \mem_{l2})$. Now on, we choose $\overline{\tscope_2'}$ to be used in the proof and resolve the second conjunction of the goal.

          From the \textsc{Soundness of Abstraction},
          we know that for the first configuration indeed exists $\tscope_1'\in\Gamma_j$ such that it types action's resulted state $(\mem_{g1}', \mem_{a1})$ (i.e. $\tscope_1'\vdash(\mem_{g1}', \mem_{a1})$).
          In the first configuration, the final state type set $\Gamma_1''$ can be either a union (if $\lbl_1 = {\low}$) or a join (if $\lbl_1 = {\high}$) of all final state type sets and including $\Gamma_1'$.
          In either case (union or join), we can establish the existence of
          $\overline {\tscope_1'} \in \Gamma_1''$ such that $\tscope_1' \sqsubseteq \overline{\tscope_1'}$ and indeed $\overline{\tscope_1'} \vdash (\mem_{g1}', \mem_{l1})$. Note that if $\Gamma_1''$ resulted from a join, we follow the same steps of the (second configuration) in the previous step; otherwise, if it resulted from a union, it is trivially true. Thus, the first conjunction of the final goal is proved.

          The final goal left to prove is $(\mem_{g1}', \mem_{l1}) \lequiv{\overline{\tscope_1'}\sqcup\overline{\tscope_2'}} (\mem_{g2}', \mem_{l2})$ holds. In the following, let $(\overline{\tscope_{g1}'},\overline{\tscope_{l1}'}) = \overline{\tscope_1'}$ and $(\overline{\tscope_{g2}'},\overline{\tscope_{l2}'}) = \overline{\tscope_2'}$. 

          Next, we proceed to implement cases based on whether an $\lval$'s type label is {\high} or {\low}. 

          \begin{enumerate}[label=\textbf{case}]
              \item $(\overline{\tscope_1'}\sqcup\overline{\tscope_2'}(\lval) = \type) \band \getlabel(\type) =  {\high}$: holds trivially.

              \item $(\overline{\tscope_1'}\sqcup\overline{\tscope_2'}(\lval) = \type) \band \getlabel(\type) =  {\low} $:
                    this case entails that each state type individually holds $\overline{\tscope_1'}(\lval) = \type_1' \band \getlabel(\type_1') =  {\low}$ and also $\overline{\tscope_2'}(\lval) = \type_2'' \band \getlabel(\type_2'') =  {\low}$.

                    Given that the $\lval$'s type is {\low} in $\overline{\tscope_2'}$, and
                    considering $\overline{\tscope_2'} \in \text{join}(\Gamma_2')$, it follows that the $\lval$ is also {\low} in any state type within $(\Gamma_2')$.
                    Consequently, the $\lval$ is {\low} in $\Gamma_k$, thus {\low} in $(\tscope_{g2}',\tscope_{l_k})$.

                    In the second configuration, we type the action's body $\stmt_2$ with a
                    {\high} context, where $\stmt_2$ reduces to $(\mem_{g2}',\mem_{a2}')$.
                    And since we previously showed that $\lval$ is {\low} in $(\tscope_{g2}',\tscope_{l_k})$ such that $(\tscope_{g2}',\tscope_{l_k}) \vdash (\mem_{g2}', \mem_{l2})$, then $\lval$ is in $\mem_{g2}'$ or $\mem_{l2}$, however not in $\mem_{a2}'$.

                    Since $\lval$'s type is {\low} in $\tscope_{g2}'$,
                    we can use Lemma \ref{lemma:branch_on_high_state_preservation} to infer that the global initial and final states remain unchanged for {\low} lvalues, which means $\mem_{g2}(\lval) = \mem_{g2}'(\lval)$.
                    Then we can use Lemma \ref{lemma:high_program_final_types_subset_of_initial} to infer that $\refine{\tscope_2}{\phi_k} \sqsubseteq (\tscope_{g2}',\tscope_{l2}')$, this entails that the $\lval$'s type label is indeed {\low} in the global of refined state $\tscope_{g_k}$.
                    It is easy to see that $\tscope_2\sqsubseteq\refine{\tscope_2}{\phi_k}$, thus  $(\tscope_{g2}, \tscope_{l2}) \sqsubseteq (\tscope_{g_k},\tscope_{l_k})$, therefore
                    $\lval$'s type is also {\low} in the global initial state type $\tscope_{g2}$, i.e. $\tscope_{g2}(\lval) = \type_2' \band \getlabel(\type_2') = {\low}$.

                    For the first configuration, in this sub-case, we have
                    $\overline{\tscope_1'}(\lval)= \type_1' \band \getlabel(\type_1') = {\low}$, and $\overline{\tscope_1'} \in \Gamma_1''$.
                    We previously showed $\tscope_1' \sqsubseteq \overline{\tscope_1'}$.
                    Since $\lval$'s typing label is {\low} in $\overline{\tscope_1'}$ and we know that the state types in $\overline{\tscope_1'}$ are more restrictive than the state types in $\tscope_1'$, we can conclude that $\tscope_1' \in \Gamma_j$ also types $\lval$ as {\low}.

                    Now, we will prove that $\mem_{g1}'\lequiv{\overline{\tscope_{g1}'}\sqcup\overline{\tscope_{g2}'}} \mem_{g2}'$ by conducting cases on $\lbl_1$:

                    \begin{enumerate}[label=\textbf{case}]
                        \item If $\lbl_1$ is {\high}: \\
                              We can replicate the same exact steps done for the second configuration to deduct $\mem_{g1}(\lval)=\mem_{g1}'(\lval)$ and $(\tscope_{g1}, \tscope_{l1}) \sqsubseteq (\tscope_{g_j},\tscope_{l_j})$ and the $\lval$'s type is {\low} in the global initial state type $\tscope_{g1}$, i.e. $\tscope_{g1}(\lval) = \type_2' \band \getlabel(\type_2') = {\low}$.

                        \item If $\lbl_1$ is {\low}:\\
                              Previously, we demonstrated that the $\lval$'s typing label is
                              {\low} in all final state types in all $\Gamma_i$ in $\Gamma_2$.
                              Consequently, none of the actions' bodies can modify $\lval$, this is true because all actions in the first and second semantics and the contracts are identical, therefore $\mem_{g1}(\lval)=\mem_{g1}'(\lval)$. Thus, we know that indeed $\stmt_1$ is typed under a high $pc$ in the second configuration, we conclude that the $\lval$ remains unchanged there as well. Consequently, we can now apply Lemma \ref{lemma:branch_on_high_state_lemma} to deduce that the $\lval$'s typing label in the first configuration $\tscope_1' \in \Gamma_j$ are more restrictive than the one we find in the refined state $\refine{\tscope_1}{\phi_j}$, and because $(\tscope_{g_j},\tscope_{l_j})= \refine{\tscope_1}{\phi_j}$ then $(\tscope_{g_j},\tscope_{l_j})(\lval) = \type_1' \band \getlabel(\type_1') = {\low}$.
                              Therefore, $\lval$'s type is also {\low} in the global initial state type $\tscope_{g2}$, i.e. $\tscope_{g2}(\lval) = \type_2' \band \getlabel(\type_2') = {\low}$
                    \end{enumerate}

                    Then,
                    we need to prove $\mem_{l1} \lequiv{\overline{\tscope_{l1}'}\sqcup\overline{\tscope_{l2}'}} \mem_{l2}$. 
                    First,
                    we prove that $\overline{\tscope_{l1}'} = \tscope_{l_j}$ directly from the definition of join as $\join(\tscope_{l_j},\tscope_{l_j}) = \tscope_{l_j}$ and we know $\join(\tscope_{l_j},\tscope_{l_j}) = \overline{\tscope_{l1}'}$ thus $\overline{\tscope_{l1}'} = \tscope_{l_j}$ holds. Similarly, we know that $\overline{\tscope_{l2}'} = \tscope_{l_k}$ holds.
                    Since
                    $\lval$ is {\low} in $\overline{\tscope_{l2}'}$ then it is also {\low} in $\tscope_{l_k}$. And since $\lval$ is {\low} in $\overline{\tscope_{l2}'}$, then it is also {\low} in $\tscope_{l_j}$.
                    From the previous subgoal, we proved $\mem_{l1}\lequiv{\tscope_{l_j}\sqcup\tscope_{l_k}} \mem_{l2}$. This is property hold.

                    Finally, since we proved $\mem_{g1}'\lequiv{\overline{\tscope_{g1}'}\sqcup\overline{\tscope_{g2}'}} \mem_{g2}'$ and $\mem_{l1} \lequiv{\overline{\tscope_{l1}'}\sqcup\overline{\tscope_{l2}'}} \mem_{l2}$ now we can use Lemma \ref{lemma:low_equivalence_distribution} to deduct that $(\mem_{g1}', \mem_{l1}) \lequiv{\overline{\tscope_1'}\sqcup\overline{\tscope_2'}} (\mem_{g2}', \mem_{l2})$ holds.
          \end{enumerate}

\end{enumerate}

\casesItem{Case transition}: similar to the conditional statement proof. 

\end{proof}

\else
\fi

\end{document}